\documentclass{biometrika}

\usepackage{amsmath}
\usepackage{amssymb}
\usepackage{graphicx}
\graphicspath{{./figs/}}
\usepackage{times}
\usepackage{bm}
\usepackage{natbib}

\usepackage[plain,noend]{algorithm2e}
\usepackage{booktabs}

\makeatletter
\renewcommand{\algocf@captiontext}[2]{#1\algocf@typo. \AlCapFnt{}#2} 
\def\@algocf@capt@plain{top}
\renewcommand{\algocf@makecaption}[2]{%
  \addtolength{\hsize}{\algomargin}%
  \sbox\@tempboxa{\algocf@captiontext{#1}{#2}}%
  \ifdim\wd\@tempboxa >\hsize
    \hskip .5\algomargin%
    \parbox[t]{\hsize}{\algocf@captiontext{#1}{#2}}
  \else%
    \global\@minipagefalse%
    \hbox to\hsize{\box\@tempboxa}
  \fi%
  \addtolength{\hsize}{-\algomargin}%
}

\newcommand{\cX}{\mathcal{X}}

\newcommand{\cH}{\mathcal{H}}
\newcommand{\cF}{\mathcal{F}}
\newcommand{\cG}{\mathcal{G}}
\newcommand{\cC}{\mathcal{C}}

\newcommand{\cR}{\mathcal{R}}
\newcommand{\cK}{\mathcal{K}}
\newcommand{\cE}{\mathcal{E}}

\newcommand{\ANOVA}{ANOVA}
\newcommand{\textFDP}{FDP}
\newcommand{\textFDR}{FDR}
\newcommand{\textFWER}{FWER}

\newcommand{\Star}{STAR}
\renewcommand{\star}{STAR}
\newcommand{\adapt}{AdaPT}
\newcommand{\sabha}{SABHA}
\newcommand{\ihw}{IHW}
\newcommand{\bh}{BH}
\newcommand{\DAG}{DAG}
\newcommand{\EM}{EM}
\newcommand{\hFDP}{\widehat{\textnormal{FDP}}}

\newcommand{\htt}{\tau}

\newcommand{\website}{\texttt{https://github.com/lihualei71/STAR}}

\newcommand{\FDP}{\textnormal{FDP}}
\newcommand{\FDR}{\textnormal{FDR}}
\newcommand{\Pow}{\textnormal{Pow}}
\newcommand{\TPR}{\textnormal{TPR}}

\newcommand{\pth}{p_{*}}
\newcommand{\exps}[1]{\exp\lb#1\rb}

\newcommand{\eps}{\epsilon}
\newcommand{\lb}{\left(}
\newcommand{\rb}{\right)}
\newcommand{\td}{\tilde}
\newcommand{\R}{\mathbb{R}}
\newcommand{\E}{\mathbb{E}}
\renewcommand{\P}{\mathbb{P}}

\newcommand{\1}{\mathbf{1}}
\newcommand{\s}{s}
\newcommand{\eqAS}{\overset{a.s.}{=}}

\usepackage{color}

\DeclareMathOperator*{\argmin}{arg\,min}
\DeclareMathOperator*{\argmax}{arg\,max}

\makeatother

\begin{document}

\jname{Biometrika}


\markboth{L. Lei et~al.}{\textFDR ~control under structural constraints}

\title{STAR: A general interactive framework for \textFDR ~control under structural constraints}

\author{Lihua Lei}
\affil{Departments of Statistics, Stanford University \\ 202 Sequoia Hall, Stanford, CA94305, U.S.A. \email{lihualei@stanford.edu}}

\author{Aaditya Ramdas}
\affil{Department of Statistics and Data Science, Carnegie Mellon University \\ 132H Baker Hall, CMU, Pittsburgh, PA15232, U.S.A. \email{aramdas@cmu.edu}}

\author{William Fithian}
\affil{Departments of Statistics, University of California, Berkeley \\ 301 Evans Hall, UC Berkeley, Berkeley, CA94720, U.S.A. \email{wfithian@berkeley.edu}}

\maketitle

\begin{abstract}
  We propose a general framework based on \emph{selectively traversed accumulation rules} (\star) ~for interactive multiple testing with generic structural constraints on the rejection set. It combines accumulation tests from ordered multiple testing with data-carving ideas from post-selection inference, allowing for highly flexible adaptation to generic structural information. Our procedure defines an interactive protocol for gradually pruning a candidate rejection set, beginning with the set of all hypotheses and shrinking with each step. By restricting the information at each step 
  via a technique we call masking, our protocol enables interaction while controlling the false discovery rate (\textFDR) in finite samples for any data-adaptive update rule that the analyst may choose. We suggest update rules for a variety of applications with complex structural constraints, show that \star ~performs well for problems ranging from convex region detection to \textFDR ~control on directed acyclic graphs, and show how to extend it to regression problems where knockoff statistics are available in lieu of $p$-values.
\end{abstract}

\begin{keywords}
interactive multiple testing, data carving, masking, knockoffs, false discovery rate, accumulation test
\end{keywords}

\section{Introduction}

A classical statistical perspective divides data analysis into two distinct types: exploratory analysis is a flexible and iterative process of searching the data for interesting patterns while only pursuing a loose error guarantee, or none at all, while confirmatory analysis involves performing targeted inferences on questions that were pre-selected for focused study. Selective inference blends exploratory and confirmatory analysis by allowing for inference on questions that may be selected in a data-adaptive way, but most selective inference methods still require the analyst to pre-commit to selection rules before observing the data, falling short of the freewheeling nature of true exploratory analysis. By contrast, interactive methods are a subset of selective inference methods allowing the analyst to react to data, consult her own internal judgment, and revise her models and research plans to adapt to patterns she may not have expected to find, while still achieving valid inferences.

We consider the problem of multiple hypothesis testing with $p$-values $p_1,\ldots,p_n$, with each $p_i$ corresponding to a different null hypothesis $H_i$. A multiple testing method examines the $p$-values, possibly along with additional data, and decides which null hypotheses to reject. Let $\cH_0 = \{i:\; H_i \text{ is true}\}$ denote the set of truly null hypotheses and let $\cR = \{i:\; H_i \text{ is rejected}\}$ denote the rejection set. Then $R = \left|\cR\right|$ is the number of rejections and $V = \left|\cR \cap \cH_0 \right|$ is the number of erroneous rejections. \citet{bh95} defined the false discovery proportion (\textFDP) as $V/\max\{1,R\}$ and famously proposed controlling its expectation, the false discovery rate (\textFDR), at some pre-specified level $\alpha$.

This work proposes a new framework for interactive multiple testing in structured settings with \textFDR ~control, called selectively traversed accumulation rules (\star). Our method is especially well-suited to settings where we wish to impose structural constraints on the set of rejected hypotheses --- for example, to enforce a hierarchy principle for interactions in a regression or analysis of variance (\ANOVA) problem, or to detect a convex spatial region where the signal exceeds a certain level in a signal processing application. In certain cases, enforcing a structural constraint is important for logical coherence or interpretability. For instance, in the case of hierarchical testing  \citep[e.g.][]{yekutieli08} where hypotheses are represented as nodes on a tree, the parent of a rejected hypothesis must be rejected as well. In other cases, the structural constraint serves as a type of side information which reflects the scientific domain knowledge. By using this information judiciously, the statistical power may be boosted if the true signals meet the constraint, exactly or at least approximately.

More formally, our procedure controls the \textFDR ~in finite samples while guaranteeing that the rejection set $\cR\in \cK$ where $\cK$ is a collection of subsets satisfying the constraint. For instance, in the case of hierarchical testing, $\cK$ includes all subsets that correspond to rooted subtrees. The notion of structure is meant in a rather general way: $\cK$ might also depend on auxiliary covariate information $x_i$ about the hypothesis $H_i$. For instance, in the spatial testing setting, $x_i$ gives the geographic location in $\R^{2}$ and $\cK$ may include all subsets that can be written as the intersection of a convex set on $\R^{2}$ with $(x_i)_{i=1}^{n}$. 

\Star ~generalizes the notion of masking used by the adaptive p-value thresholding (\adapt) method of \citet{lei2018adapt}, which achieves its error control guarantee by judiciously limiting the analyst's knowledge about the data. As such, it is natural to view \star ~as an iterative interaction between two agents: the analyst, who drives the search for discoveries based on partial observation of the data, and a hypothetical oracle, who observes the full data set and gradually reveals information to the analyst based on her actions. Typically the oracle is a computer program and the analyst is either a human or an automated adaptive search algorithm based on pre-defined modeling assumptions. In Section~\ref{sec:implementation} we discuss generic strategies for defining good automated rules. 

A key difference between \star ~and most earlier works is that they give the analyst power to enforce structural constraints on the final rejection set. 
Previous works such as the Benjamini-Hochberg (\bh) procedure \citep{bh95}, independent hypothesis weighting (\ihw) \citep{ignatiadis2016data}, structure-adaptive Benjamini-Hochberg algorithm (\sabha) \citep{li2019multiple} and \adapt ~\citep{lei2018adapt} all produce potentially heterogeneous thresholds for $p$-values and then reject all hypotheses whose $p$-values are below their corresponding thresholds. Because any given $p$-value could be above the chosen threshold, none of these methods can enforce structural constraints. 
On the other hand, there are various other algorithms that are tailored to particular structural constraints, e.g. \cite{yekutieli08, lynch16} for hierarchical testing on trees and \cite{lynch16, ramdas2017dagger} for testing on directed acyclic graphs. However, these methods are non-adaptive and non-interactive in the sense that they can neither learn the structural information from data nor incorporate extra side information, and they apply to very specific types of constraints. To our knowledge, our method is the first data adaptive and interactive multiple testing framework that can impose generic structural constraints on the rejection set.

\section{Selectively Traversed Accumulation Rules}\label{sec:STAR}

\subsection{The framework}
We assume that for each hypothesis $H_i$, $i=1,\ldots,n$, the oracle observes a generic covariate $x_i\in \cX$ and a $p$-value $p_i\in [0,1]$; Section~\ref{sec:disc} discusses a generalization to the setting where we have knockoff statistics instead of $p$-values. In addition, the analyst may impose a generic structural constraint $\cK \subseteq 2^{[n]}$ denoting the allowable rejection sets, where $\subseteq$ denotes a subset. We require that $\emptyset \in \cK$.

\star ~proceeds by adaptively generating a sequence of candidate rejection sets $[n] = \cR_0 \supsetneq \cR_1 \supsetneq \cdots$, where $\supsetneq$ denotes a strict superset. At step $t$, the oracle estimates the \textFDP ~of the current rejection set as 
\begin{equation}\label{eq:accumFDP}
\hFDP_t = \frac{h(1) + \sum_{i\in\cR_t}h(p_i)}{1 + |\cR_t|},
\end{equation}
where $h:\; [0,1] \to [0,\infty)$ is non-decreasing and bounded, with $\int_0^1 h(p)\,dp = 1$, for example, $h(p) = 2\cdot \1\{p \geq 0.5\}$. The function $h$ is called an accumulation function, and the estimator~\eqref{eq:accumFDP} is based on the accumulation test of \citet{li2016accumulation}, itself a generalization of procedures proposed by \citet{barber15} and \citet{gsell2016sequential}. Informally, $\sum_{i\in\cR_t}h(p_i)$ plays the role of estimating $V_t = \left|\cR_t \cap \cH_0\right|$.

To allow for the analyst to make data-dependent choices without inflating Type I error, we generalize a technique by \citet{lei2018adapt} called masking. Specifically, for any choice of accumulation function $h$ as described above, we show how to derive a masking function $g$, constructed so that $h(p)$ is mean-independent of $g(p)$ when $p \sim U[0,1]$. \begin{equation}~\label{eq:masking-condition}
\E_{u \sim U[0,1]}[h(u) \mid g(u)] \eqAS \E_{u \sim U[0,1]}[h(u)] = 1.
\end{equation}
For example, if we choose $h(p) = 2\cdot \1\{p \geq 0.5\}$, then we may choose $g(p) = \min\{p,1-p\}$, by observing that for a null uniform $p$-value $p$, we have $\E[h(p)|g(p)]=\E[h(p)]=1$. Section~\ref{subsec:masking} describes a general recipe for constructing such a function $g$. Even though the masking function $g$ is constructed using condition \eqref{eq:masking-condition}, we next show that when null $p$-values are not exactly uniform, the same $g,h$ pairs satisfy a more general property that will be crucial in the proof of \textFDR ~control. The proof is presented in Appendix \ref{subapp:proof_proposition_density}.
\begin{proposition}\label{prop:density}
If the density of a null $p$-value $p$ is non-decreasing and functions $h,g$ are chosen such that condition~\eqref{eq:masking-condition} holds, then we have
\[\E[h(p) \mid g(p)] \geq 1 \,\,\mbox{ almost surely}.\]
\end{proposition}

The \star ~protocol works by first revealing all masked $p$-values $g(p_i)$s to the analyst. Then, as the analyst shrinks the rejection set, 
$p$-values that can no longer be rejected get unmasked, meaning that the analyst observes all of the $p$-values $p_i$ for $i \notin \cR_t$. Revealing $g(p)$ to the analyst is an example of data carving \citep{fithian2014optimal}, where a part of a random variable, $g(p_i)$, is used for selection, while the remainder $h(p_i)$ is used for inference. Because these two views of the data are designed to be orthogonal to each other under the null, the masked $p$-values and covariates together, along with prior information, provide guidance to the analyst on how to adaptively shrink the rejection set. Unlike other approaches \citep[e.g.][]{dwork2015reusable}, masking does not introduce extra randomness. This is desirable in scientific research to prevent cheating by specifying a favorable random seed.

At each step $t$, the oracle reports $\hFDP_t$ to the analyst. If $\hFDP_t \leq \alpha$, then the entire procedure halts and $\cR_t$ is rejected. Otherwise, the analyst is responsible to select a smaller rejection set $\cR_{t+1} \subsetneq \cR_t$ using covariates, intuition, and any desired statistical model or procedure, along with the oracle's revealed information, subject to the constraint that $\cR_{t+1} \in \cK$. After the analyst chooses $\cR_{t+1}$, the oracle then re-estimates the \textFDP ~and reveals $p_i$ for $i\in \cR_{t}\setminus \cR_{t+1}$. The analyst may then update their model, prior, constraints or intuition, and the process repeats until $\hFDP_t\le \alpha$.

The update rule from $\cR_{t}$ to $\cR_{t+1}$ is a user-specified sub-routine and should only exploit the information contained in the $\sigma$-field representing all information the analyst is allowed to observe by time $t$:
\begin{equation}\label{eq:Ft}
\cF_{t} = \sigma\lb \{ x_{i}, g(p_{i})\}_{i=1}^{n},\;  (p_{i})_{i\notin \cR_{t}},\; \sum_{i\in \cR_{t}}h(p_{i})\rb.
\end{equation}
For instance, at step $t=0$, the analyst is free to use $\{x_{i}, g(p_{i})\}_{i=1}^{n}$ to train an arbitrary model to estimate which hypotheses are most likely to be true, and choose $\cR_1$ by eliminating the most likely hypothesis subject to the constraint $\cR_1 \in \cK$. We will discuss specific rules tailored to different problems in Sections \ref{sec:convex} and \ref{sec:DAG}. For now, the update rule is any sub-routine that produces $\cR_{t+1}\subsetneq \cR_{t}$, with $\cR_{t+1}\in \cK$, and with its outcome $\cF_{t}$-measurable. Because $\cR_t$ shrinks with each step, revealing more information to the analyst, the $\sigma$-fields form a filtration with $\cF_0 \subseteq \cF_1 \subseteq \cdots$. This can be easily proved using induction; see Lemma \ref{lem:Ft_filtration} in Appendix \ref{subapp:proof_FDR} for details. As a consequence, the information available to the analyst accrues over the time. Algorithm \ref{algo:STAR} summarizes the procedure. 

\begin{algo}\label{algo:STAR}
  \Star
  \begin{tabbing}
  \quad \textbf{Input: }Predictors and $p$-values $(x_{i}, p_{i})_{i=1}^{n}$, constraint set $\cK$, target \textFDR ~level $\alpha$.\\
  \quad $\cR_{0} = [n]$\\
  \quad While $\cR_{t}\not= \emptyset$\\
      \quad \qquad\enspace $\hFDP_{t}\gets \frac{1}{1 + |\cR_{t}|}\left\{ h(1) + \sum_{i\in \cR_{t}}h(p_{i})\right\}$\\
      \quad \qquad \enspace If $\hFDP_{t} \le \alpha$ or $\cR_t = \emptyset$\\
      \quad \qquad \qquad\enspace Stop and return $\cR_t$, and reject $\{H_{i}: i\in \cR_{t}\}$\\
      \quad \qquad\enspace Select $\cF_t$-measurable $\cR_{t+1}\subsetneq\cR_{t}$ with $\cR_{t+1}\in \cK$\\
  \quad Output $\cR_{t}$ as the rejection set
  \end{tabbing}
\end{algo}

\begin{remark}
  In some applications, the analyst may want to choose the structural constraint $\cK$ after looking at data. Indeed, \star ~allows $\cK$ to be chosen based on $\cF_{0}$ and even to vary with $t$. Likewise, there is no requirement that $\cR_t \in \cK$ for every $t$, provided we require $\cR_t \in \cK$ as an additional requirement for stopping the algorithm. We discuss these details in Appendix \ref{subapp:Kt} to avoid extra complication.
\end{remark}

The mechanisms used by \star ~to enforce structural constraints and to learn structural information are different. The former is enabled by the fact that $\cR_{t+1}$ can be updated without relying on the size of $p$-values, in contrast to \bh -type algorithms, while the latter is enabled by masking functions which provide a partial view of data without sacrificing validity.

\Star ~can be viewed as a generalization of the ordered multiple testing setting of \citet{li2016accumulation}, in which a full pre-ordering of hypotheses based on outside data or prior knowledge must be supplied as an input to the analysis, and with a low-quality pre-ordering the method may be powerless, as shown in \citet{li2016accumulation} and \citet{lei2016power}. 
By contrast, for our method, a pre-ordering is just one potential source of side information that may or may not be available in any given case; other possibilities include spatial structure, covariate information, or a partial ordering from a directed acyclic graph (\DAG). Our method then determines a data-adaptive ordering using interactive guidance from the scientist, or an algorithm acting on their behalf. This interactive ordering respects any constraints the analyst has imposed, and using both the side information and the masked $p$-values $g(p_i)$. This interactive ordering, which is the counterpart to accumulation tests' pre-ordering, is simply the order in which the scientist/algorithm decides to peel off unpromising hypotheses from the candidate rejection set, based on masked $p$-values and prior information. If the ordering is pre-specified before seeing the data and masked $p$-values are ignored entirely, and there is no interaction, then our method reduces to an accumulation test, with $H_n$ peeled off first, then $H_{n-1}$, and so on.

The flexibility of our method is enabled by carving the $p$-value into two parts $g(p)$ and $h(p)$, where the first is used to adaptively determine the ordering, and the second part is used for controlling the \textFDR. Our use of the masked $p$-values is a selective-inference free lunch: compared to accumulation tests, we are using the same $h(p_i)$ values in the same way to estimate the \textFDP; we have brought more information  $g(p_i)$ to bear on guiding our methodology without inflating the \textFDR ~or requiring any further correction. In other words, accumulation tests can be thought of as also calculating $g(p_i)$ values and then simply discarding them. On the other hand, the masked $p$-values can be highly informative about which hypotheses are non-null. For example, to observe that $g(p_{i}) = \min\{p_{i}, 1 - p_{i}\} = 10^{-8}$ is extremely suggestive (a) that $H_{i}$ is most likely false, (b) that $h(p_{i}) = 2I(p_{i} > 0.5)$ is most likely 0, and possibly (c) that other hypotheses $H_{j}$ with nearby spatial locations or covariate values to $H_{i}$ are more likely to be false as well and have $h(p_{j}) \approx 0$. More generally, by examining in aggregate all of the masked $p$-values, and most of the unmasked ones too in later stages of the procedure, we can learn areas of the covariate space with many non-null hypotheses and focus the power on those by placing them closer to the front of the list, that is, by peeling them off last. When the true signals do not meet the constraint, the enforced constraint narrows down potential rejection sets and may affect power negatively compared to algorithms which do not enforce the constraint; however, the data-adaptive exploration made possible by masking allows users to learn the structure that could compensate for the power loss at no cost of false discoveries. In fact, the theory and algorithm would be exactly identical if the prespecified $\cK$ is replaced by a data-dependent constraint $\cK_t$, as long as $\cK_t$ is predictable, that is $\cK_t$ is $\cF_{t-1}$-measurable.

In Appendix \ref{app:asymptotics}, we conduct an asymptotic analysis under a slightly more general framework of \cite{li2016accumulation} to quantify the benefit of using masking functions. In a nutshell, in the absence of an informative pre-ordering, the accumulation test is powerless while the masking functions have some power in most practical cases. Moreover, even when an informative pre-ordering is available, we can still improve the power further by combining it with the masked $p$-values to obtain an even better ordering.

\subsection{False discovery rate control}

Intuitively, the quantity $\sum_{i\in\cR_t} h(p_i)$ is a conservative estimator of $V_t = |\cH_0 \cap \cR_t|$, the number of false rejections we would incur if we rejected the set $\cR_t$: each null hypothesis in $\cR_t$ contributes at least 1 in expectation, and the non-null hypotheses contribute a non-negative amount. Hence $|\cR_t|^{-1}\sum_{i\in\cR_t}h(p_i)$ can be interpreted as an upwardly biased estimator of the \textFDP ~of rejection set $\cR_t$. With the correction term in equation \eqref{eq:accumFDP}, we can prove that Algorithm~\ref{algo:STAR} controls \textFDR ~in finite samples however the analyst update the rejection set $\cR_{t}$, under certain regularity conditions on the joint distribution of $p$-values.

\begin{theorem}\label{thm:main}
Assume that 
\begin{enumerate}[\textbf{A}1]
\item the null $p$-values $(p_i)_{i\in \cH_0}$ are mutually independent, and independent of the non-nulls $(p_i)_{i\notin \cH_0}$, conditional on the covariates $(x_i)_{i=1}^n$;
\item each null $p$-value has a non-decreasing density, which may differ across $p$-values.
\end{enumerate}
If the accumulation function $h$ is non-decreasing, and the masking function $g$ satisfies condition~\eqref{eq:masking-condition}, then, conditional on $\{x_i, g(p_i)\}_{i=1}^n$, \star ~controls the \textFDR ~at level $\alpha$. Hence, unconditionally, the set $\cR_\tau$ chosen using algorithm \ref{algo:STAR} satisfies $\E\{\FDP(\cR_\tau)\} \leq \alpha$.
\end{theorem}

Assumption 1 is common in the multiple testing literature. Assumption 2 strengthens the usual assumption that a null $p$-value is stochastically larger than uniform, but is significantly weaker than assuming exact uniformity. It also strengthens the mirror-conservatism proposed in \citep{lei2018adapt}. This theorem is proved via a martingale argument, which is elaborated in Appendix \ref{subapp:proof_FDR}.

\begin{remark}\label{rem:bounded-h}
  Following \citet{li2016accumulation}, we could allow for $h$ to be unbounded and replace $h(1)$ in \eqref{eq:accumFDP} with a constant $C>0$, and halt when $\hFDP_t$ is below a corrected level:
\[
\frac{C + \sum_{i\in\cR_t}h(p_i)}{1 + |\cR_t|} ~\leq~ \alpha\int_0^1 \{h(p) \wedge C\}\,dp.
\]
However, one can show that replacing such an unbounded accumulation function by its bounded counterpart $h^C(p) = \{h(p) \wedge C\}/\int_0^1 \{h(p) \wedge C\}\,dp$ results in a strictly more powerful procedure. Thus, we assume without loss of generality that $h$ is bounded.
\end{remark}

\subsection{Masking functions}\label{subsec:masking}

We now give a recipe to construct a masking function $g$ for a generic accumulation rule $h$:

\begin{theorem}\label{thm:gp}
  Let $h:\;[0,1]\to [0,\infty)$ be any non-negative non-decreasing function with $\int_{0}^{1} h(p) = 1$, and let
\[H(p) = \int_{0}^{p}\{h(x) - 1\}\, dx.\]
Then, we have the following two conclusions: 
\begin{enumerate}[(i)]
\item There exists a continuous and strictly decreasing function $s(p)$ which is also differentiable except on a set of zero measure, namely a Lebesgue null set, such that 
\begin{equation}\label{eq:integral_equation}
H(s(p)) = H(p), \quad s(0) = 1, s(1) = 0.
\end{equation}
\item The masking function
$g(p) = \min\{p, \s(p)\}$
satisfies condition~\eqref{eq:masking-condition}.
\end{enumerate}
\end{theorem}
This theorem is proved in Appendix~\ref{subapp:proof_theorem_gp}. The aforementioned function $\s(p)$ can be obtained numerically by a quick binary search whenever $H(\cdot)$ can be computed efficiently. All accumulation functions used in this paper satisfy the conditions of the above theorem and thus their associated $s(p)$ is almost surely differentiable. Letting $\pth$ denote the unique solution to $s(\pth)=\pth$, for any $q<\pth$ the set $\{p: g(p) = q\}$ contains exactly two points and hence such a $g(p)$ is very informative because $g(p)$ only masks $1$ bit of information. For instance, if $h(p) = 2I(p\ge 0.5)$, then it is easy to show from Theorem \ref{thm:gp} that $g(p) = \min\{p, 1 - p\}$. Then $g(p) = 0.01$ implies that $p = 0.01$ or $p = 0.99$ and the analyst just needs to make a guess from two candidate values. 

For brevity, we will focus throughout the main text on the simple accumulation function $h(p) = 2I(p\ge 0.5)$ with masking function $g(p) = \min\{p, 1-p\}$; we have found that this simple choice works reasonably well in all settings we have tried. We provide several other examples in Appendix \ref{subapp:masking_functions} and examine their performance in Appendix \ref{sec:more}.

\section{Implementation}\label{sec:implementation}

\subsection{Guidance to update rejection sets}\label{subsec:guidance}

Theorem~\ref{thm:main} showed that the \textFDR ~is controlled no matter how the analyst chooses $\cR_{t+1}$ based on $\cF_t$; however, having a good update rule will be vital to operationalizing \star ~in any given context. Still, we recommend that any human-in-the-loop interactions be grounded in principled data analysis. For example, our method is to some degree susceptible to the same free-rider dynamic as other data-adaptive multiple testing procedures like AdaPT and knockoffs: namely, if we find 100 strong signals in one part of the data set, we might be able to throw in two or three more favorite hypotheses from elsewhere without threatening false discovery rate control. Well-chosen and principled constraints on the rejection set can serve as a safeguard against this behavior, which is epistemically problematic even if not formally disallowed. In addition, too-frequent looks at the data may tempt users to look back and modify their previous rejection sets; such actions are prohibited by Theorem~\ref{thm:main} and would break the theoretical guarantee.

In general, we can describe an update rule in three steps: (1) find all candidate sets of hypotheses that we can peel off from $\cR_t$ without leaving the constraint set $\cK$; (2) compute a score using all the information in $\cF_{t}$ that measures the likelihood that each candidate set has non-nulls; and (3) delete the candidate set with the worst score. We provide a flowchart in Section \ref{subapp:flowchart} summarizing the pipeline schematically. The inclusion of candidate set marks the fundamental difference between \star ~and \adapt ~because the essential candidate sets of the latter are simply all remaining hypotheses while the former operates in a more greedy way.

As a concrete example, suppose that the hypotheses correspond to vertices of a tree  and one aims at detecting a rooted subtree of signals. Then the deletion candidates are all hypotheses on leaf nodes of the subtree given by $\cR_{t}$ because deletion of any leaf node does not change the rooted subtree structure of the rejection set. 

The next step is to compute a score for each candidate. Heuristically, the score should be highly correlated to the $p$-values. As discussed in Section \ref{sec:STAR}, the most straightforward score for the $i$-th hypothesis is $g(p_{i})$. We refer to it as the canonical score. For candidate sets that contains multiple hypotheses, we define the canonical score as the average of $g(p_i)$'s. A larger canonical score gives stronger evidence that the candidate set is mostly null. Although the canonical score is straightforward to use, the user is allowed to fit any model using the covariates and the partially-masked $p$-values. Thus, one can estimate the signal strength, or a posterior probability of being null, as the score: we call this a model-assisted score. 

Finally, given the score, it is natural to remove the least favorable candidate. For instance, when using the canonical score, the candidate with largest score will be removed. On the other hand, if the score measures the signal strength or the likelihood of being non-null, the candidate with smallest score will be removed.

Our principle here can be summarized as follows: given a working model or belief about the data generating process, we have a generic EM-algorithm based pipeline that incorporates it into our method to produce scores yielding the adaptive ordering. The researcher can always incorporate their favorite data-generating model into our method, improving power if their model is correct/good, but never violating the \textFDR ~control if their model is inaccurate.

\subsection{Conditional one-group model as the working model}\label{subsec:one-group}

Although the canonical scores are effective in many cases as will be shown in Section \ref{sec:DAG}, they do not fully exploit covariate information, apart from enforcing the constraint. For instance, when the hypotheses are arranged spatially, we may expect the non-nulls will concentrate on a few clusters, and/or that the signal strength will be smooth on the underlying space. This prior knowledge may neither be reflected directly from the $p$-values nor be explicitly used to strengthen the $p$-values; instead, we can use a working model to assist calculating the scores. We emphasize that no matter how misspecified our working model is, the \textFDR~is still controlled.

\cite{lei2018adapt} proposed a conditional two-group model and used the estimated local \textFDR ~as model-assisted scores. They proved in their Theorem 2 that the local \textFDR ~gives the optimal score to order hypotheses. The model can be fitted by an expectation-maximization (\EM) algorithm \citep[Appendix A of ][]{lei2018adapt}. 

Despite the approximate optimality of the above approach, the \EM ~algorithm for conditional two-group models is computationally intensive because it fits two separate models for the proportion and the signal strength of non-nulls at each iteration. Additionally, \cite{lei2018adapt} pointed out in their Appendix A.3 some instability issues of the algorithm. For these reasons, we proposed a conditional \textit{one}-group model as an alternative:
\begin{align}
  p_i \sim f(p; \mu_{i})\quad\text{with } \eta(\mu_i) = \beta'\phi(x_i),\label{eq:onegroup_GLM}
\end{align}
where $f$ is the density function, $\eta$ is the link function and $\phi$ is an arbitrary featurization. As will be detailed in Appendix \ref{subapp:EM}, $\mu_{i}$ can be estimated via an \EM ~algorithm. The model-assisted scores are then given by the estimates $\hat{\mu}_{i}$. Although $\hat{\mu}_{i}$'s lose the optimality guarantee, they provide good proxy of how promising each hypothesis is. On the other hand, \eqref{eq:onegroup_GLM} is easier to fit as it only involves one set of parameters. At step $t$, the algorithm only needs to impute the masked $p$-values. Thus it is computationally more efficient and partly solves the issue raised in \cite{lei2018adapt}. 

Finally, as discussed in \cite{lei2018adapt}, $p$-values may not be the objects that are most amenable to modeling, in which case one can either model transformed $p$-values or directly model the data used to produce them. For instance, in many applications, z-values are available and one-sided $p$-values are obtained by the transformation $p_{i} = 1 - \Phi(z_{i})$, where $\Phi$ is the distribution function of a standard Gaussian. In this case, we can directly model $z_{i}$ as $z_{i}\sim N(\mu_{i}, 1)$.

\section{Example 1: convex region detection}\label{sec:convex}
\subsection{Problem Setup}
In some applications, the $p$-values may be associated with features $x_{i}\in \mathcal{X} \subseteq \R^d$, which encode some contextual information such as predictor variables or a spatial location, and which may be associated with the underlying signal. We may wish to use this feature information to discover regions of the feature space where the signal is relatively strong; for example, \citet{drevelegas10} seek a convex region to locate the boundary of tumors.

As a concrete mathematical example, suppose that for each point $x_i$ on a regular spatial grid, we observe an independent observation $z_i \sim N(f(x_i), 1)$ for some non-negative function $f$, and we hope to discover the region $\cC = \{x:\; f(x) > 0\}$, where we have some prior belief that $\cC$ is convex; for example, if $f$ is known to be a concave function, then $\cC$ is a superlevel set and is hence convex. Since we cannot expect to perfectly find $\cC$, we may hope to discover a smaller region $\widehat{\cC}$ which is mostly contained within $\cC$, that is, 
\begin{equation}\label{eq:convFDR}
\frac{\text{Vol}(\widehat{\cC} \cap \cC^c)}{\text{Vol}(\widehat{\cC})} \leq \alpha.
\end{equation}

We can frame the above as a multiple testing problem by computing a one-sided $p$-value $p_i = 1-\Phi(z_i)$ for each $H_i$ and constraining the rejection set to be of the form $\cR = \{i:\; x_i \in \widehat{\cC}\}$, for some convex set $\widehat{\cC}$, leading to a constraint $\cK$ on the allowable rejection sets. If the grid is relatively fine, then \textFDR ~is a natural error criterion to control since the \textFDP ~of $\cR$ approximates criterion \eqref{eq:convFDR}.

As another application in supervised learning, we may observe a pair of features $x_i$ and response $y_i$ for $i=1,\ldots,n$, and hope to find a subregion of the feature space $\cX$ where the $y$ values tend to be relatively large. In bump-hunting \citep{friedman99}, we seek to discover a rectangle in predictor space; Appendix~\ref{sec:bump} discusses this application.

More generally, we may want to discover a set of hypotheses $\cR = \{i:\; i\in \widehat{\cC}\}$ where $\widehat{\cC}$ is convex, or is a rectangle, or satisfies some other geometric property. To the best of our knowledge, no previously existing procedure can solve the above problems while guaranteeing \textFDR ~control. To implement these goals using the \star ~framework, we need only define the procedure-specific functions elaborated in Section \ref{sec:implementation}. To illustrate, we first consider an example where $x_{i}\in \R^{2}$. 

\subsection{Procedure}\label{subsubsec:convex_proc}
To preserve convexity, we consider an automated procedure that gradually peels off the boundaries of the point cloud $\{x_{i}\}$. At each iteration, we choose a direction $\theta\in [0, 2\pi)$ and a small constant $\delta$, and peel off a proportion $\delta$ of points that are farthest along this direction.

Specifically, for each angle $\theta$, we define a candidate set $C(\theta; \delta)$ to be observed as the set of indices corresponding to the $\delta$-proportion of points that are farthest along the direction $\theta$. 

If the goal is to detect an axis-parallel box, $\theta$ can be restricted to only take values in $\{0, \frac{\pi}{2}, \pi, \frac{3\pi}{2}\}$; otherwise we set $\theta$ to an equi-spaced grid on $[0, 2\pi]$ of length 100. Given a score $S_{i}$ for each hypothesis, which we soon define, we may evaluate the signal strength of each candidate set by the average of $S_{i}$. Then we may update $\cR_{t}$ as  
\[\cR_{t} := \cR_{t-1}\setminus C(\hat{\theta}; \delta), \quad\mbox{with}\quad \hat{\theta} := \argmin_{\theta}\mathrm{Avg}\{S_{i}: i\in C(\theta; \delta)\}.
\]
Finally, to define the score $S_{i}$'s, we can directly use the canonical score $g(p_{i})$. However, in most problems of this type where $x_{i}$ represents the location in some continuous space, it is reasonable to assume that the distributions of $p$-values are smoothly varying. In particular, we use the conditional one-group model on $p$-values with Beta distribution, i.e. $p_i \sim \mathrm{Beta}(1 / \mu_i, 1)$, as the working model and fit a generalized additive model \citep{hastie90} using the the smooth spline basis of $x_{i}$ as the featurization $\phi$ defined in \eqref{eq:onegroup_GLM}. This working model is motivated by the conditional two-group Gamma generalized linear model in \cite{lei2018adapt}.

\subsection{Simulation results}\label{subsubsec:convex_simulation}
We consider an artificial dataset where the predictors form an equi-spaced $50\times 50$ grid in the area $[-100, 100]\times [-100, 100]$. 
Let $\cC_{0}$ be a convex set on $\R^{2}$ and set $\cH_{0}^{c} = \{x_{i}: x_{i}\in \cC_{0}\}$. We generate $p$-values i.i.d. from a one-sided normal test, i.e. 
\begin{equation}\label{eq:simul_pvals}
p_{i} = 1 - \Phi(z_{i}), \quad \mbox{and} \quad z_{i}\sim N(\mu, 1),
\end{equation}
where $\Phi$ is the cumulative distribution function of $N(0, 1)$. For $i\in \cH_{0}$ we set $\mu = 0$ and for $i\not\in \cH_{0}$ we set $\mu = 2$. Figure~\ref{fig:convex_expr_truth} shows three types of $\cC_{0}$ that we conduct tests on.

\begin{figure}
  \centering
  \includegraphics[width = 0.75\textwidth]{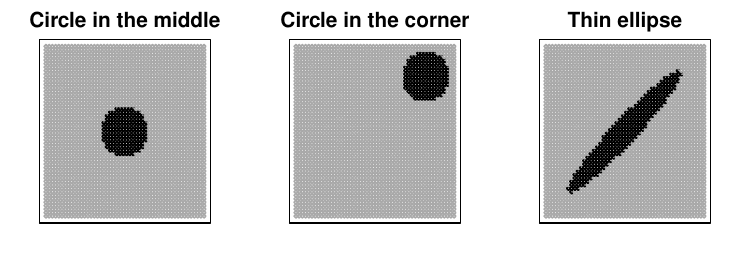}
  \caption{True underlying signal for our convex region detection simulation. Each point represents a hypothesis, 2500 in total, with black representing non-nulls with $\mu = 2$.}\label{fig:convex_expr_truth}
\end{figure}

Although our procedure is the only one that is able to enforce the convexity, it is still illuminating to compare it with other procedures to examine the power. In particular, we consider the \bh ~procedure \citep{bh95} and \adapt ~\citep{lei2018adapt}. We implement \star ~with the model-assisted score based on the generalized additive model. 

For each procedure and level $\alpha$, we calculate the \textFDP ~and the power as
\begin{align}\label{eq:defs}
\mathrm{FDP}(\alpha) &= \frac{|\cR(\alpha)\cap \cH_{0}|}{|\cR(\alpha)|}, \quad \mathrm{power}(\alpha) ~= \frac{|\cR(\alpha)\cap \cH_{0}^{c}|}{|\cH_{0}^{c}|},
\end{align}
where $\cR(\alpha)$ is the rejection set at level $\alpha$. We then estimate the \textFDR ~and the power by the average of $\mathrm{FDP}(\alpha)$ and $\mathrm{power}(\alpha)$ over 100 sets of independently generated $p$-values. The results are plotted in Figure \ref{fig:rej_convex_methods} for a list of $\alpha$'s from 0.01 to 0.3. We see that our method is comparable to \adapt, and more powerful than the \bh ~procedure, neither of which enforce the convexity constraint. Furthermore, it is worth mentioning that the model-assisted \star ~achieves high power despite using a generic and misspecified generalized additive beta model for the $p$-values.

\begin{figure}[h]
  \centering
  \includegraphics[width = 0.7\textwidth]{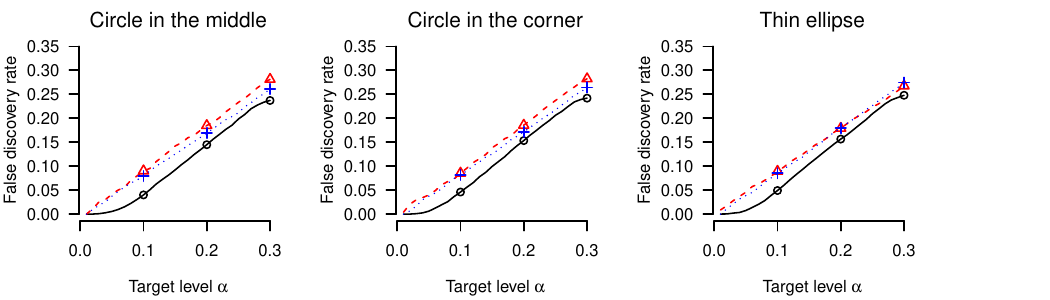}
  \includegraphics[width = 0.7\textwidth]{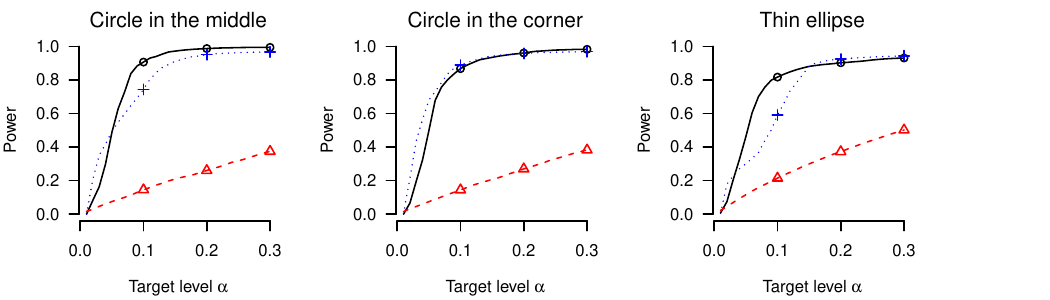}
  \caption{Comparing our method (black solid), \bh ~procedure (red dashed) and \adapt ~(blue dotted) for convex region detection.}\label{fig:rej_convex_methods}
\end{figure}

The power gain over the \bh ~procedure is in part from the fact that the convexity constraint reflects the truth and our method implicitly builds it into the selection. It is also driven by the effective learning of underlying spatial structure. To see that, 
we plot $\hat{\mu}(x)$ in Figure \ref{fig:convex_expr_score}. The top panel shows the initial score that only uses partially masked $p$-values $g(p_i)$ and $x_i$. The bottom panel shows the oracle result when fitting the model to fully observed $p$-values. It is surprising that even the initial estimate is good enough to clearly show the contour of the non-nulls, and is nearly as good as the final estimate using fully observed $p$-values; in fact, the correlation between the two estimates is above 0.98 in all three cases. This explains why our method can accurately pinpoint the non-nulls and hence enhance the power.

\begin{figure}[t]
  \centering
  \includegraphics[width = 0.75\textwidth]{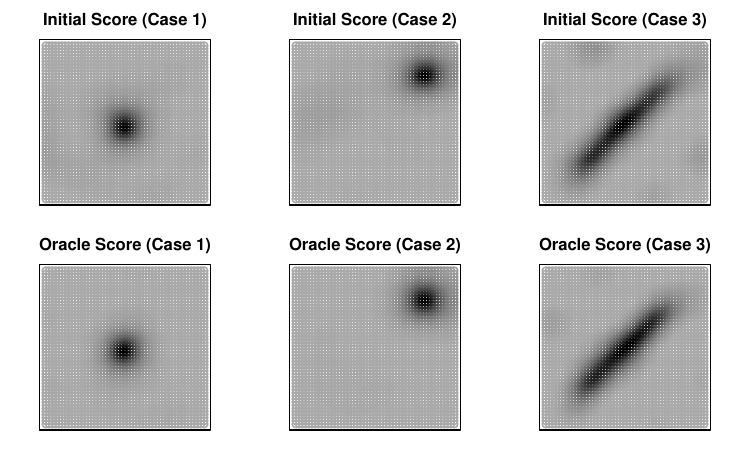}
  \caption{Model-assisted score of \star. Darker pixels represents higher score.}\label{fig:convex_expr_score}
\end{figure}

\section{Example 2: testing on directed acyclic graphs}\label{sec:DAG}
\subsection{Setup and Procedure}\label{subsubsec:DAG_proc}
Another case involves hypotheses arranged on a \DAG ~and the non-nulls are known apriori to satisfy the heredity principle. The strong or weak heredity principle, also referred to as effect hierarchical principle~\citep{wu00}, states that an effect is significant only if all or one of its parent effects are significant. An application on variable selection in factorial experiments under heredity principle is discussed in Section \ref{sec:factorial}.

Multiple testing on \DAG s has extensive applications in genomics \citep[e.g.,][]{goeman08, saunders14, meijer15} and clinical trials \citep[e.g.,][]{dmitrienko13}. However, most prior work deals with the family-wise error rate (\textFWER) control, but the setting of \textFDR ~control is relatively under-studied. To the best of our knowledge, the only existing \textFDR ~control procedures for \DAG s were proposed in Gavin Lynch's thesis \citep{lynch14} and in the sequential setting~\citep{ramdas2017dagger} and the multi-layer setting~\citep{ramdas2017unified}. All aforementioned works were designed for the strong heredity principle. 

It is straightforward to apply our method to guarantee \textFDR ~control under both the strong and the weak heredity principle. For the strong heredity principle, we select the candidates as all the leaf nodes. For the weak heredity principle, we select the candidates as all nodes by removing which the remaining graph satisfies the principle. We also present an application to a factorial experiment in Appendix \ref{sec:factorial}.

\subsection{Simulation results}

\begin{figure}[t]
  \centering
  \includegraphics[width = 0.85\textwidth]{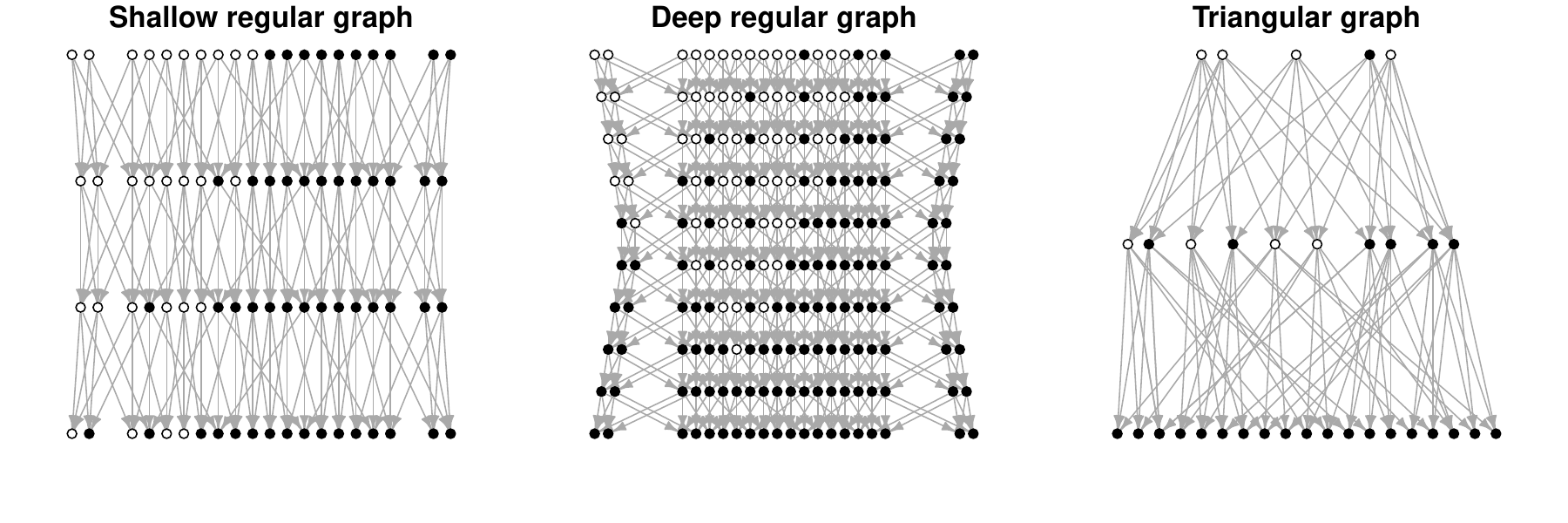}
  \caption{Three types of graphs in simulation studies, with fewer nodes for illustration. Each node represents a hypothesis, 1000 in total, with black ones being the nulls.}\label{fig:DAG_expr_truth}
\end{figure}

\begin{figure}[t]
  \centering
  \includegraphics[width = 0.7\textwidth]{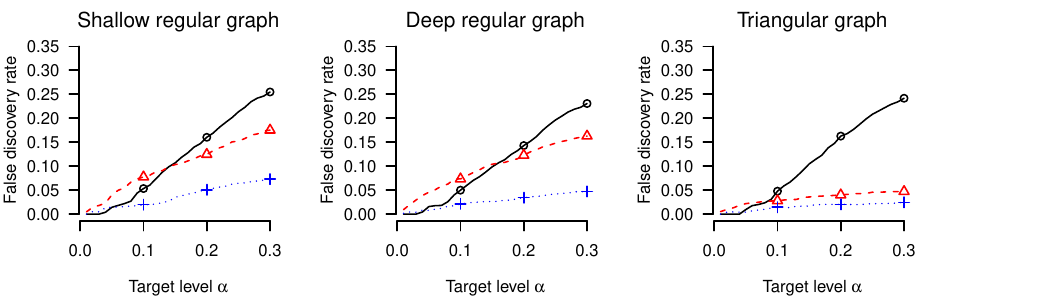}
  \includegraphics[width = 0.7\textwidth]{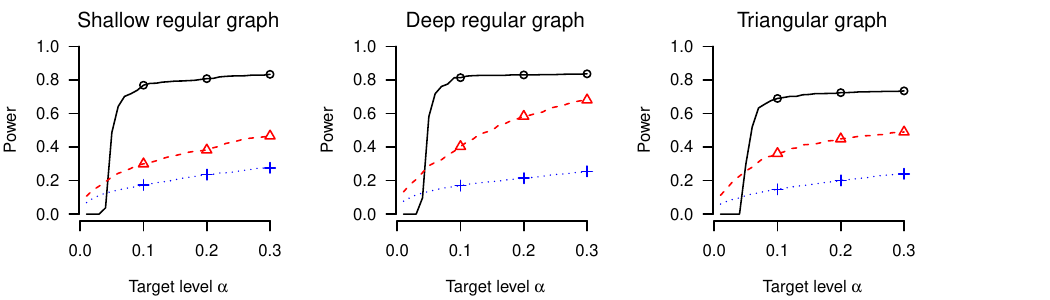}
  \caption{Comparing our method (black solid), \cite{lynch16}'s method (red dashed) and \cite{ramdas2017dagger}'s method (blue dotted) for testing on \DAG s.}\label{fig:rej_DAG_methods}
\end{figure}

We focus on the strong heredity principle in order to compare \star ~with existing methods. To account for the structure, we consider three types of \DAG s: shallow regular graph, with 4 layers and 250 nodes in each layer, deep regular graph, with 10 layers and 100 nodes in each layer, and triangular graph, with 5 layers with 50, 100, 200, 300, 350 nodes in each layer respectively. For each graph, we set 50 nodes which satisfy SHP to be non-null and generate $p$-values using equation~\eqref{eq:simul_pvals} with $\mu = 2$ for non-nulls. The settings are illustrated in Figure~\ref{fig:DAG_expr_truth}, with fewer nodes for readability.

We compare our method with canonical scores described in Section \ref{subsec:guidance}, to the method proposed in \cite{lynch14}, referred to as Self-Consistent Rejection procedure, as well as a recently developed sequential method, referred to as Greedily Evolving Rejections, a generalization of Lynch's hierarchical test by \cite{ramdas2017dagger}. The results are plotted in Figure \ref{fig:rej_DAG_methods}. It is clear that in all cases our method is more powerful than other methods when $\alpha$ is not too small. When $\alpha$ is small, our method is powerless in this setting because of the finite sample correction constant $h(1)$ in \eqref{eq:accumFDP}, which requires at least $h(1) / \alpha - 1$ rejections to get non-empty rejection set at level $\alpha$. In our case, $h(1) = 2$ but we can reduce $h(1)$ by choosing other accumulation functions. We provide a few examples that resolve this issue in Appendix \ref{sec:more}.


\section{Discussion}
\label{sec:disc}

Using knockoff statistics instead of $p$-values.
We have focused on the case where the test statistics are independent. For more general settings, our framework dovetails naturally with the knockoff framework proposed by \citet{barber15} and extended by \citet{candes2016panning}, which convert complex regression problems into independent one-bit $p$-values for each variable. When running sparse regression algorithms, the underlying variables often have some structure---for example, a tree structure with wavelet coefficients in compressed sensing---and we may want to use the knockoff procedure to select a structured subset of variables. 

Keeping this motivation in mind, let $W_i$ denote the knockoff statistic for hypothesis $i$ and define $p_i = (1+\mathbf{1}\{W_i>0\})/2$. Then, knockoff constructions guarantee that $(p_i)_{i\in \cH_0}$ are independent $p$-values, conditional on $(|W_i|)_{i=1}^n$ and $(p_i)_{i\notin \cH_0}$. The absolute value $|W_i|$ may be viewed as the free side-information $g(p_i)$ in Algorithm~\ref{algo:STAR}, while the location of the variable on the tree would be the structural information $x_i$. Although the constructed $p_i$ does not technically have a decreasing density, the accumulation function $h(p_i) = 2 I\{p_i > 1/2\} = \text{sign}(W_i) + 1$ nevertheless has expectation 1 under the null. The analyst may then use \star ~to interactively pick a structured subset of variables.

Like \adapt, knockoffs as defined in \citet{barber15} do not enforce constraints on the rejection set, nor do they allow for interactive adaptation as more $p$-values are unmasked. Combining our method with knockoff statistics allows for controlling the \textFDR ~in many interesting regression problems with constraints on the rejection set, such as hierarchy constraints for interactions in a regression. 

\vspace{0.05in}

Handling dependence.
When the underlying problem is a sparse regression problem, the aforementioned knockoff procedure can be used to construct independent $p$-values even when the covariates are not orthogonal.
However, it is not yet known how to construct knockoffs in most settings, but we can often still construct dependent $p$-values, and it is an important open problem to provide guarantees for such settings. 
As we show in Appendix \ref{sec:sensitivity}, the experiments under dependence are encouraging, especially under negative dependence, but we do not currently know how to prove any results about robustness to deviations from independence. Such results would immediately be applicable in several other settings, such as ordered testing \citep{li2016accumulation}, and knockoffs \citep{barber15}.

\bibliographystyle{biometrika}
\bibliography{biblio}

\newpage
\allowdisplaybreaks
\appendix

\begin{center}
  \begin{Large}
    \textbf{Supplementary Materials}
  \end{Large}
\end{center}

The supplementary materials include all technical proofs (Section \ref{sec:proofs}), a more detailed description of \star ~(Section \ref{sec:star}), four more applications on bump hunting, hierarchical testing, wavelet thresholding and interaction selection for factorial experiments (Section \ref{sec:bump} - \ref{sec:factorial}), comparison of \star ~with different masking functions (Section \ref{sec:more}), a sensitivity analysis for \star ~under dependent p-values (Section \ref{sec:sensitivity}) and an asymptotic power analysis for a subclass of \star ~(Section \ref{app:asymptotics}). The \texttt{R} code to replicate all results in the paper are available at \website. The sub-folder \texttt{movies} includes GIFs illustrating our method on several applications.

\section{Technical Proofs}\label{sec:proofs}

\subsection{Proof of Proposition \ref{prop:density}}\label{subapp:proof_proposition_density}
\begin{proof}
Fix any $x\in [0, 1]$ and a set $A$ with non-zero Lebesgue measure. Let $P_0$ denote the true distribution of $p$, and $P_u$ denote the uniform distribution on $[0, 1]$. Then, we have
\begin{align*}
P_0(p\in [0, x] \mid g(p)\in A) &= \frac{P_0([0, x] \cap g^{-1}(A))}{P_0(g^{-1}(A))
}\\
 & = \frac{P_0([0, x] \cap g^{-1}(A))}{P_0([0, x]\cap g^{-1}(A)) + P_0([x, 1]\cap g^{-1}(A))}\\
& = \frac{1}{1 + \frac{P_0([x, 1]\cap g^{-1}(A))}{P_0([0, x]\cap g^{-1}(A))}}.
\end{align*}
Let $f$ be the true density of $p$. Then, we may write 
\begin{align*}
  \frac{P_0([x, 1]\cap g^{-1}(A))}{P_0([0, x]\cap g^{-1}(A))}  = \frac{\int_{x}^{1} I(y\in g^{-1}(A)) f(y)dy}{\int_{0}^{x} I(y\in g^{-1}(A)) f(y)dy} &\ge \frac{f(x) \cdot \int_{x}^{1} I(y\in g^{-1}(A)) dy}{f(x) \cdot \int_{0}^{x}I(y \in g^{-1}(A))dy}\\
& =  \frac{P_u([x, 1]\cap g^{-1}(A))}{P_u([0, x]\cap g^{-1}(A))}.
\end{align*}
As a consequence, we conclude that for all $x$, we have
\[P_0(p\in [0, x] \mid g(p)\in A)\le P_u(p\in [0, x] \mid g(p)\in A),\]
and consequently, since $P_0([0,1]) = P_u([0,1]) = 1$, we also have
\[P_0(p\in [x,1] \mid g(p)\in A) \ge P_u(p\in [x,1] \mid g(p)\in A),\]
This entails that $P_u$ stochastically dominates $P_0$, when conditioning on $\{g(p)\in A\}$. Since accumulation functions $h$ are non-decreasing, they are larger on $[x,1]$ than on $[0,1]$, and hence we have
\begin{equation}\label{eq:density1}
\E_0[h(p) \mid g(p)\in A] \ge \E_u[h(p) \mid g(p)\in A].
\end{equation} 
Note that equation~\eqref{eq:density1} holds for all sets $A$ with nonzero Lebesgue measure. This immediately yields the theorem. To see this more formally, fix any $\eps > 0$, and let 
\[B_{\eps} = \{x: \E_u[h(p) \mid g(p)=x] > \E_0[h(p) \mid g(p)=x] + \eps\}.\]
Then, $B_{\epsilon}$ must be a Lebesgue null set, in order to not contradict equation \eqref{eq:density1}. 
Therefore, 
\[P_u(B_{0}) = P_u(\cup_{n\ge 1}B_{1/n})= 0.\]
As a result,
\[\E_0[h(p) \mid g(p)]\stackrel{a.s.}{\ge} \E_u[h(p) \mid g(p)]\ge 1,\]
as desired.
\end{proof}

\subsection{Proof of Theorem \ref{thm:main}}\label{subapp:proof_FDR}
We start from a lemma that shows $\cF_{t}$ defined in \eqref{eq:Ft} is a filtration.
\begin{lemma}\label{lem:Ft_filtration}
Let 
\[\cF_{t} = \sigma\lb \{x_{i}, g(p_{i})\}_{i=1}^{n},\;  (p_{i})_{i\notin \cR_{t}},\; \sum_{i\in \cR_{t}}h(p_{i})\rb.\]
Then $\cF_{t}$ is a filtration in the sense that for all $t\ge 0$,
\[\cF_{t}\subset \cF_{t+1}.\]
\end{lemma}
\begin{proof}
 The proof is completed by observing that (1) $\cF_{0} \subset \cF_{t}$ for any $t$; (2) $\cR_{u+1}\in \cF_{u}\subset \cF_{t}$ implies $\cF_{u+1}\subset \cF_{t}$.
\end{proof}

The proof of Theorem \ref{thm:main} is based on an optional stopping argument, generalizing the one presented in \citet{lei2016power}, which in turn generalized arguments from \citet{li2016accumulation} and \citet{barber2016knockoff}. 

\begin{lemma}\label{lem:bernoulli}[Lemma 1 of \citet{lei2018adapt}]
  Suppose that, conditionally on the $\sigma$-field $\cG_{-1}$, $b_1,\ldots,b_n$ are independent Bernoulli random variables with 
\[\P(b_i = 1 \mid \cG_{-1}) = \rho_i \geq \rho > 0, \mbox{almost surely}.\]
Let $(\cG_{t})_{t=0}^{\infty}$ be a filtration with $\cG_{0}\subseteq \cG_{1}\subseteq \cdots$ and suppose that $[n] \supseteq \cC_0 \supseteq \cC_1 \supseteq \cdots$, with each subset $\cC_{t+1}$ measurable with respect to $\cG_{t}$. If we have
  \[
  \cG_t \supset \sigma\left(\cG_{-1}, \cC_t, (b_i)_{i \notin \cC_t}, \sum_{i \in \cC_t} b_i\right),
  \]
  and $\htt$ is an almost-surely finite stopping time with respect to the filtration $(\cG_t)_{t \geq 0}$, then
  \[
  \E\left[\frac{1 + |\cC_{\htt}|}{1 + \sum_{i\in \cC_{\htt}} b_i} \right]  \leq \rho^{-1}.
  \]
\end{lemma}

\begin{proof}[\textbf{of Theorem~\ref{thm:main}}]
\label{sec:proof-thm1}
By Proposition \ref{prop:density}, 
\begin{equation}\label{eq:conditional_expectation_1}
\E[h(p_i)\mid g(p_i)]\stackrel{a.s.}{\ge} 1,\quad \forall i\in \cH_{0}.
\end{equation}
Since $h$ is non-decreasing, \eqref{eq:conditional_expectation_1} implies that $h(1)\ge 1$. Generate $(V_{i})_{i\in \cH_{0}}\stackrel{i.i.d.}{\sim} U([0, 1])$, which are also independent of $(x_{i}, p_{i})_{i=1}^{n}$ and all operational randomness involved in the procedure. Let $b_i = I(V_i\le \frac{h(p_{i})}{h(1)})$ and recall that $\tau$ is the smallest $t$ such that $\hFDP_{t}\le \alpha$, then
\begin{align*}
\FDR & = \E [\FDP_\tau ] = \E\left[ \frac{|\cR_\tau \cap \cH_{0}|}{1\vee |\cR_\tau |}\right]\le \E\left[ \frac{1 + |\cR_\tau \cap \cH_{0}|}{1 + |\cR_\tau |}\right]\\
& \stackrel{(\mathrm{i})}{\le} \frac{\alpha}{h(1)}\cdot \E\left[ \frac{1 + |\cR_\tau \cap \cH_{0}|}{1 + \sum_{i\in \cR_{\tau}}\frac{h(p_{i})}{h(1)}}\right]\\
& \stackrel{(\mathrm{ii})}{\le} \frac{\alpha}{h(1)}\cdot \E\left[ \frac{1 + |\cR_{\tau}\cap \cH_{0}|}{1 + \sum_{i\in \cR_{\tau}\cap\cH_{0}}\E[b_{i} \mid p_{i}]}\right]\\
& \stackrel{(\mathrm{iii})}{\le} \frac{\alpha}{h(1)}\cdot \E\left[ \frac{1 + |\cR_{\tau}\cap \cH_{0}|}{1 + \sum_{i\in \cR_{\tau}\cap\cH_{0}}\E[b_{i} \mid (p_{i})_{i\in \cH_{0}}]}\right]\\
& \stackrel{(\mathrm{iv})}{\le} \frac{\alpha}{h(1)} \cdot \E\left[\E\left[ \frac{1 + |\cR_{\tau}\cap \cH_{0}|}{1 + \sum_{i\in \cR_{\tau}\cap\cH_{0}}b_{i}} \mid (p_{i})_{i\in \cH_{0}}\right]\right]\\
& = \frac{\alpha}{h(1)} \cdot \E\left[ \frac{1 + |\cR_{\tau}\cap \cH_{0}|}{1 + \sum_{i\in \cR_{\tau}\cap\cH_{0}}b_{i}} \right],
\end{align*}
where (i) follows because at time $\tau$, we have $\frac{h(1) + \sum_{i \in \cR_{\tau}} h(p_i)}{1 + |\cR_\tau|} \leq \alpha$, (ii) follows by substituting the definition of $b_i$ and restricting the indices of the denominator summation to just the rejected nulls, (iii) follows because of the independence of null p-values, while (iv) uses Jensen's inequality and the convexity of the mapping $y\mapsto \frac{1}{1 + y}$. Define the initial $\sigma$-field as
\[\cG_{-1} = \sigma\bigg( \{x_{i}, g(p_{i})\}_{i=1}^{n}, (p_{i})_{i\not\in \cH_{0}}\bigg).\]
Then 
\begin{equation}\label{eq:bi_cond}
\E [b_{i} | \cG_{-1}] = \E \left[\frac{h(p_{i})}{h(1)}\mid \cG_{-1}\right] = \E \left[\frac{h(p_{i})}{h(1)}\mid g(p_{i})\right]\ge \frac{1}{h(1)} \quad \text{ for all } i \in \cH_{0}.
\end{equation}
Recall that $\cR_{0} = [n]$ and define the filtration $(\cG_{t})_{t\ge 0}$ as
\begin{align*}
\cG_{t} &= \sigma\bigg( \cG_{-1}, (p_{i}, V_{i})_{i\notin \cR_{t}\cap \cH_{0}}, \{(p_{i}, V_{i}) : {i\in \cR_{t}\cap \cH_{0}}\}\bigg),
\end{align*}
where $\{\cdot\}$ denotes the unordered set. Then we have the following observations:
\begin{enumerate}[(a)]
\item Since $\cR_{0}\supset \cR_{1}\supset\cdots$, we necessarily have $\cG_{0}\subset \cG_{1}\subset\cdots$.
    
\item By definition \eqref{eq:Ft}, note that we have
\[\cF_{t} = \sigma\bigg(\cG_{-1}, (p_{i})_{i\notin \cR_{t}}, \sum_{i\in\cR_{t}}h(p_{i})\bigg)\subseteq \sigma\bigg(\cG_{-1}, (p_{i})_{i\notin \cR_{t}}, \{p_{i}: i\in \cR_{t}\}\bigg) \subseteq \cG_t.\]
As a consequence, $\tau \leq n$ is also a finite stopping time with respect to filtration $(\cG_{t})_{t\ge 0}$. 
\item Since $b_{i}$ is a function of $(p_{i}, V_{i})$ and $(p_{i}, V_{i})_{i\not\in \cR_{t}\cap \cH_{0}} \in \cG_{t}$, we have
\begin{equation}\label{eq:accept_bi}
(b_{i})_{i\not\in \cR_{t}\cap \cH_{0}}\in \cG_{t}.
\end{equation}
\item Lastly, observe that
\[
\sum_{i\in \cR_{t}\cap \cH_{0}}b_{i} \in \sigma\bigg(\{b_{i}: i\in \cR_{t}\cap \cH_{0}\}\bigg)\subseteq \sigma\bigg(\{(p_{i}, V_{i}): i\in \cR_{t}\cap \cH_{0}\}\bigg) \subseteq \cG_t.
\]
\end{enumerate}
Putting the pieces together  and applying Lemma \ref{lem:bernoulli} with $\cC_{t} = \cR_{t}\cap \cH_{0}$, we conclude that
\begin{equation}\label{eq:FDR}
\E\left[ \frac{1 + |\cR_{\tau}\cap \cH_{0}|}{1 + \sum_{i\in \cR_{\tau}\cap\cH_{0}}b_{i}}\right]\le h(1).
\end{equation}
As a result, we may conclude that 
\[\FDR ~\le~ \frac{\alpha}{h(1)} \cdot \E\left[ \frac{1 + |\cR_{\tau}\cap \cH_{0}|}{1 + \sum_{i\in \cR_{\tau}\cap\cH_{0}}b_{i}}\right] ~\le~ \alpha,\]
as claimed by the theorem.
\end{proof}

\subsection{Proof of Theorem \ref{thm:gp}}\label{subapp:proof_theorem_gp}
\label{sec:proof-thm2}
\begin{proof}
  We first prove statement (i). We start by assuming that $\{p: h(p) = 1\}$ is a Lebesgue null set. Since $h$ is non-decreasing and $\int_{0}^{1}h(p)\, dp = 1$, we must have $h(0) < 1 < h(1)$. Let $\pth = \sup\{p: h(p)\le 1\}$, then
\[\lim_{p\uparrow \pth} h(p)\le 1 \le \lim_{p\downarrow \pth}h(p).\]
As a consequence, $H(x)$ is strictly decreasing on $[0, \pth]$ and strictly increasing on $[\pth, 1]$ and $H(0) = H(1) = 0$. First we define the function $s(p)$ on $[\pth, 1]$: for any $p\in [\pth, 1]$, let $s(p)$ be the unique solution on $[0, \pth]$ such that $H(s(p)) = H(p)$. Then, it is easy to see that $s(\cdot)$ is strictly decreasing on $[\pth, 1]$ with $s(1) = 0$ and $s(\pth) = \pth$. Since the function $H$ is continuous and strictly decreasing on $[0, \pth]$, we know that $s(\cdot)$ is continuous on $[\pth, 1]$. Similarly we can define $s(p)$ on $[0, \pth]$. The continuity is guaranteed at $\pth$ since $s(\pth) = \pth$. 

~\\
Next we prove that $s(\cdot)$ is differentiable except on a Lebesgue null set. Let 
\[\mathcal{D} = \{p: h(\cdot)\mbox{ is continuous at both }p\mbox{ and }s(p)\}.\]
Since $h(\cdot)$ is increasing on $[0, 1]$, the standard argument in real analysis (e.g. \cite{rudin}) implies that $\mathcal{D}$ is countable and hence a Lebesgue null set. It is left to prove that $s(\cdot)$ is differentiable on $\mathcal{D}^{c}$. By the definition of $s(\cdot)$, for any $0\le p_{1}\le p_{2}\le 1$,
\begin{equation}\label{eq:differentiability}
\frac{H(s(p_{2})) - H(s(p_{1}))}{s(p_{2}) - s(p_{1})} = \frac{H(p_{2}) - H(p_{1})}{s(p_{2}) - s(p_{1})} = \frac{H(p_{2}) - H(p_{1})}{p_{2} - p_{1}}\cdot \frac{p_{2} - p_{1}}{s(p_{2}) - s(p_{1})}.
\end{equation}
Take any $p\in \mathcal{D}^{c}$ and by definition we know that $h(\cdot)$ is continuous on both $p$ and $s(p)$. By Newton-Leibniz theorem \citep{rudin}, 
\[H'(p) = h(p) - 1, \quad \text{ and ~}~ H'(s(p)) = h(s(p)) - 1.\]
Now, letting $p_{1} = p$ and $p_{2}\rightarrow p$ in \eqref{eq:differentiability}, the continuity of $s(\cdot)$ implies that $s(p_{2})\rightarrow s(p)$ and the differentiability of $H(\cdot)$ at $p$ implies that 
\[h(s(p)) - 1 = H'(s(p)) = H'(p)\cdot \lim_{p_{2}\rightarrow p}\frac{p_{2} - p}{s(p_{2}) - s(p)} = (h(p) - 1)\cdot \lim_{p_{2}\rightarrow p}\frac{p_{2} - p}{s(p_{2}) - s(p)}.\]
This entails that the derivative of $s(p)$ can be written as
\begin{equation}\label{eq:sderiv}
s'(p)\triangleq \lim_{p_{2}\rightarrow p}\frac{s(p_{2}) - s(p)}{p_{2} - p} = \frac{h(p) - 1}{h(s(p)) - 1}.
\end{equation}
Now suppose $\{p: h(p) = 1\}$ is not a Lebesgue null set. Since $h$ is non-decreasing, it must be an interval. Let $[p_{1}, p_{2}]$ be the closure of $\{p: h(p) = 1\}$. Then $H$ is strictly decreasing on $[0, p_{1}]$, strictly increasing on $[p_{2}, 1]$ and is flat on $[p_{1}, p_{2}]$. For $p\in [0, p_{1})\cup (p_{2}, 1]$, we can define $s(p)$ is the same way as above. By construction, $s(p_{2}) = p_{1}, s(p_{1}) = p_{2}$ On $[p_{1}, p_{2}]$, we simply define $s(p)$ as the linear interpolation between $p_{1}$ and $p_{2}$, i.e. $s(p) = p_{1} + p_{2} - p$. It is easy to see that $s(p)$ is continuous, strictly decreasing and differentiable almost everywhere. 

Now we prove theorem statement (ii). Take any $q\not\in s(\mathcal{D})$ and write $s^{-1}(q)$ as $\td{q}$ for short. Note that $\{p: g(p) = q\}$ only contains two points $\{q, \td{q}\}$. If $g^{-1}(q)\subseteq \{p: h(p) = 1\}$, then 
\[\E_{p\sim U([0, 1])} [ h(p) \mid  g(p) = q] = 1.\]
Otherwise, by equation~\eqref{eq:sderiv} and the fact that $q = s(\td{q})$, we infer that 
\begin{align}
\E_{p\sim U([0, 1])} [ h(p) \mid  g(p) = q] &= \frac{h(q) - h(\td{q}) / \s'(\td{q})}{1 - 1 / \s'(\td{q})}.\nonumber\\
& = \frac{h(q) - h(\td{q}) (h(q) - 1) / (h(\td{q}) - 1)}{1 - (h(q) - 1) / (h(\td{q}) - 1)} ~=~ 1.\label{eq:ode_pre}
\end{align}
On the other hand, since $s(\cdot)$ is strictly decreasing, $\{p: s(p) = q\}$ contains at most two points for any $q$. As a result,
\begin{align*}
P_{p\sim U([0,1])}\lb \E_{p\sim U([0, 1])} [ h(p) \mid  g(p) = q] = 1\rb &\ge 1 - P_{p\sim U([0, 1])}\lb s^{-1}(\mathcal{D})\rb\\ &\ge 1 - 2P_{p\sim U([0, 1])}(\mathcal{D}) ~=~ 1.
\end{align*}
Hence, we have proved that our choice of $g$ satisfies condition~\eqref{eq:masking-condition}, and this concludes the proof of the theorem.
\end{proof}

\section{More Details About Selectively Traversed Accumulation Rules}\label{sec:star}
\subsection{Flowchart of the framework}\label{subapp:flowchart}
The scheme in Section \ref{subsec:guidance} is presented explicitly in Figure~\ref{fig:flowchart}, and the three steps of a generic update rule are highlighted in red.

\begin{figure}[h]
  \centering
  \includegraphics[width = 0.85\textwidth]{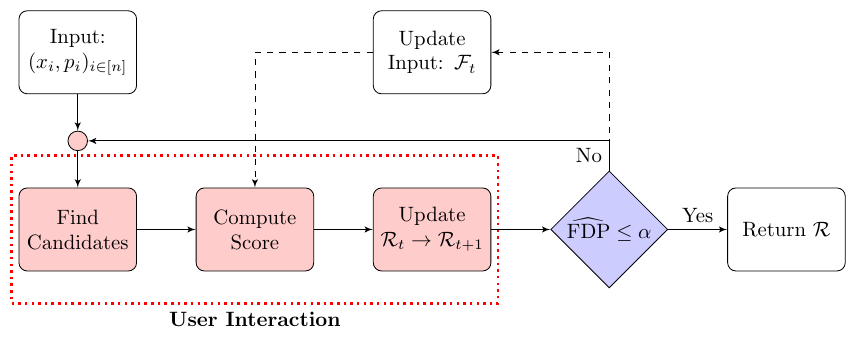}
  \caption{Flowchart of the framework}\label{fig:flowchart}
\end{figure}

\subsection{Data adaptive structural constraint}\label{subapp:Kt}
Let $\cK_{0}, \cK_{1}, \ldots$ be a sequence of structural constraints with $\cK_{t}\in \cF_{t}$. We can then generalizes Algorithm \ref{algo:STAR} by incorporating time-varying structural constraints and allowing the rejection set to temporarily leave the constraint. For example, if the analyst had started by wanting to find a convex set, but the masked p-values very clearly reveal a banana shape, or two circles in opposite corners of the grid, then she can change her mind and update $\cK$.

\begin{algo}\label{algo:STARKt}
  {\em \Star}
  \begin{tabbing}
  \quad \textbf{Input: }Predictors and $p$-values $(x_{i}, p_{i})_{i=1}^{n}$, constraint set $\cK$, target \textFDR ~level $\alpha$.\\
  \quad $\cR_{0} = [n]$\\
  \quad While $\cR_{t}\not= \emptyset$\\
      \quad \qquad\enspace $\hFDP_{t}\gets \frac{1}{1 + |\cR_{t}|}\lb h(1) + \sum_{i\in \cR_{t}}h(p_{i})\rb$\\
      \quad \qquad \enspace If ($\hFDP_{t} \le \alpha$ and $\cR_t \in \cK_t$) or $\cR_t = \emptyset$\\
      \quad \qquad \qquad\enspace Stop and return $\cR_t$, and reject $\{H_{i}: i\in \cR_{t}\}$\\
      \quad \qquad\enspace Select $\cR_{t+1}\subseteq\cR_{t}$ with $\cR_{t+1}\in \cK_{t}\cap \cF_{t}$\\
      \quad \qquad\enspace Select $\cK_{t+1}\in\cF_{t+1}$\\
  \quad Output $\cR_{t}$ as the rejection set\\
  \end{tabbing}
\end{algo}

\subsection{Examples of masking functions}\label{subapp:masking_functions}
We show masking functions of several accumulation functions that are used in literature.
\begin{enumerate}
\item (SeqStep, \cite{barber15}) When $h(p) = \frac{1}{1-\pth} 1\{ p > \pth\}$, one may derive
\[s(p) = \frac{\pth}{1-\pth}(1 - p).\]
\item (ForwardStop, \cite{gsell2016sequential}) For the unbounded accumulation function $h(p) = -\log(1 - p)$, we can obtain a bounded function $h^C(p)$ by truncating at $C>0$ and renormalizing as in Remark~\ref{rem:bounded-h}; in order to avoid a large renormalization (corresponding to a large correction of the \textFDR ~level), we fix $C = -\log(0.01) = 4.605$, in which case $\int_0^1(h(p)\wedge C)\,dp = 0.99$. For any $C>0$, one can derive
\[H(p; C) = \left\{
    \begin{array}{ll}
      e^{-C}p + (1 - p)\log (1 - p) & \text{ if } p\le 1 - e^{-C},\\
      (1 - p)(1 - C - e^{-C}) & \text{ if } p > 1 - e^{-C}.
    \end{array}
\right.\]
  and solve for $s(p)$ numerically as shown in Figure~\ref{fig:masking_fun}.

\item (HingeExp, \cite{li2016accumulation}) ForwardStop may be generalized to obtain the unbounded accumulation function $h(p) = \frac{1}{1 - \pth}\log \frac{1 - \pth}{1 - p}1\{p\ge \pth\}$ for some $\pth\in (0, 1)$ ($\pth=0$ gives ForwardStop after reparametrization).
Using a similar reasoning to ForwardStop, for each $\pth$ we recommend truncating $h(p)$ at $C = \frac{-\log (0.01)}{1 - \pth}$ so that $\int_0^1(h(p)\wedge C)\,dp = 0.99$.
After truncating and renormalizing using any $C>0$, we have
\[
H(p; C) = \left\{
    \begin{array}{ll}
      -p\int_0^1(h(p)\wedge C)\,dp & \text{ if } p < \pth,\\
      e^{-C(1-\pth)}p + \frac{1-p}{1-\pth}\log\frac{1-p}{1-\pth} - \frac{\pth}{1-\pth}(1-p) & \text{ if } \pth\le p\le 1 - (1 - \pth)e^{-C(1 - \pth)},\\
      (1 - p)(1 - C - e^{-C(1-\pth)}) & \text{ if } p > 1 - (1 - \pth)e^{-C(1 - \pth)}.
    \end{array}
\right.
\]
Once more we can calculate $s(p)$ numerically, as shown in Figure~\ref{fig:masking_fun} below.
\end{enumerate}

\begin{figure}[h]
  \centering
  \includegraphics[width = 0.95\textwidth]{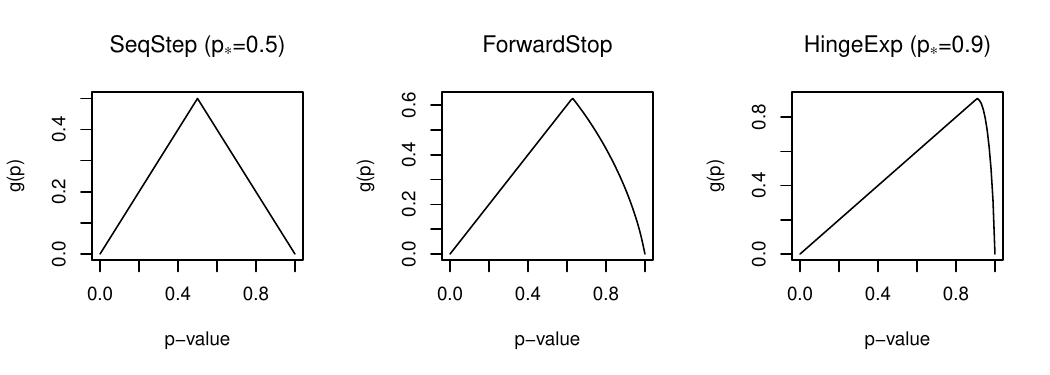}
  \caption{Masking functions for different accumulation functions.}\label{fig:masking_fun}
\end{figure}

\subsection{Details of \EM ~algorithm}\label{subapp:EM}
Consider the working model \eqref{eq:onegroup_GLM}. At step $t$, let 
\[\td{p}_{t, i} = p_{i}I(i\not\in \cR_{t}) + g(p_{i})I(i\in \cR_{t}),\]
where $g(p) = \min\{p, s(p)\}$. For simplicity, we assume that $f(p; \theta)$ is the model of the original $p$-values. Note that the following derivation directly carries over to the transformed $p$-values.

Define a sequence of hypothetical labels $w_{t, i} = I(\td{p}_{t, i} = p_{i})$. Note that for unmasked p-values, $w_{t, i} = 1$. Then the joint log-likelihood of $\{\td{p}_{t, i}\}$ and $\{w_{t, i}\}$ is 
\begin{align}
&\ell(\{\td{p}_{t, i}, w_{t, i}\}) = \sum_{i\not\in \cR_{t}}\log f(\td{p}_{t, i}; \mu(x_{i}))\nonumber\\ 
& \quad + \sum_{i\in \cR_{t}}w_{t, i}\log f(\td{p}_{t, i}; \mu(x_{i})) + \sum_{i\in \cR_{t}} (1 - w_{t, i}) \log f(s^{-1}(\td{p}_{t, i}); \mu(x_{i})).\nonumber
\end{align}

The standard \EM ~algorithm replaces $w_{t, i}$ by its conditional mean $\E (w_{t, i} \mid \td{p}_{t, i})$ in the E-step. Using a similar argument as equation~\eqref{eq:ode_pre}, we have
\begin{align}
\td{w}_{t, i}\triangleq \E (w_{t, i} \mid \td{p}_{t, i}, \theta_{\mathrm{old}}(x_{i}))& = \P (p_{i} = \td{p}_{t, i} \mid g(p_{i}) = \td{p}_{t, i})\nonumber\\ 
& = \frac{f(\td{p}_{t, i}; \theta_{\mathrm{old}}(x_{i}))}{f(\td{p}_{t, i}; \theta_{\mathrm{old}}(x_{i})) - (s^{-1})'(\td{p}_{t, i})\cdot f(s^{-1}(\td{p}_{t, i}); \theta_{\mathrm{old}}(x_{i}))}.\label{eq:Estep}
\end{align}
where $\theta_{\mathrm{old}}(\cdot)$ is from the last iteration. Here $(s^{-1})'(\cdot) = 1 / s'(s^{-1}(\cdot))$ is known to exist almost everywhere by Theorem \ref{thm:gp}. Then in the M-step, we replace $\theta_{\mathrm{old}}(\cdot)$ by 
\begin{align}
& \theta_{\mathrm{new}}(\cdot) = \argmax_{\theta(\cdot)\in \Theta}  \sum_{i\not\in \cR_{t}}\log f(\td{p}_{t, i}; \theta(x_{i}))\nonumber\\ 
& \quad + \sum_{i\in \cR_{t}}\td{w}_{t, i}\log f(\td{p}_{t, i}; \theta(x_{i})) + \sum_{i\in \cR_{t}} (1 - \td{w}_{t, i}) \log f(s^{-1}(\td{p}_{t, i}); \theta(x_{i})).\label{eq:Mstep}
\end{align}
The above optimization problem is equivalent to solving a weighted MLE on an artificial dataset $\{(\td{p}_{t, i})_{i=1}^{n}, (s^{-1}(\td{p}_{t, i}))_{i\in \cR_{t}}\}$. Therefore any algorithm that solves the weighted MLE can be embedded into this framework. 

We provide two instantiations below, which will be used in later sections. For illustration, we only consider the accumulation function $h(p) = 2 I(p\ge 0.5)$ with $s(p) = 1 - p$. 

\vspace{0.3em}
\noindent \textbf{Example 1: Beta family for p-values.} Consider the model \eqref{eq:onegroup_GLM} with $h(p; \mu) = \frac{1}{\mu}p^{\frac{1}{\mu} - 1}$. The E-step \eqref{eq:Estep} simplifies to\begin{equation*}
  \td{w}_{t, i} = \frac{\td{p}_{t,i}^{\frac{1}{\beta_{\mathrm{old}}'\phi(x_{i})} - 1}}{\td{p}_{t,i}^{\frac{1}{\beta_{\mathrm{old}}'\phi(x_{i})} - 1} + (1 - \td{p}_{t, i})^{\frac{1}{\beta_{\mathrm{old}}'\phi(x_{i})} - 1}},
\end{equation*}
and the M-step \eqref{eq:Mstep} can be calculated as 
\begin{align}
   \beta_{\mathrm{new}} &= \argmax_{\beta\in \R^{m}} \sum_{i\not\in \cR_{t}}\{(\log \td{p}_{t, i})\beta_{\mathrm{old}}'\phi(x_{i}) + \log (\beta_{\mathrm{old}}'\phi(x_{i}))\}\nonumber\\ 
& \quad + \sum_{i\in \cR_{t}}\left\{\lb\td{w}_{t, i}\log \td{p}_{t, i} + (1 - \td{w}_{t,i})\log \lb 1 - \td{p}_{t,i}\rb\rb\beta_{\mathrm{old}}'\phi(x_{i}) + \log (\beta_{\mathrm{old}}'\phi(x_{i}))\right\}.\nonumber
\end{align}
Define $y_{t, i}$ as
\[y_{t, i} = -\log(\td{p}_{t, i})I(i\not\in \cR_{t}) - \lb\td{w}_{t, i}\log \td{p}_{t, i} + (1 - \td{w}_{t,i})\log \lb 1 - \td{p}_{t,i}\rb\rb I(i\in \cR_{t}).\]
Then, we have
\begin{equation*}
  \beta_{\mathrm{new}} = \argmax_{\beta\in \R^{m}} \sum_{i=1}^{n}\{-y_{t, i}\beta_{\mathrm{old}}'\phi(x_{i}) + \log (\beta_{\mathrm{old}}'\phi(x_{i}))\},
\end{equation*}
which is equivalent to the solution of an \emph{unweighted} Gamma generalized linear model with a inverse link function on data $\{y_{t, i}\}$ with covariate $\phi(x_{i})$.

\vspace{0.3em}
\noindent \textbf{Example 2: Gaussian family for z-values.} Consider the model $p_{i} = 1 - \Phi(z_{i})$ with $z_{i}\sim N(\mu_{i}, 1)$. Define the partially-masked z-values as
\[\td{z}_{t, i} = \max\{z_{i}, \Phi^{-1}(1 - (1 - p_{i}))\} = \max\{z_{i}, -z_{i}\} = |z_{i}|.\]
Thus the E-step \eqref{eq:Estep}, replacing the $p$-values by z-values, can be simplified as
\[\td{w}_{t, i} = \frac{\exps{-\frac{(|z_i| - \mu_{i,\mathrm{old}})^2}{2}}}{\exps{-\frac{(z_i - \mu_{i,\mathrm{old}})^2}{2}} + \exps{-\frac{(-z_i - \mu_{i,\mathrm{old}})^2}{2}}} = \frac{1}{1 + \exps{-2\mu_{i, \mathrm{old}}\cdot |z_i|}}.\]
In the M-step, $\mu_{i}$'s are updated by
\begin{align}
(\mu_{i, \mathrm{new}}) &= \argmin ~\sum_{i\not\in \cR_{t}}(z_{i} - \mu_{i})^2 + \sum_{i\in \cR_{t}}\td{w}_{t, i}(z_{i} - \mu_{i})^{2} + (1 - \td{w}_{t, i})(-z_{i} - \mu_{i})^2\nonumber\\
& = \argmin ~ \sum_{i\not\in \cR_{t}}(z_{i} - \mu_{i})^2 + \sum_{i\in \cR_{t}}((2\td{w}_{t, i} - 1)z_{i} - \mu_{i})^{2},\nonumber
\end{align}
which reduces to an \emph{unweighted} least-squares problem on a pseudo-dataset $\{\td{z}_{t, i}: i = 1, \ldots, n\}$ where $\td{z}_{t, i} = z_{i}$ for unmasked hypotheses and $\td{z}_{t, i} = (2\td{w}_{t, i} - 1)z_{i}$ for masked hypotheses. Note that we can solve it as a non-parametric least-squares problem if $\phi(x)$ corresponds to some basis functions, or as a constrained problem with $\mu_{i}$'s lying in an isotonic cone.

\section{Example 3: bump hunting}\label{sec:bump}

Bump hunting is widely applied in areas such as astronomy~\citep{good80}, risk management~\citep{becker01}, bioinformatics~\citep{jiang06}, and epidemiology~\citep{jaffe12}. In these areas, one collects a response $y$ together with a possibly high dimensional vector $x$ of predictors and aims at obtaining knowledge of $f(x) = \E [y | x]$. In many applications, it is not necessary to estimate $f(x)$ uniformly over the domain but simply detect a scientifically interesting subregion of the predictor space instead. In bump hunting, we usually aim to detect a subregion  within which the average of $y$ is larger than that on the entire space. However, most existing procedures lack formal statistical guarantees.

We can cast the problem as a nonparametric multiple testing problem by defining the null hypothesis $H_i$ that the conditional response distribution at the $i$th data point is 
\[
H_i:\; \mathcal{L}(y_i | x_{i})\preceq \mathcal{L}(y_i + B),
\]
where $\preceq$ denotes stochastic dominance and $\mathcal{L}$ denotes the marginal or conditional distribution of $y$. Informally, we wish to find a clustered set of non-nulls, corresponding to a rectangular region of the feature space where the response is unusually large, by some fixed location offset $B\geq 0$.

Let $F_{0}$ denote the marginal distribution function of $y$. If $F_{0}$ is known, one can define the $p$-value as $p_{i} = 1 - F_{0}(y_{i}-B)$ (if $F_0$ is not continuous, we may use a randomized version instead). To discover a rectangular region, we can apply the convex region detection algorithm of Section~\ref{subsubsec:convex_proc} with the restriction that we always peel off an axis-parallel rectangle in the form of $\{i: x_{ij}\ge v_{j}\}$ or $\{i: x_{ij}\le v_{j}\}$. More precisely, given a patience parameter $\delta\in (0, 1)$, the candidate sets are given by $\{C(j, b; \delta): j \in \{1,\ldots, p\}, b \in \{-1, 1\}\}$, where 
\[C(j, -1; \delta) = \{i: x_{ij}\le v_{j}\}, \quad C(j, 1; \delta) = \{i: x_{ij}\ge v_{j}\},\]
and $v_{j}$ is set to be the minimal value such that $|C(j, b; \delta)| \ge \lceil n \delta\rceil$. 

\begin{figure}[htp]
  \centering
  \includegraphics[width = 0.7\textwidth]{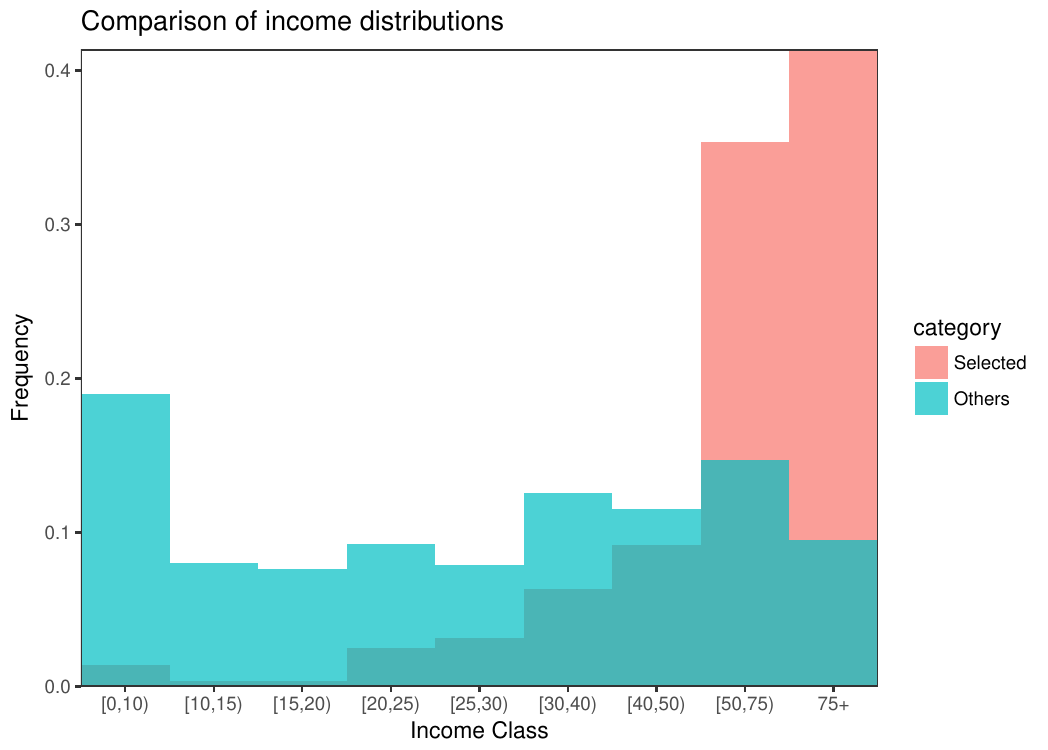}
  \caption{Comparison of income distributions before and after selection by \star.}\label{fig:income_dist}
\end{figure}

For illustration, we consider a moderately sized demographics dataset, that contains questions from  $n = 9409$ questionnaires filled out by shopping mall customers; see Section 14.2.3 of \cite{ESL} for details. The goal is to predict the income using the first 13 questions, listed in the first column of Table~\ref{tab:bump_hunting}, that provide the basic demographics. All variables are either binary or ordinal. We use the empirical distribution of $y_{i}$, the income, as a proxy for $F_{0}$, and use $B=0$ for the location offset. Since the $y_{i}$'s are discrete, the $p$-values are made continuous by randomization; to account for the effect of randomization, we repeat the entire experiment 100 times. 

We find that the box produced by \star ~is quite stable across experiments and the target \textFDR ~level $\alpha$. The results are reported in Table~\ref{tab:bump_hunting}. The last column details the interval for each variable of the most frequent box among 100 repetitions for $\alpha \in \{0.05, 0.1, 0.2\}$. The middle three columns contain the frequency of this particular box among 100 repetitions. Because the box is quite stable for most predictors, we conclude the randomization of the $p$-values does not substantially destabilize the discovered region. We also plot the income distribution of this sub-population and that of the overall population in Figure \ref{fig:income_dist}; thus, we see that our method has detected a subpopulation with significantly higher income than the overall population. Compared to other bump hunting algorithms, \star ~has statistical guarantees (\textFDR ~control). 

\begin{table}
\small
\centering
\begin{tabular}{lllll}
  \hline
Attributes  & box20 freq. & box10 freq. & box5 freq. & box \\ 
  \hline
sex & 1.00 & 1.00 & 1.00 & male/female \\ 
  marital status & 1.00 & 1.00 & 1.00 & married/single \\ 
  age & 0.92 & 0.92 & 0.54 & [18, 54] \\ 
  education & 0.99 & 0.99 & 0.99 & $>=$ high school \\ 
  occupation & 0.92 & 0.75 & 0.53 & professional/manager/student \\ 
  years in bay area & 0.63 & 0.63 & 0.63 & $>$10 \\ 
  dual incomes & 0.92 & 0.92 & 0.92 & not married/yes \\ 
  number in household & 1.00 & 1.00 & 1.00 & [2,4] \\ 
  number of children & 0.59 & 0.59 & 0.59 & $<=$2 \\ 
  householder status & 1.00 & 1.00 & 1.00 & own \\ 
  type of home & 1.00 & 1.00 & 1.00 & house \\ 
  ethnic classification & 1.00 & 1.00 & 1.00 & white \\ 
  language in home & 1.00 & 1.00 & 1.00 & english \\ 
   \hline
\end{tabular}
\caption{Results of bump hunting on the income dataset: the first column reports the variable names; the last column reports the selected interval of each variable in the detected box; the second to the fourth colums report the frequency of the box listed in the last column among 100 randomizations}\label{tab:bump_hunting}
\end{table}

\section{Example 4: hierarchical testing}\label{sec:tree}

\subsection{Problem Setup}
A well-studied case of structured multiple testing is that of hierarchical testing where the hypotheses have an intrinsic rooted tree structure and the non-null hypotheses form a rooted subtree. Most earlier works focus on \textFWER ~control~\citep[e.g.,][]{dmitrienko06, meinshausen08, huque08, brechenmacher11, goeman12}. However \textFWER ~controlling procedures are often quite conservative, having low power. 
In contrast, \cite{yekutieli06,yekutieli08} proposed a novel procedure in microarray analysis that guarantees \textFDR ~control under independence. 
\textFDR ~controlling methods  have since been applied to other areas including genomics~\citep[e.g.,][]{heller09, guo10, benjamini14, li14, lynch16} and neural image analysis~\citep[e.g.,][]{benjamini07, singh10, schildknecht16}. New procedures have also been recently introduced for multi-layer or multi-resolution \textFDR ~guarantees~\citep{barber16,peterson16, katsevich17, bogomolov17}.

Note that in many hierarchical testing problems, the $p$-value for a given node is derived from the $p$-values of the nodes descending from it (using, for example, the Simes test); in such problems, the $p$-values would be dependent and \star ~would not be applicable. However, when it is applicable, it is quite straightforward to be applied to hierarchical testing problems. Similar to Section \ref{sec:convex}, we consider an artificial example to describe the procedure and compare the performance of \star ~with other existing methods. 

\subsection{Procedure}\label{subsec:tree_proc}
 In order to maintain a subtree structure of the rejection set, at any step of the algorithm, we can simply set the candidates to be observed as all leaf nodes of the subtree of still masked $p$-values. Equivalently, \star ~will peel off the leaf nodes that have least favorable scores at each step. 

In this article, we only consider the canonical score. However, it is worth mentioning that there are various reasonable model-assisted scores that can be applied in hierarchical testing. For a certain class of problems such as wavelet-based image-denoising, it is common to assume that the signal strength has an isotonic ordering on the tree under which the signal strength of the parent node is higher than that of the child node. With this prior knowledge, we can combine the \EM ~algorithm and isotonic regression \citep[e.g.,][]{best90, mair09, stout13}, with a tree ordering, to compute the model-assisted score; See Example 2 in Section \ref{subapp:EM} for implementation details.

\subsection{Simulations}

\begin{figure}[h]
  \centering
  \includegraphics[width = \textwidth]{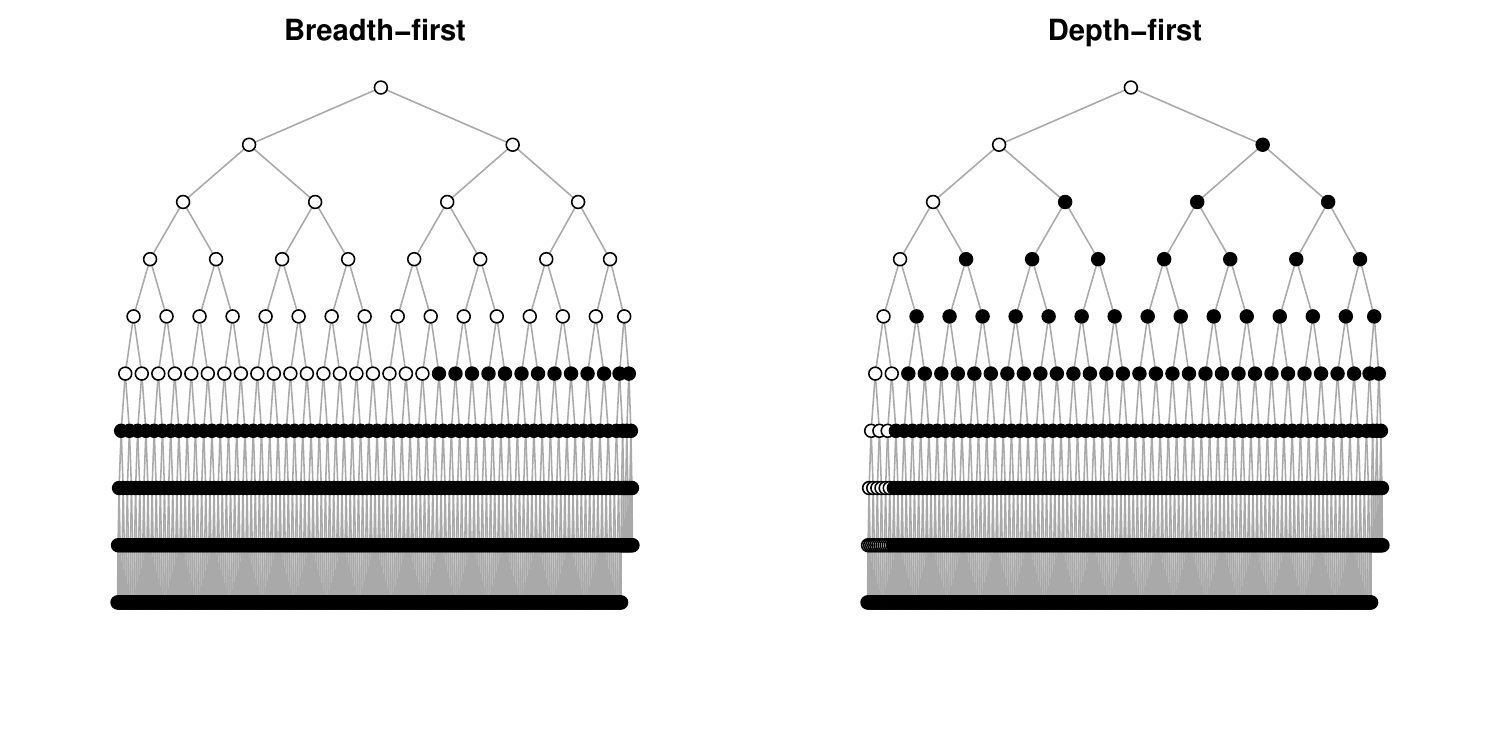}
  \caption{Hierarchical testing problems with non-nulls arranged with breadth-first-search ordering and depth-first-search ordering. Each node represents a hypothesis (1000 in total) with black ones being the nulls and white ones being non-nulls.}\label{fig:tree_expr_truth}
\end{figure}

To illustrate, we construct a balanced binary tree with $n = 1000$ nodes and set 50 nodes as non-nulls. We place the non-nulls as the first 50 nodes either in breath-first-search ordering or in depth-first-search ordering. The settings are plotted in Figure~\ref{fig:tree_expr_truth}. Heuristically, the methods by \cite{yekutieli08} and \cite{lynch16} may prefer the breath-first-search ordering since they are top-down algorithms that proceed layer-by-layer, and only proceed to child nodes when the parent node is rejected. When the non-nulls are placed in the DFS ordering, as shown in the right panel of Figure \ref{fig:tree_expr_truth}, those methods run the risk of stopping early in the long chain of $p$-values, and may therefore be less powerful. By contrast, \star ~proceeds adaptively in a bottom-up manner from leaves to the root, and we may expect it to be more robust to the layout of non-nulls. 

Another important factor that affects the power is the pattern of signal strength along the tree. The top-down procedures should be more favorable if the signal strength is in an isotonic ordering on the tree where the root node has the strongest signal. However, when the signals in top nodes are weak, these procedures risk being powerless. To account for this effect, we generate $p$-values by 
\begin{equation}\label{eq:simul_pvals}
p_{i} = 1 - \Phi(z_{i}), \quad \mbox{and} \quad z_{i}\sim N(\mu_{i}, 1).
\end{equation}
where the null $\mu_{i}$'s  equal  0 and the non-null $\mu_{i}$'s are set in one of the following three ways:
\begin{enumerate}[{Case} 1:]
\item $\mu_{i}\equiv 2, \,\, \forall i\in \cH_{0}^{c}$;
\item  $\displaystyle \mu_{i} = \left\{\begin{array}{ll}
2.5 & i\in \{25 \mbox{ nodes with smallest indices in }\cH_{0}^{c}\},\\
1.5 & \mbox{otherwise;}
\end{array}\right.$
\item $\displaystyle \mu_{i} = \left\{\begin{array}{ll}
1.5 & i\in \{25 \mbox{ nodes with smallest indices in }\cH_{0}^{c}\},\\
2.5 & \mbox{otherwise.}
\end{array}\right.$
\end{enumerate}

In summary, we consider six cases: the non-nulls are placed in breath-first-search ordering or depth-first-search ordering and the $p$-values are set in one of the above three cases. For each setting, we apply \star, with $h(p) = 2I(p\ge 0.5)$ and canonical scores, as well as \cite{yekutieli08}'s procedure and \cite{lynch16}'s procedures in their sections 4.1 and 4.3. We plot the results in Figures~\ref{fig:rej_tree_BFS_methods} and~\ref{fig:rej_tree_DFS_methods}. 

\begin{figure}[h!]
  \centering
  \includegraphics[width = 0.85\textwidth]{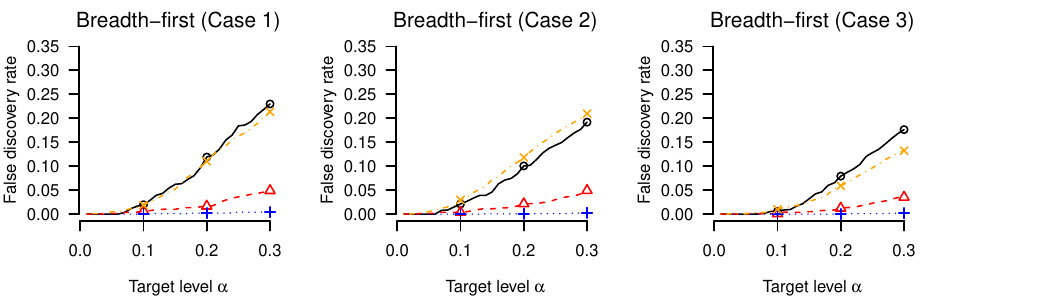}
  \includegraphics[width = 0.85\textwidth]{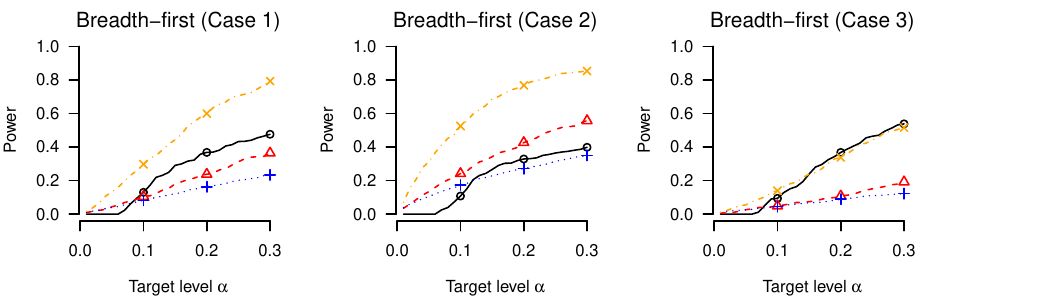}
  \caption{Comparison of \star ~with canonical score (black solid line), \cite{yekutieli08}'s procedure (red dashed line) and \cite{lynch16}'s procedures in their sections 4.1 (blue dotted line) and 4.3 (yellow dot-dashed line). The non-nulls are arranged in the breadth-first-search ordering.}\label{fig:rej_tree_BFS_methods}
\end{figure}

From Figure~\ref{fig:rej_tree_BFS_methods}, we see that all methods control \textFDR ~exactly. In cases 1 and 2, \star ~has lower power than \cite{lynch16}'s second procedure, but is competitive with other procedures. When the top non-nulls are weak as in Case 3, the forward procedures lose power remarkably while \star ~gains power as expected. It is clearly shown that the power of our method is quite stable across the different layouts as opposed to top-down procedures.

\begin{figure}[h!]
  \centering
  \includegraphics[width = 0.85\textwidth]{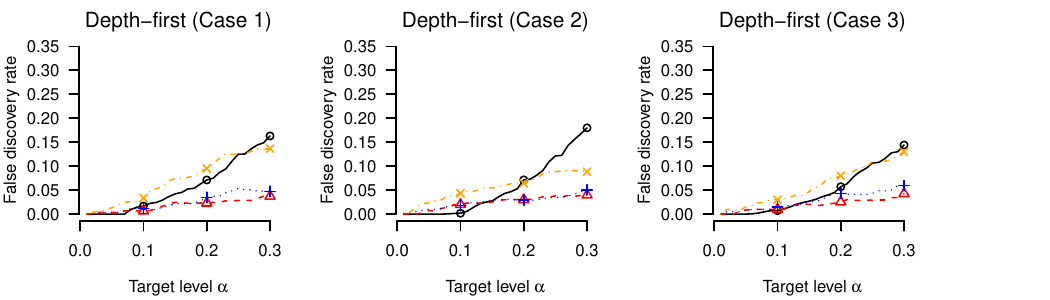}
  \includegraphics[width = 0.85\textwidth]{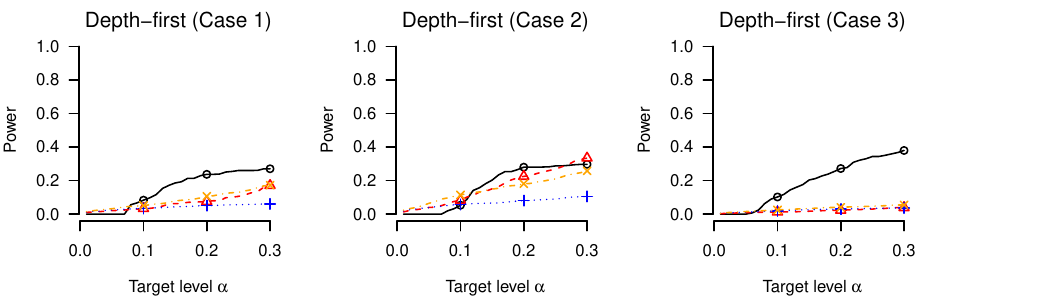}
  \caption{Comparison of \star ~with canonical score (black solid line), \cite{yekutieli08}'s procedure (red dashed line) and \cite{lynch16}'s procedures in their sections 4.1 (blue dotted line) and 4.3 (yellow dot-dashed line). The non-nulls are arranged in the depth-first-search ordering.}\label{fig:rej_tree_DFS_methods}
\end{figure}

From Figure~\ref{fig:rej_tree_DFS_methods}, we see that \star ~is most powerful even when the non-nulls are placed in a DFS ordering. Comparing to Figure~\ref{fig:rej_tree_BFS_methods}, the performance of our method does not degrade much. However, the top-down procedures lose power considerably and even become powerless in Case 1 and Case 3. 

\section{Example 5: wavelet thresholding}\label{sec:wavelet}

\begin{figure}[h]
  \centering
  \includegraphics[width = 0.65\textwidth]{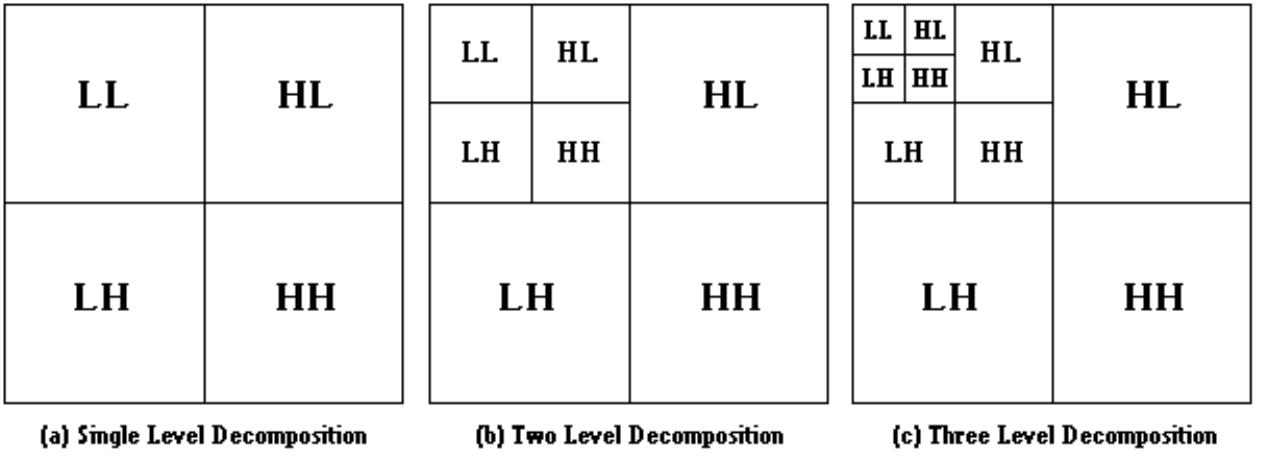}
  \includegraphics[width = 0.3\textwidth]{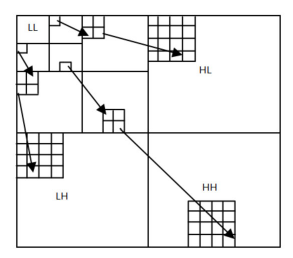}
  \caption{Schematic representation of the 2-dimensional discrete wavelet decomposition: panels (a) - (c) gives the schematic representation of first three levels of decomposition and the rightmost panel gives the description of the hierarchical structure. The figure is copied from \texttt{http://www.debugmode.com/imagecmp/classify.htm}}\label{fig:pyramid}
\end{figure}

Wavelet decomposition has been an efficient tool in signal processing for decades \citep[e.g.,][and references therein]{mallat99}.
 It provides an efficient and elegant methodology that represents signals at different scales, ranging from "backgrounds/trends'' to "edges/anomalies".

Due to the hierarchical nature of wavelet decomposition, the wavelet coefficients can be described by a balanced tree. Figure~\ref{fig:pyramid} gives a schematic description of the 2-dimensional discrete wavelet decomposition which is widely used in image processing. Given an image with size $2^{k}\times 2^{k}$, a high-pass filter and a low-pass filter are applied to the rows and columns to decompose the image into four sub-bands, LL, LH, HL and HH, where LL contains all information in lower frequencies and the last three contain the high-frequency information in different orientations. The procedure then proceeds recursively on LL to decompose the low-frequency sub-band as illustrated in panels (a) - (c) of Figure \ref{fig:pyramid} until LL only contains one pixel. The wavelet coefficients can be arranged in a quadtree; See the rightmost panel of Figure \ref{fig:pyramid} for illustration. We refer the readers to \cite{mallat99} for details.

Often, natural signals can be represented by a small subset of wavelet coefficients; equivalently, the wavelet coefficient vector is sparse. Under the standard assumption that the signal is multivariate normal and homoscedastic, the wavelet coefficients are independent normal variables since the transformation is unitary. Denote by $\hat{d}_{jk}$ the $j$-th wavelet coefficient in the $k$-th level, then $\hat{d}_{jk}\sim N(\mu_{jk}, \sigma^{2})$ for some common variance $\sigma^{2} > 0$. The problem of detecting ``large'' wavelet coefficients can be formalized as a selection problem that aims at detecting nonzero $\mu_{jk}$'s. The classic procedures, such as hard thresholding \citep{donoho94} and soft thresholding \citep{donoho95}, are proved to be minimax optimal from the estimation viewpoint. However, for most images it is reasonable to assume that the large coefficients form a subtree \citep[e.g.,][]{shapiro93, hegde15}. This tree structure has been exploited since \cite{shapiro93}'s Embedded Zerotrees of Wavelet transforms algorithm for efficient encoding of images. 

On the other hand, \citep{abramovich96} formalized the problem in terms of multiple hypothesis testing with $H_{jk}: \mu_{jk} = 0$ and applied the \bh ~procedure on the $p$-values calculated as $p_{jk} = 1 - \Phi(\hat{d}_{jk} / \hat{\sigma})$, where $\hat{\sigma}$ is estimated from the coefficients at the finest scale. This idea is exploited further using Bayesian \textFDR ~control methods \citep[e.g.,][]{tadesse05, lavrik08}.

However, these methods also do not take the structured sparsity into consideration. This motivates us to apply \star ~with a tree constraint that is discussed in Section~\ref{subsec:tree_proc}. To illustrate we compare our method with other methods on 48 standard gray-scale images of size $512\times 512$, available at \texttt{http://decsai.ugr.es/cvg/CG/images/base/$X$.gif}, where $X$ is an integer $\in \{1,\dots,48\}$. For each figure we add Gaussian white noise with SNR = 0.5dB, where SNR (signal-to-noise ratio) is defined as $10\log_{10}\lb s_{\mathrm{image}}^{2} / s_{\mathrm{noise}}^{2}\rb$ with unit dB (decibel). The two panels in the left column of Figure \ref{fig:image_cat} show one original image and its contaminated version. 

\begin{figure}[h]
  \centering
  \includegraphics[width = 0.8\textwidth]{{{8_la8_SNR_0.5_STAR_0.2_BH_0.05}}}
  \caption{A sample image of cat. The left column show the original image and the contaminated version. The other four panels display the recovered images of \star, the \bh ~procedure, hard thresholding (H.T.) and soft thresholding (S.T.), respectively. The signal-to-noise ratio and the compression ratio are reported in the title.}\label{fig:image_cat}
\end{figure}

We compare \star ~with the \bh ~procedure \citep{abramovich96}, hard thresholding \citep{donoho94} and soft thresholding \citep{donoho95}. We estimate the variance $\hat{\sigma}^{2}$ separately for LH, HL and HH sub-bands using the normalized median of the coefficients at the finest scale; See Chapter 11 of \cite{mallat99} for details. For \star ~and the \bh ~procedure, we calculate $p$-values by $p_{jk} = 1 - \Phi(\hat{d}_{jk} / \hat{\sigma}_{w})$ where $w\in \{LH, HL, HH\}$ depending on the location of $\hat{d}_{jk}$; for hard/soft thresholding, the threshold is chosen as $\sqrt{2\hat{\sigma}_{w}^{2}\log N}$ where $N$ denotes the total number of coefficients. For each method, we record the signal-to-noise ratio (SNR) and compression ratio (CR), defined as the ratio of the total number of wavelet coefficients and number of selected coefficients. To illustrate, we report the SNRs and CRs on the top of four panels in Figure \ref{fig:image_cat}. We observe that \star ~has the largest SNR and a more compact representation than the \bh ~procedure. Note that although the \bh ~procedure produces a visually clearer image, the compression ratio is much smaller than other algorithms. For thorough comparison, we compute the ratio of SNR and CR between \star ~and other methods and provide the boxplot in Figure \ref{fig:SNR_CR}. It is clearly shown that \star ~has larger SNR than other methods and provide a more parsimonious representation than the \bh ~procedure for most figures. We conclude that \star ~has reasonable performance in wavelet-based image denoising.

\begin{figure}[h]
  \centering
  \includegraphics[width = 0.35\textwidth]{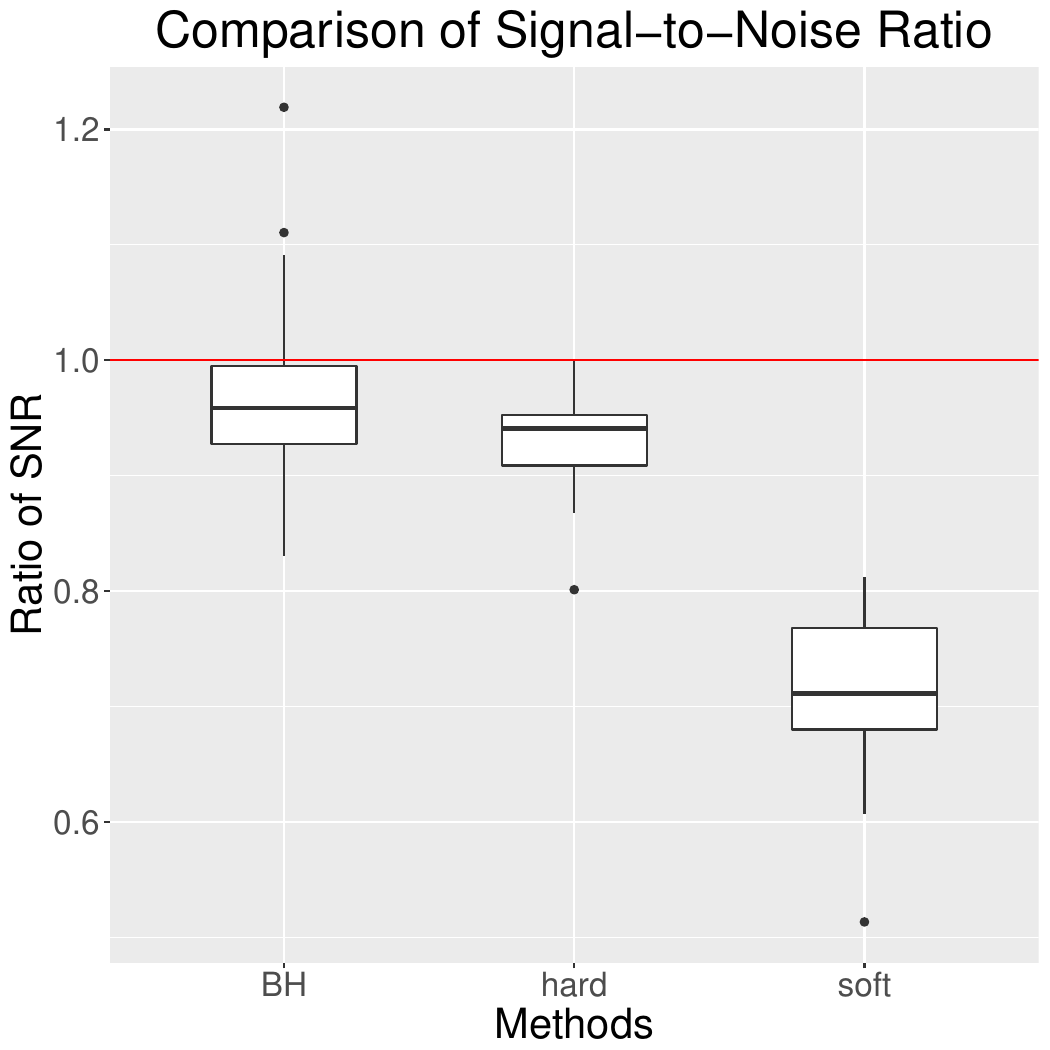}\hspace{1.5cm}
  \includegraphics[width = 0.35\textwidth]{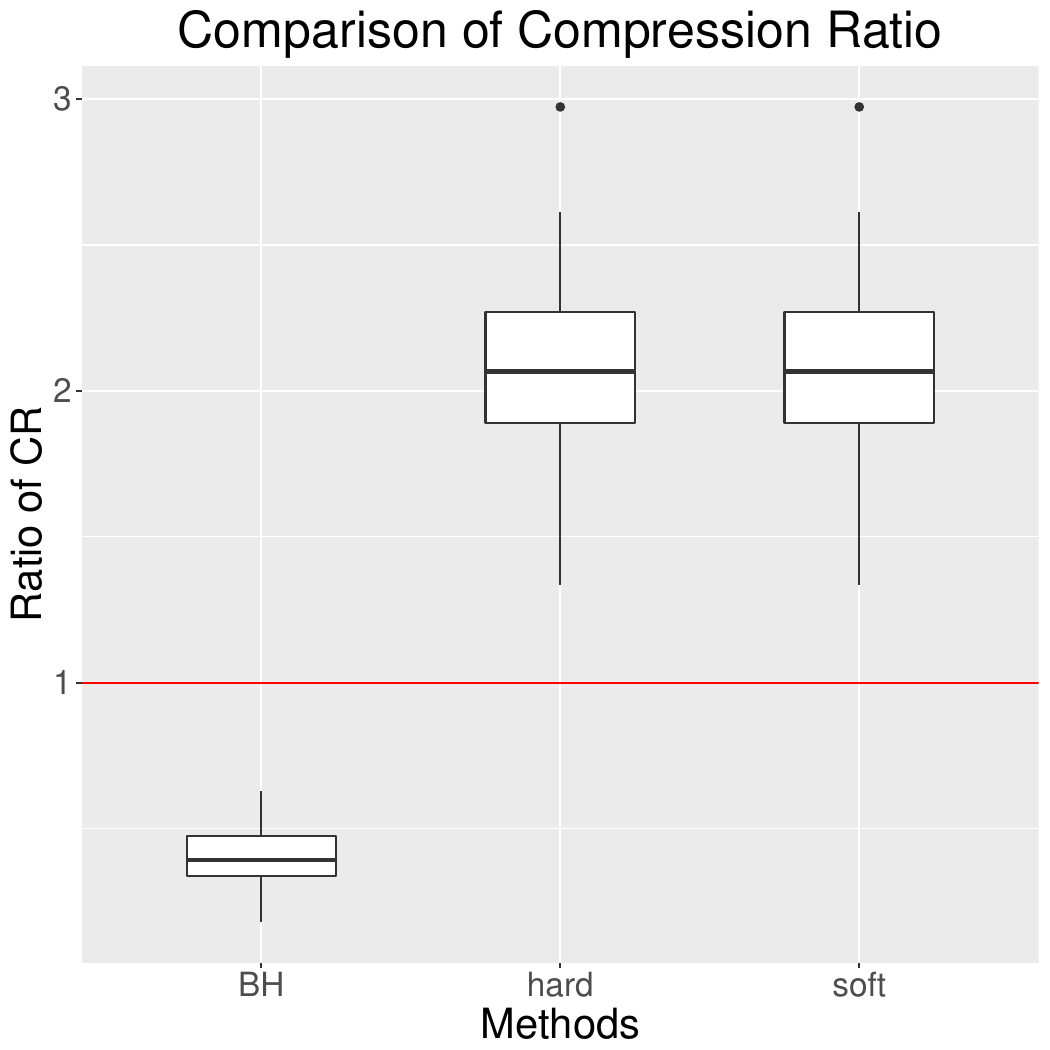}
  \caption{Comparison of SNR and CR between \star ~(normalized to the red line at unity) and other methods (the \bh ~procedure, hard thresholding, soft thresholding) on 48 gray-scale images.}\label{fig:SNR_CR}
\end{figure}

\section{Example 6: interaction selection in factorial experiments}\label{sec:factorial}
The heredity principle dates back to early work of \cite{yates37} on factorial experiments. The term ``heredity'' was coined by \cite{hamada92} in the context of experimental design and originally used to ensure the compatibility of the selected model in the presence of complex aliasing. On the other hand, \cite{nelder77} introduced the marginality principle, an equivalent version of the strong heredity principle, driven by interpretability. In recent years, this topic has been revisited under the high dimensional settings \citep[e.g.,][]{yuan09, choi10, bien13}. However, none of these works provides error measures, either \textFWER ~or \textFDR, regarding the selected variables. 

All aforementioned works consider the linear model with all main effects and second-order interaction effects:
\[
y = \beta_{0} + \sum_{j=1}^{p}\beta_{j}X_{j} + \sum_{j,k=1}^{p}\beta_{jk}X_{j}X_{k} + \eps,
\]
where $y\in \R^{n}$ is the response variable and $(X_{1}, \ldots, X_{p})\in \R^{n\times p}$ are $p$ factors. This induces a two-layer \DAG ~with the main effects $\{X_{j}\}$ in the first layer and the interaction effects $\{X_{j}X_{k}\}$ in the second layer. Write $Z$ for the design matrix including the intercept term, all main effects and interaction effects, i.e. $Z = (\textbf{1},\, X_{1}, X_{2}, \ldots, X_{p}, X_{1}X_{2}, \ldots, X_{p-1}X_{p})\in \R^{n\times \lb \frac{p^2 + p + 2}{2}\rb}$ and $\beta$ for the coefficients, i.e. $\beta = (\beta_{0}, \beta_{1}, \ldots, \beta_{p}, \beta_{12}, \ldots, \beta_{p-1,p})$. Then the model may be succinctly represented as
\[
y = Z\beta + \eps.
\]
If one can construct independent $p$-values for each entry of $\beta$, \star ~can be applied to guarantee the heredity principle and the \textFDR ~control simultaneously.

For illustration we consider a pharmaceutical dataset from \cite{jaynes13}. Their work aims at investigating the effect of six anti-viral drugs, namely Interferon-alpha (A), Interferon-beta (B), Interferon-gamma (C), Ribavirin (D), Acyclovir (E), and TNF-alpha (F), to Herpes simplex virus type 1 (HSV-1). They applied a $2^{6-1}$ fractional factorial design with $32$ runs and encode all factors by $+1$ and $-1$ (in fact, they have 35 runs with the last 3 runs being the replicated center points to evaluate the lack-of-fit).  The minimal word of the half-fraction design is $ABCDEF$ and hence it has resolution VI; see \cite{wu00} for the terminology and details. In other words, the main effects and the second-order interaction effects are not aliased with each other. This means that we can estimate all main effects and all two-factor interactions assuming that fourth-order and higher interactions are negligible, which is quite a reasonable assumption in practice \citep{jaynes13}. The response variable is set to be the logarithm of the viral infection load. 

Under the standard assumption that $\eps\sim N(0, \sigma^{2}I_{n\times n})$, the least-squares estimator is
\[\hat{\beta}\triangleq (\hat{\beta}_{0}, \hat{\beta}_{1}, \ldots, \hat{\beta}_{p}, \hat{\beta}_{1,2}, \ldots, \hat{\beta}_{p-1,p}) = (Z^{T}Z)^{-1}Z^{T}y\sim N(\beta, \sigma^{2}(Z^{T}Z)^{-1}).\]
Due to the nature of fractional factorial designs, $Z$ is an orthogonal matrix with 
\[Z^{T}Z = K\cdot I_{K\times K}, \quad K = \frac{p^2 + p + 2}{2} = 22.\]
Thus the entries of $\hat{\beta}$ are independent. Here we simply replace $\sigma$ by $\hat{\sigma}$, obtained from the regression residuals, i.e.
\[\hat{\sigma}^{2} = \frac{1}{n - \frac{p^{2} + p + 2}{2}}\|y- Z\hat{\beta}\|^{2} = \frac{1}{10}\|y- Z\hat{\beta}\|^{2}.\]
Then we can construct the $p$-values by
\begin{equation}\label{eq:DAG_pvals}
p_{j} = 1 - \Phi\lb\frac{\hat{\beta}_{j}}{\hat{\sigma}}\rb, \quad p_{jk} = 1 - \Phi\lb\frac{\hat{\beta}_{jk}}{\hat{\sigma}}\rb.
\end{equation}
Note that the constructed $p$-values may have some dependence due to sharing $\hat\sigma$. However as demonstrated experimentally in Section \ref{sec:sensitivity}, \star ~still seems to control the \textFDR ~when correlations are not too large. 

Finally, we apply \star ~on the $p$-values defined in \eqref{eq:DAG_pvals} with $\alpha = 0.2$ using the accumulation function $h(p) = 2I(p\ge 0.5)$ and the canonical scores. The selected variables include all main effects and three interaction effects: $A\times B, A\times D, C\times D$. This model identifies more effects than those in \cite{jaynes13}. To illustrate the performance of the selection procedure, we refit a linear model using these variables and find that all selected variables are marginally significant except $C$. The estimate and the $p$-values for both full model and refitted model are reported in Table~\ref{tab:DAG_pvals}. This suggests that \star ~may have successfully identified the important effects with the guarantee that the \textFDR ~is controlled at level $0.2$. 

\begin{table}[h]
  \centering
  \begin{tabular}{c|ccccccccc}
    \toprule
    & $A$ & $B$ & $C$ & $D$ & $E$ & $F$ & $A\times B$ & $A\times D$ & $C\times D$\\
    \midrule
  Estimate  & 0.04 & 0.07 & 0.02 & -0.32 & 0.11 & 0.05 & -0.05 & 0.04 & 0.05\\
  Orig. $p$-value  & 0.118 & 0.012 & 0.458 & 0.000 & 0.001 & 0.037 & 0.053 & 0.086 & 0.038\\
  Refit. $p$-value  & 0.060 & 0.002 & 0.381 & 0.000 & 0.000 & 0.011 & 0.018 & 0.038 & 0.011\\
    \bottomrule
  \end{tabular}
  \caption{Regression results for the pharmaceutical data using variables selected by \star.}\label{tab:DAG_pvals}
\end{table}

\section{More experimental results}\label{sec:more}
In this section, we provide more simulation results on the comparison of \star ~with SeqStep accumulation function with different cutoffs and the comparison of \star ~with different accumulation functions. The settings are the same as their counterpart in Sections \ref{sec:convex}, \ref{sec:tree} and \ref{sec:DAG}.

\subsection{Convex region detection}\label{app:experiment_convex}
\begin{figure}[h]
  \centering
  \includegraphics[width= 0.8\textwidth]{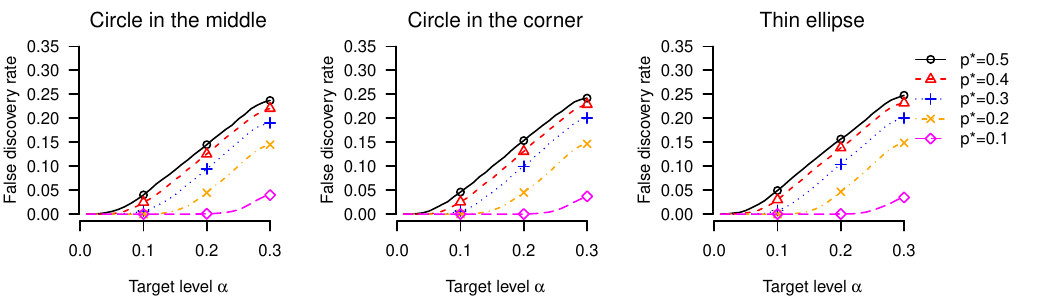}
  \includegraphics[width= 0.8\textwidth]{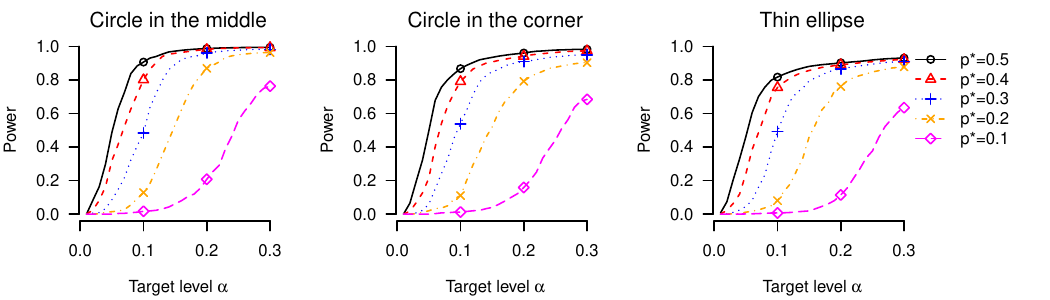}
  \caption{Comparison of \star ~with the SeqStep accumulation function with $\pth\in \{0.5, 0.4, 0.3, 0.2, 0.1\}$. }\label{fig:rej_convex_STARSS}
\end{figure}

\begin{figure}[h]
  \centering
  \includegraphics[width= 0.8\textwidth]{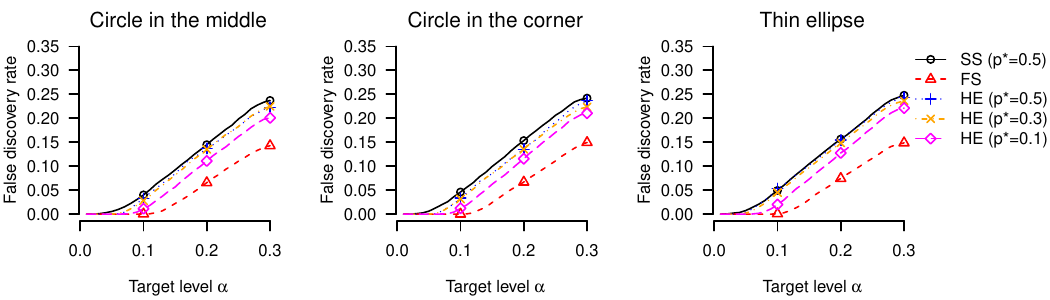}
  \includegraphics[width= 0.8\textwidth]{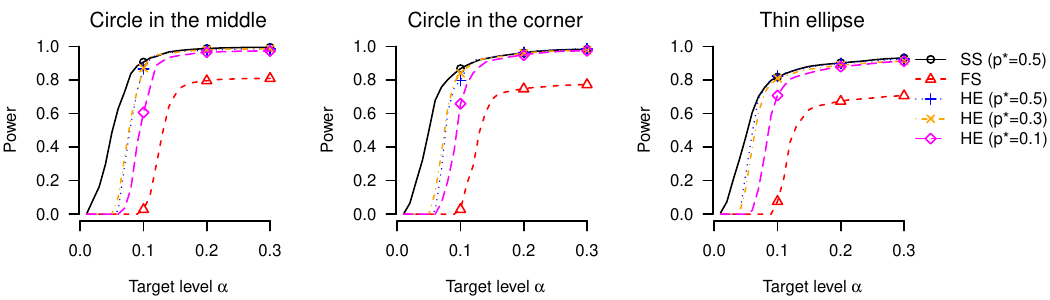}
  \caption{Comparison of \star ~with the SeqStep accumulation function (written as SS) with $\pth = 0.5$, that with the ForwardStop accumulation function (written as FS), and that with the HingeExp accumulation function (written as HE) with $\pth\in \{0.5, 0.3, 0.1\}$.}\label{fig:rej_convex_h}
\end{figure}

\newpage
\subsection{Hierarchical testing}\label{app:experiment_tree}
\begin{figure}[h]
  \centering
  \includegraphics[width= 0.8\textwidth]{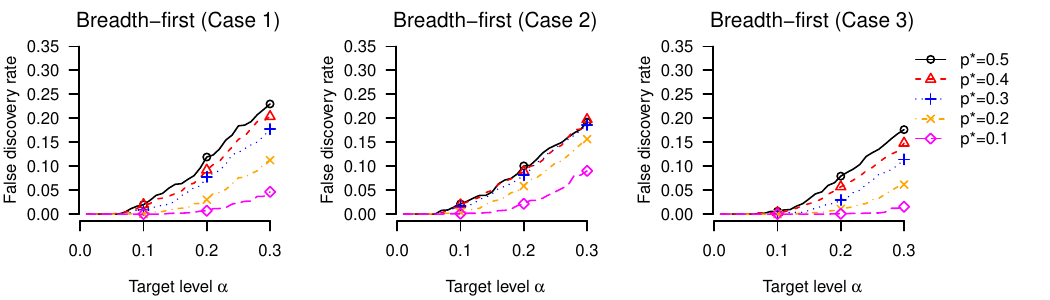}
  \includegraphics[width= 0.8\textwidth]{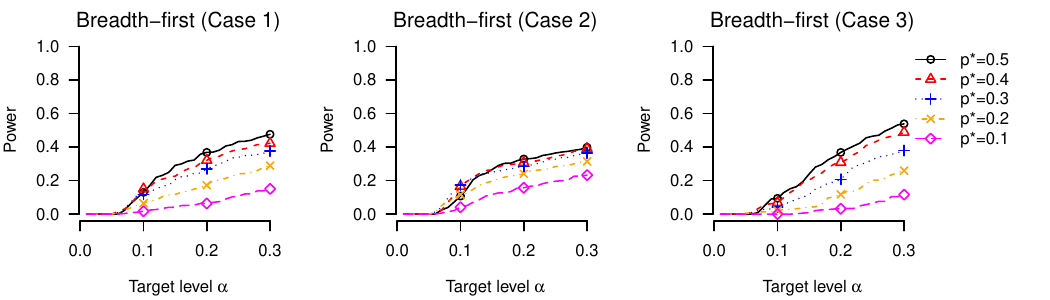}
  \caption{Comparison of \star ~with the SeqStep accumulation function with $\pth\in \{0.5, 0.4, 0.3, 0.2, 0.1\}$. The non-nulls placed in the breadth-first-search ordering. }\label{fig:rej_tree_BFS_STARSS}
\end{figure}

\begin{figure}[h]
  \centering
  \includegraphics[width= 0.8\textwidth]{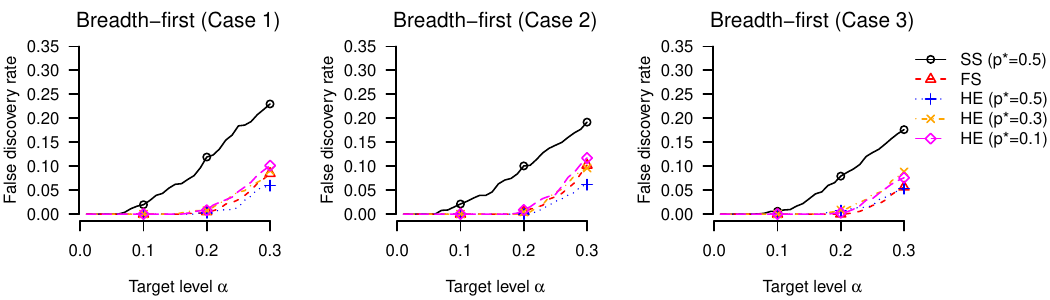}
  \includegraphics[width= 0.8\textwidth]{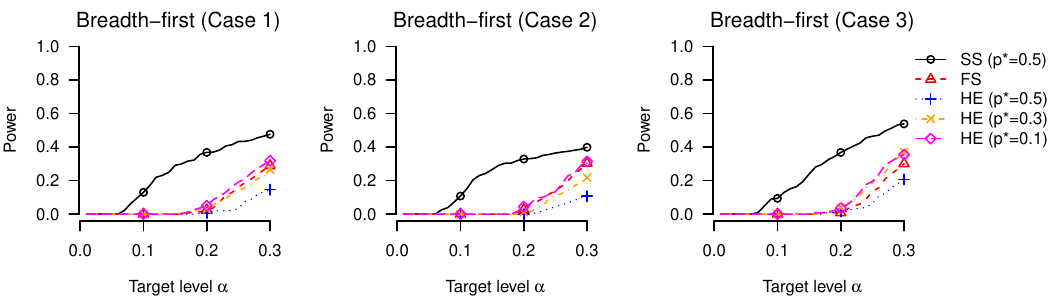}
  \caption{Comparison of \star ~with the SeqStep accumulation function (written as SS) with $\pth = 0.5$, that with the ForwardStop accumulation function (written as FS), and that with the HingeExp accumulation function (written as HE) with $\pth\in \{0.5, 0.3, 0.1\}$. The non-nulls placed in the breadth-first-search ordering. }\label{fig:rej_tree_h}
\end{figure}

\begin{figure}[h]
  \centering
  \includegraphics[width= 0.8\textwidth]{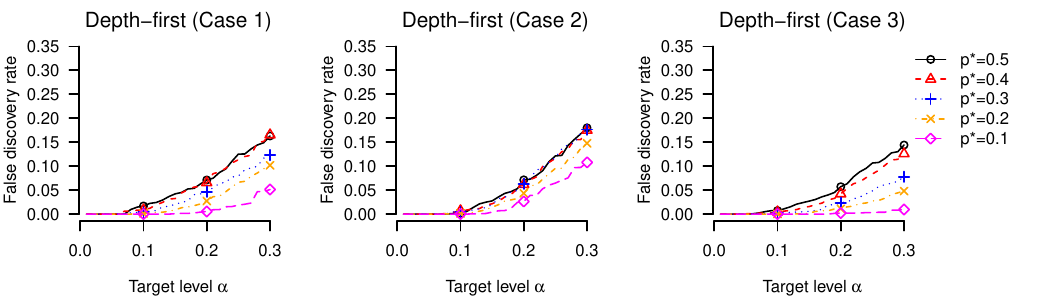}
  \includegraphics[width= 0.8\textwidth]{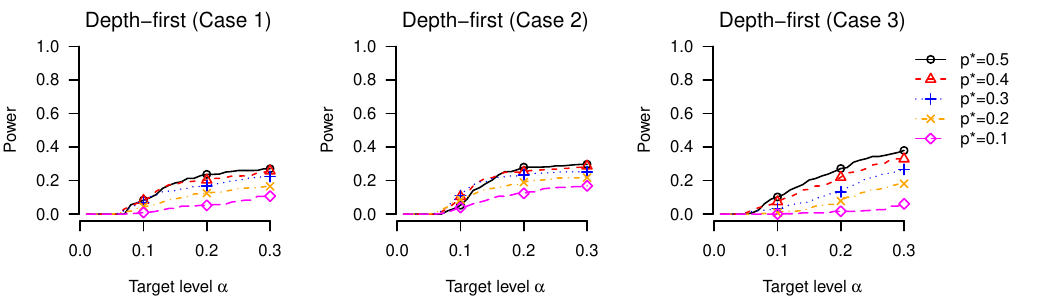}
  \caption{Comparison of \star ~with the SeqStep accumulation function with $\pth\in \{0.5, 0.4, 0.3, 0.2, 0.1\}$. The non-nulls placed in the depth-first-search ordering. }\label{fig:rej_tree_depth-first-search_STARSS}
\end{figure}

\begin{figure}[h]
  \centering
  \includegraphics[width= 0.8\textwidth]{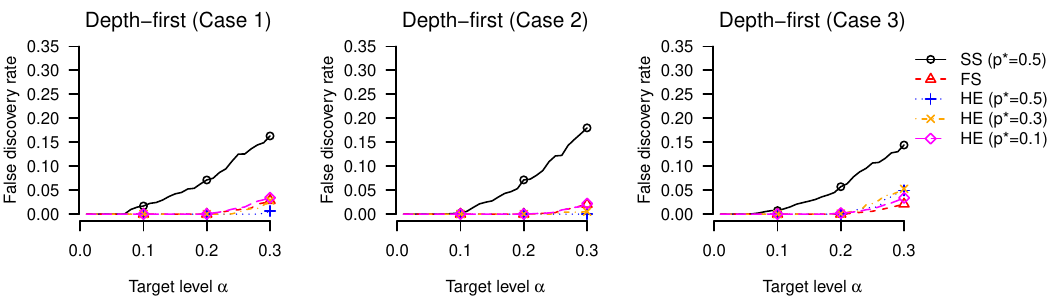}
  \includegraphics[width= 0.8\textwidth]{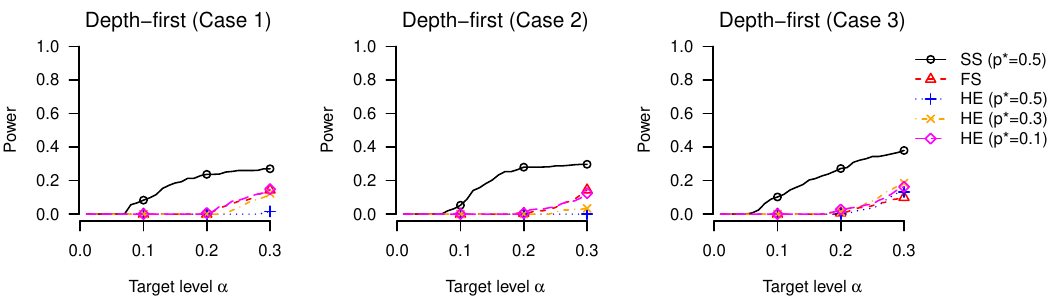}
  \caption{Comparison of \star ~with the SeqStep accumulation function (written as SS) with $\pth = 0.5$, that with the ForwardStop accumulation function (written as FS), and that with the HingeExp accumulation function (written as HE) with $\pth\in \{0.5, 0.3, 0.1\}$. The non-nulls placed in the depth-first-search ordering. }\label{fig:rej_tree_h}
\end{figure}

\newpage
\subsection{Selection under heredity principle}\label{app:experiment_DAG}

\begin{figure}[h]
  \centering
  \includegraphics[width= 0.8\textwidth]{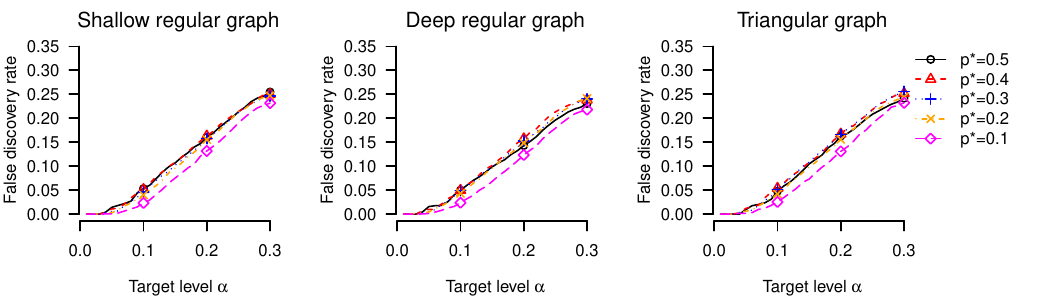}
  \includegraphics[width= 0.8\textwidth]{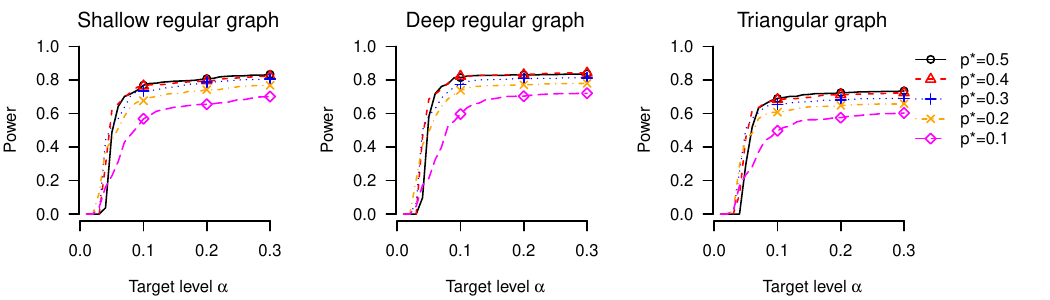}
  \caption{Comparison of \star ~with the SeqStep accumulation function with $\pth\in \{0.5, 0.4, 0.3, 0.2, 0.1\}$. }\label{fig:rej_DAG_STARSS}
\end{figure}

\begin{figure}[h]
  \centering
  \includegraphics[width= 0.8\textwidth]{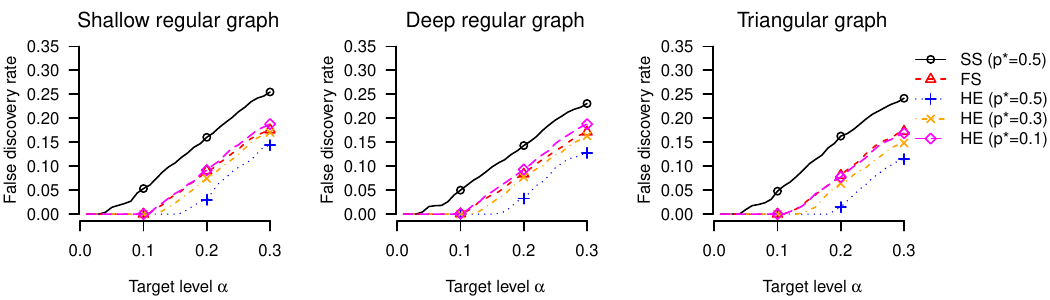}
  \includegraphics[width= 0.8\textwidth]{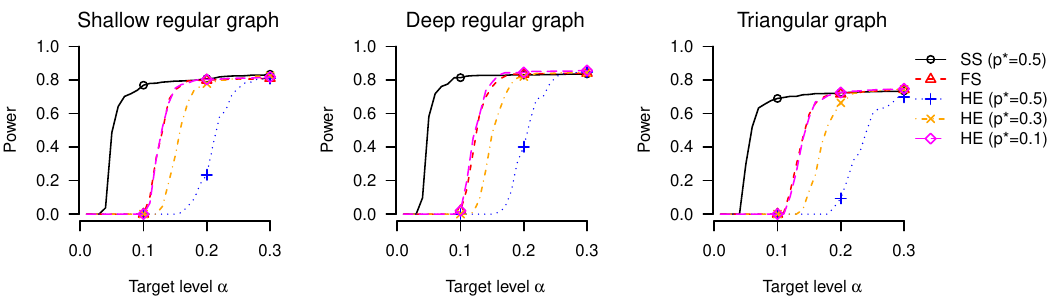}
  \caption{Comparison of \star ~with the SeqStep accumulation function (written as SS) with $\pth = 0.5$, that with the ForwardStop accumulation function (written as FS), and that with the HingeExp accumulation function (written as HE) with $\pth\in \{0.5, 0.3, 0.1\}$.}\label{fig:rej_DAG_h}
\end{figure}

\newpage

\section{Sensitivity analysis}\label{sec:sensitivity}

In this Section, we examine the performance of \star ~in the presense of correlated $p$-values in all three cases considered in Section \ref{sec:convex}, \ref{sec:tree} and \ref{sec:DAG}. We generate $p$-values from one-sided normal test with 
\[p_{j} = 1 - \Phi(z_{i}), \quad \mbox{where }z = (z_{1}, \ldots, z_{n})\sim N(\mu, \Sigma).\]
where $\mu = (\mu_{1}, \ldots, \mu_{n})$ is set to the same as in each section. Instead of letting $\Sigma = I_{n\times n}$ in the main text, we set $\Sigma$ as an equi-correlated matrix, i.e.
\[\Sigma = \lb
  \begin{array}{cccc}
    1 & \rho & \cdots & \rho\\
    \rho & 1 & \cdots & \rho\\
    \vdots & \vdots & \vdots & \vdots\\
    \rho & \rho & \cdots & 1
  \end{array}
\rb.\]
In the following analysis, we consider both the positive correlated case where $\rho = 0.5$ and the negative correlated case where $\rho = -0.5 / n$; in the latter case, we set the coefficient proportional to $1/n$ in order to guarantee that $\Sigma$ is positive semi-definite. 

It turns out that in all cases, the \textFDR ~is still controlled at the target level and the power remains high compared to other competitors. The results are plotted in the following subsections. Therefore, we conclude that \star ~is not sensitive to the correlation of $p$-values and can be used safely when the correlation between the $p$-values is not high.

\subsection{Convex region detection}

\begin{figure}[h]
  \centering
  \includegraphics[width= 0.8\textwidth]{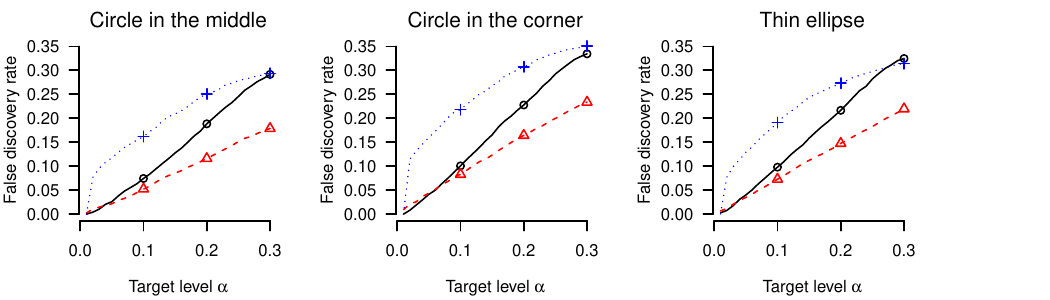}
  \includegraphics[width= 0.8\textwidth]{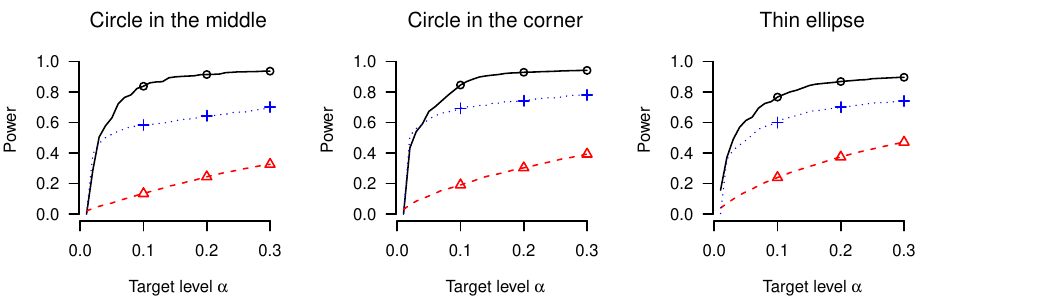}

  \caption{Comparison of \star ~(black solid), the \bh ~procedure (red dashed) and \adapt ~(blue dotted) in the positive correlated case $\rho = 0.5$. This is a counterpart of Figure \ref{fig:rej_convex_methods}.}\label{fig:rej_convex_methods_5}
\end{figure}

\begin{figure}[h]
  \centering
  \includegraphics[width= 0.8\textwidth]{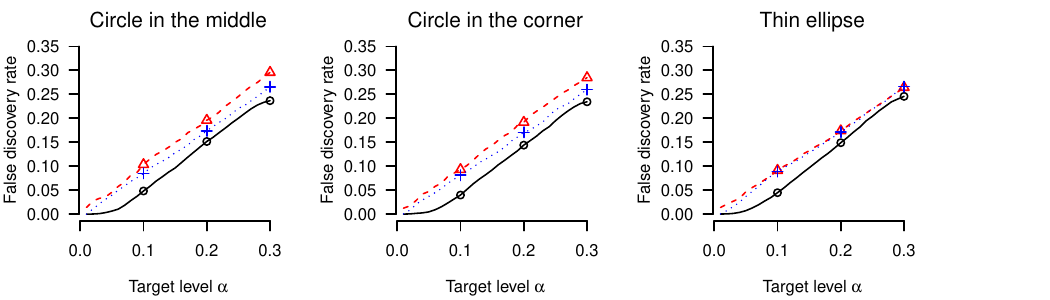}
  \includegraphics[width= 0.8\textwidth]{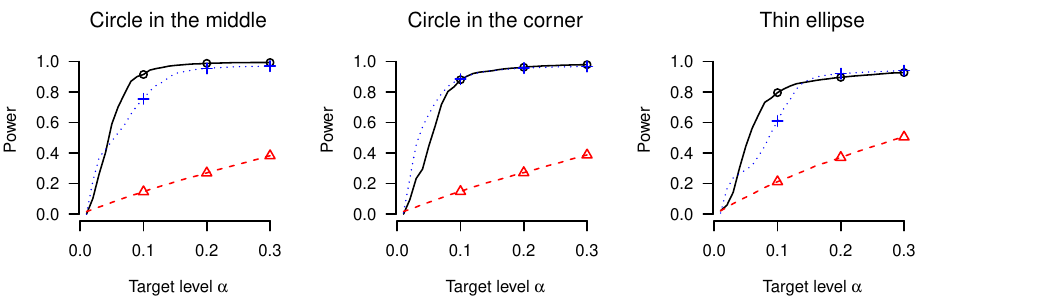}
  \caption{Comparison of \star ~(black solid), the \bh ~procedure (red dashed) and \adapt ~(blue dotted) in the negative correlated case $\rho = -0.5/n$. This is a counterpart of Figure \ref{fig:rej_convex_methods} in Section \ref{sec:convex}.}\label{fig:rej_convex_methods_-5}
\end{figure}

\newpage
\subsection{Hierarchical testing}
\begin{figure}[h]
  \centering
  \includegraphics[width= 0.8\textwidth]{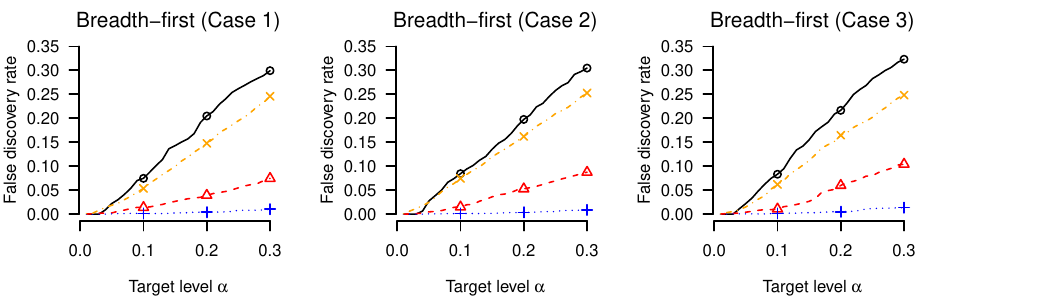}
  \includegraphics[width= 0.8\textwidth]{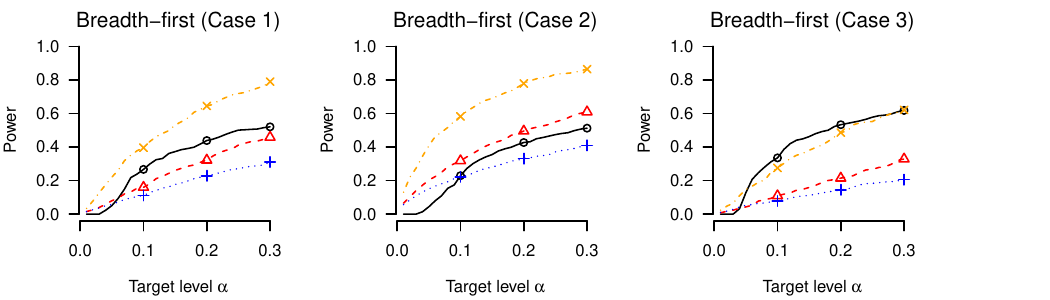}
  \caption{Comparison of \star ~(black solid), \cite{yekutieli08}'s procedure (red dashed) and \cite{lynch16}'s procedures in  sections 4.1 (blue dotted) and 4.3 (yellow dot dashed) in the positive correlated case $\rho = 0.5$. The non-nulls are arranged in the breadth-first ordering and this is a counterpart of Figure \ref{fig:rej_tree_BFS_methods} in Section \ref{sec:tree}.}\label{fig:rej_tree_BFS_methods_5}
\end{figure}

\begin{figure}[h]
  \centering
  \includegraphics[width= 0.8\textwidth]{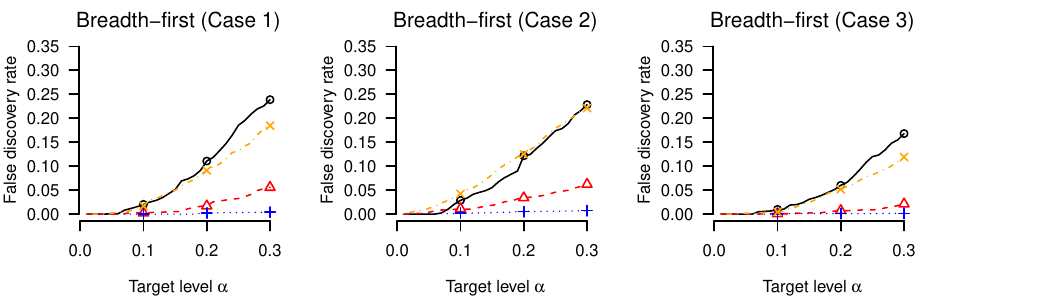}
  \includegraphics[width= 0.8\textwidth]{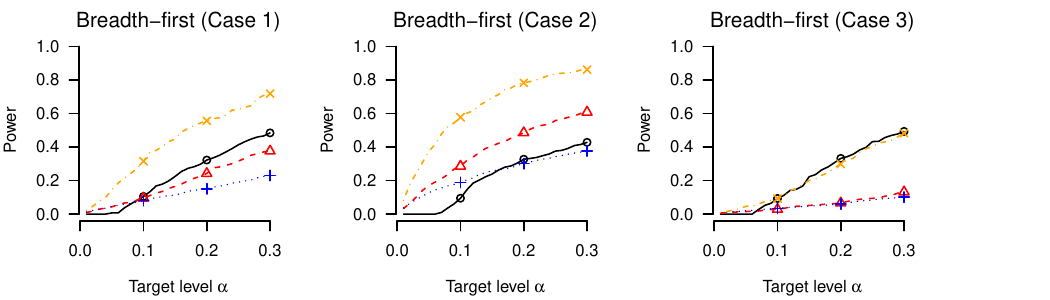}
  \caption{Comparison of \star ~(black solid), \cite{yekutieli08}'s procedure (red dashed) and \cite{lynch16}'s procedures in  sections 4.1 (blue dotted) and 4.3 (yellow dot dashed) in the negative correlated case $\rho = -0.5/n$. The non-nulls are arranged in the breadth-first ordering and this is a counterpart of Figure \ref{fig:rej_tree_BFS_methods} in Section \ref{sec:tree}.}\label{fig:rej_tree_BFS_methods_-5}
\end{figure}

\begin{figure}[h]
  \centering
  \includegraphics[width= 0.8\textwidth]{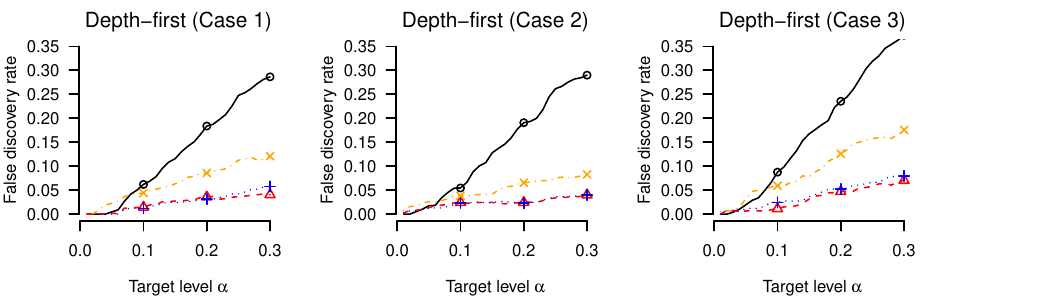}
  \includegraphics[width= 0.8\textwidth]{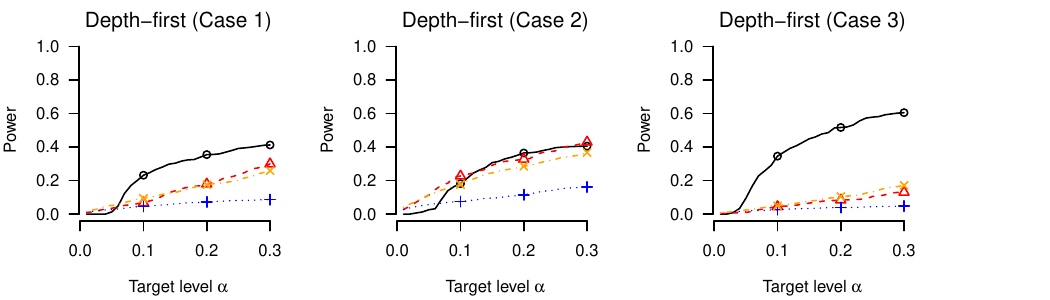}
  \caption{Comparison of \star ~(black solid), \cite{yekutieli08}'s procedure (red dashed) and \cite{lynch16}'s procedures in  sections 4.1 (blue dotted) and 4.3 (yellow dot dashed) in the positive correlated case $\rho = 0.5$. The non-nulls are arranged in the depth-first ordering and this is a counterpart of Figure \ref{fig:rej_tree_DFS_methods} in Section \ref{sec:tree}.}\label{fig:rej_tree_DFS_methods_5}
\end{figure}

\begin{figure}[h]
  \centering
  \includegraphics[width= 0.8\textwidth]{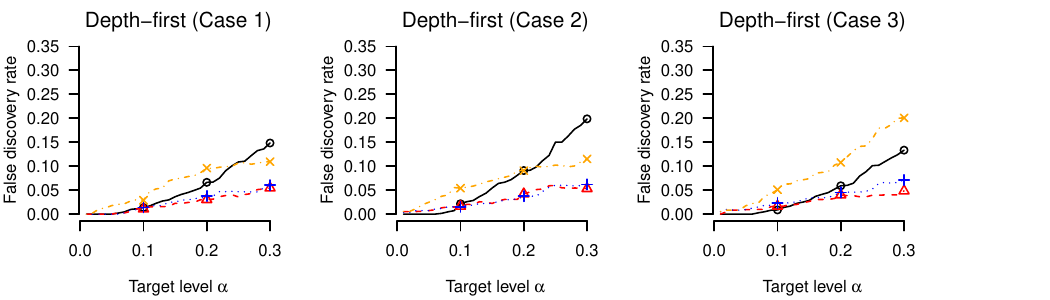}
  \includegraphics[width= 0.8\textwidth]{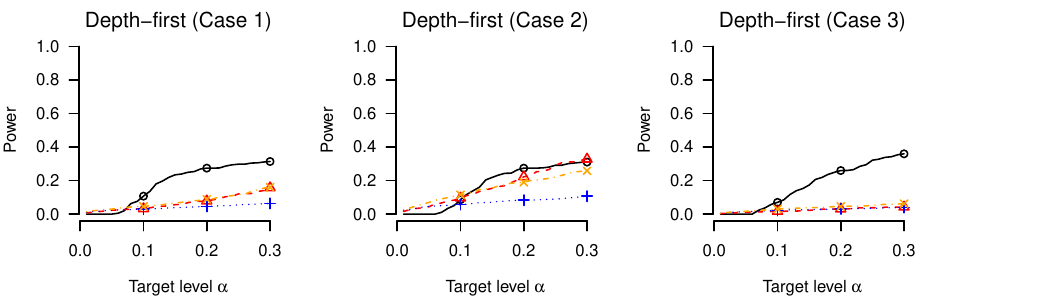}
  \caption{Comparison of \star ~(black solid), \cite{yekutieli08}'s procedure (red dashed) and \cite{lynch16}'s procedures in  sections 4.1 (blue dotted) and 4.3 (yellow dot dashed) in the negative correlated case $\rho = -0.5/n$. The non-nulls are arranged in the depth-first ordering and this is a counterpart of Figure \ref{fig:rej_tree_DFS_methods} in Section \ref{sec:tree}.}\label{fig:rej_tree_DFS_methods_-5}
\end{figure}

\clearpage
\newpage
\subsection{Selection under heredity principle}

\begin{figure}[h]
  \centering
  \includegraphics[width = 0.8\textwidth]{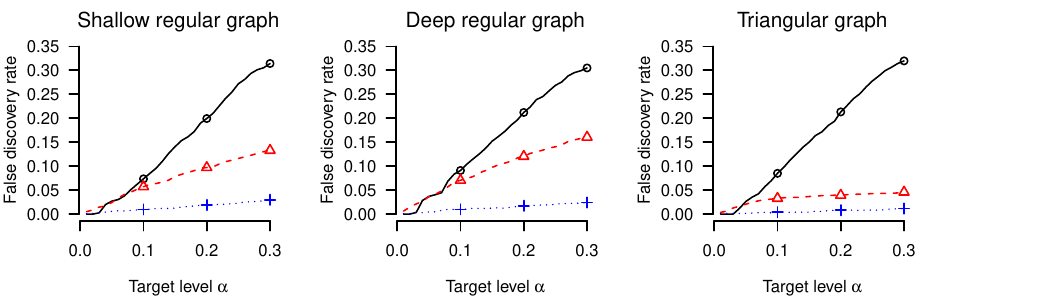}
  \includegraphics[width = 0.8\textwidth]{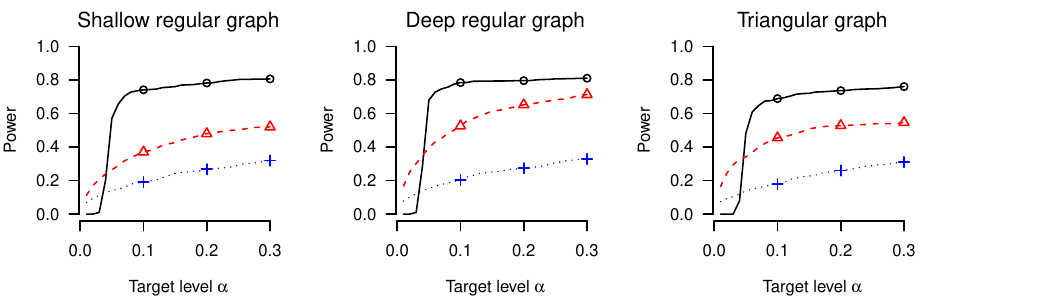}
  \caption{Comparison of \star ~(black solid), \cite{lynch16}'s method (red dashed) and \cite{ramdas2017dagger}'s method (blue dotted) in the positive correlated case $\rho = 0.5$. This is a counterpart of Figure \ref{fig:rej_DAG_methods} in Section \ref{sec:DAG}.}\label{fig:rej_DAG_methods_5}
\end{figure}

\begin{figure}[h]
  \centering
  \includegraphics[width = 0.8\textwidth]{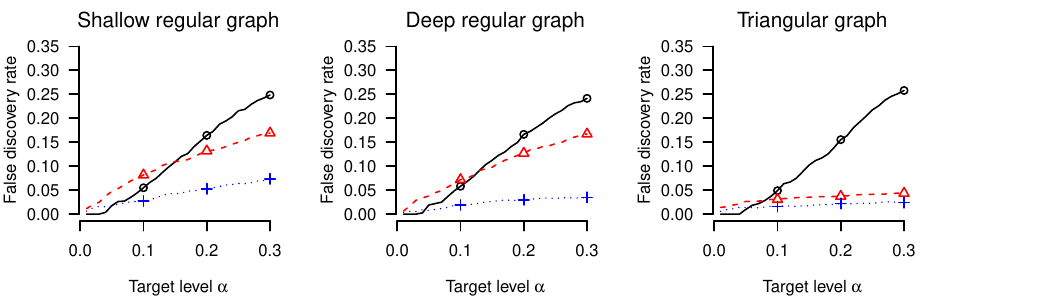}
  \includegraphics[width = 0.8\textwidth]{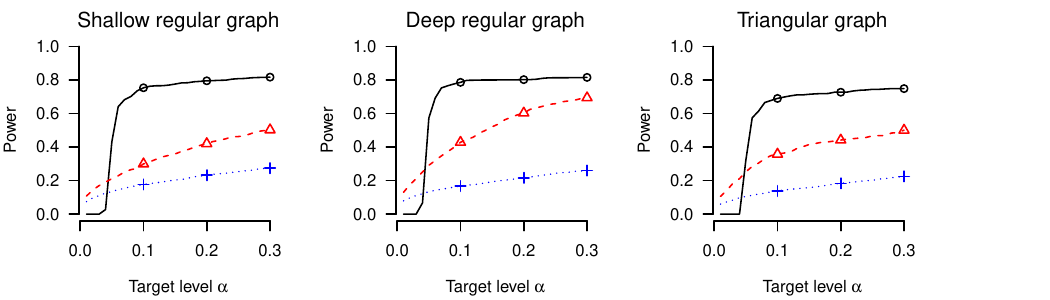}
  \caption{Comparison of \star ~(black solid), \cite{lynch16}'s method (red dashed) and \cite{ramdas2017dagger}'s method (blue dotted) in the negative correlated case $\rho = -0.5/n$. This is a counterpart of Figure \ref{fig:rej_DAG_methods} in Section \ref{sec:DAG}.}\label{fig:rej_DAG_methods_-5}
\end{figure}

\clearpage
\newpage

\section{Benefit of Using Masking Functions}\label{app:asymptotics}
\subsection{Asymptotic false discovery rate and power}
In this section, we investigate the performance of our method under certain asymptotic regimes. The goal of this section is to characterize the benefit of using the masking function through the comparison between our method and the plain accumulation test. For illustration, we focus on the cases without structural constraints, namely $\cK = 2^{[n]}$. Additionally, we restrict the attention into the non-interactive version of our method that computes a score $T_{i}$ for each p-value, described in Section \ref{sec:implementation}, only in the initial stage using $\{x_{i}, g(p_{i})\}_{i=1}^{n}$ and never updates it. This is equivalent to the accumulation test with p-values sorted by $T_{i}$. It is worth emphasizing that even in this basic setting where our method certainly loses many advantageous features, we can still observe the gain of using masking functions. The power analysis for interactive versions under general structural constraints is left to the future research. 

The aforementioned non-interactive version can be equivalently formulated as rejecting all p-values less than or equal to $\hat{t}_{n}$, where
\begin{equation}
  \label{eq:hatt}
  \hat{t}_{n} = \sup\left\{t: \frac{h(1) + \sum_{i=1}^{n}h(p_{i})I(T_{i}\le t)}{1 + \sum_{i=1}^{n}I(T_{i}\le t)}\le \alpha\right\}.
\end{equation}
Since the rejection set only depends on the ordering of $T_{i}$'s, we can assume without loss of generality that $T_{i}\in
 (0, 1]$; otherwise we can transform $T_{i}$ by $\mathrm{arctan}(1 / T_{i})/\pi + 1 / 2$. The accumulation test is a special case with $T_{i} = i / n$. As in previous works \citep{li2016accumulation, lei2016power}, assume that for each $t\in [0, 1]$,
\begin{equation}
  \label{eq:Fn}
  F_{n}(t)\triangleq \frac{1}{n}\sum_{i=1}^{n}I(T_{i} \le t)\stackrel{p}{\rightarrow} F(t)
\end{equation}
and 
\begin{equation}
  \label{eq:Fn}
  H_{n}(t)\triangleq \frac{1}{n}\sum_{i=1}^{n}h(p_{i})I(T_{i} \le t)\stackrel{p}{\rightarrow} H(t),
\end{equation}
for some functions $F(t)$ and $H(t)$, which are not necessarily continuous. Note that both $F(t)$ and $H(t)$ are non-decreasing with $F(0) = H(0) = 0$ and we can assume $F(t) > 0$ for all $t\in (0, 1]$ without loss of generality; otherwise if $t_{0} = \sup\{t: F(t) = 0\}$ we can tranform $T_{i}$ to $(T_{i} - t_{0}) / (1 - t_{0})$. Intuitively, 
\[\frac{h(1) + \sum_{i=1}^{n}h(p_{i})I(T_{i}\le t)}{1 + \sum_{i=1}^{n}I(T_{i}\le t)} = \frac{h(1) + nH_{n}(t)}{1 + nF_{n}(t)}\approx \frac{H(t)}{F(t)}.\]
This motivates us to define $t^{*}$ as
\begin{equation}
  \label{eq:tstar}
  t^{*} = \sup\{t\in [0, 1]: H(t)\le \alpha F(t)\}.
\end{equation}
Note that $t^{*}$ is always well-defined because $H(0) = 0 \le \alpha F(0)$. The following lemma justifies the above heuristic argument that $t^{*}$ is the limit of $\hat{t}_{n}$. The proof is relegated to Section \ref{subapp:power_proof}.
\begin{lemma}\label{lem:thattstar}
  Let $H_{\alpha}(t) = H(t) - \alpha F(t)$. Assume that 
  \begin{equation}
    \label{eq:uniform_convergence}
    \sup_{t\in [0, 1]}|F_{n}(t) - F(t)|\stackrel{p}{\rightarrow} 0, \quad \sup_{t\in [0, 1]}|H_{n}(t) - H(t)|\stackrel{p}{\rightarrow} 0.
  \end{equation}
Then $\hat{t}_{n}\stackrel{p}{\rightarrow} t^{*}$ if either of the following conditions hold:
  \begin{enumerate}
  \item $t^{*} = 0$ and 
    \begin{equation}
      \label{eq:tstar0}
      \inf_{t' \ge t}H_{\alpha}(t') > 0, \quad \mbox{for any }t > t^{*};
    \end{equation}
  \item $t^{*} = 1$ and there exists a sequence $t_{m}\uparrow t^{*}$ such that
    \begin{equation}
      \label{eq:tstar1}
      H_{\alpha}(t_{m}) < 0, \quad \mbox{for any }m;
    \end{equation}
  \item $t^{*}\in (0, 1)$ and both \eqref{eq:tstar0} and \eqref{eq:tstar1} hold.
  \end{enumerate}
\end{lemma}
\begin{remark}
  If $H_{\alpha}(t)$ is continuous, then the condition \eqref{eq:tstar0} can be removed from part 1 and part 3. This is because if there exists $t > t^{*}$ such that
\[\inf_{t' \ge t}H_{\alpha}(t')\le 0,\]
then the continuity of $H_{\alpha}$ implies the existence of $\td{t}\ge t > t^{*}$ with $H_{\alpha}(\td{t})\le 0$. This contradicts the definition of $t^{*}$. 
\end{remark}
Lemma \ref{lem:thattstar} enables us to compute the asymptotic false discovery rate and power. To be precise, the \emph{false discovery proportion} and the \emph{true positive rate} are defined as
\[\FDP_{n} \triangleq \frac{\sum_{i=1}^{n}I(h_{i} = 0, T_{i}\le \hat{t}_{n})}{\sum_{i=1}^{n}I(T_{i}\le \hat{t}_{n})}, \quad \TPR_{n} = \frac{\sum_{i=1}^{n}I(h_{i} = 1, T_{i}\le \hat{t}_{n})}{\sum_{i=1}^{n}I(h_{i} = 1)},\]
where $h_{i} = 1$ iff $H_{i}$ is false. Then by definition the \emph{false discovery rate} and the \emph{power} can be written as 
\[\FDR_{n} = \E [\FDP_{n}], \quad \Pow_{n} = \E [\TPR_{n}].\]
Assume that for all $t \in [0, 1]$, 
\begin{equation}
  \label{eq:Fn01}
  F_{n0}(t)\triangleq \frac{1}{n}\sum_{i=1}^{n}I(h_{i} = 0, T_{i}\le t)\stackrel{p}{\rightarrow} F_{0}(t), \quad F_{n1}(t)\triangleq \frac{1}{n}\sum_{i=1}^{n}I(h_{i} = 1, T_{i}\le t)\stackrel{p}{\rightarrow} F_{1}(t),
\end{equation}
for some functions $F_{0}(t)$ and $F_{1}(t)$. Note that 
\[F_{n}(t) = F_{n0}(t) + F_{n1}(t), \quad F(t) = F_{0}(t) + F_{1}(t).\]
The following theorem establishes the asymptotic false discovery rate and power for procedures in the form of \eqref{eq:hatt}.
\begin{theorem}\label{thm:FDR_power}
  Assume that $F_{n0}(t)$ and $F_{n1}(t)$ are continuous at $t = t^{*}$. Then under the assumptions of Lemma \ref{lem:thattstar}, 
\[\FDR_{n}\rightarrow \frac{F_{0}(t^{*})}{F(t^{*})}, \quad \mbox{if }t^{*} > 0,\]
and
\[\Pow_{n}\rightarrow \frac{F_{1}(t^{*})}{F_{1}(1)},\]
where $t^{*}$ is defined in \eqref{eq:tstar}. Note that the asymptotic power does not require $t^{*} > 0$.
\end{theorem}

As in \cite{li2016accumulation}, we assume there exists a continuous function $f(t)$ such that
\begin{equation}
  \label{eq:ft}
  \sup_{k\in [n]}\bigg|\frac{1}{k}\sum_{i=1}^{k}I(h_{i} = 1) - f\lb\frac{k}{n}\rb\bigg|\rightarrow 0,
\end{equation}
where $h_{i}$'s are treated as fixed and 
\[p_{i}\mid h_{i} = 0\sim \P_{0}, \quad p_{i}\mid h_{i} = 1\sim \P_{1},\]
where $\P_{0}$ has a non-decreasing density. Recall Proposition \ref{prop:density} that 
\[\E_{0} [h(p_{i})\mid g(p_{i})]\le 1\,\, \mbox{almost surely}.\]
Denote by $E_{0}$ (resp. $G_{0}$) and $\E_{1}$ (resp. $G_{1}$) the expectation (resp. distribution) given $h_{i} = 0$ and $h_{i} = 1$ respectively. The following lemma yields the form of $F(t), F_{0}(t), F_{1}(t), H(t)$. 
\begin{lemma}\label{lem:uniform}
  Assume that $p_{i}$'s are independent. Then 
\[F_{0}(t) = \lim_{n\rightarrow \infty}\frac{1}{n}\sum_{i=1}^{n}I(h_{i} = 0)\P_{0}(T_{i}\le t), \quad F_{1}(t) = \lim_{n\rightarrow \infty}\frac{1}{n}\sum_{i=1}^{n}I(h_{i} = 1)\P_{1}(T_{i}\le t)\]
and 
\[F(t) = F_{0}(t) + F_{1}(t), \quad H(t) = F_{0}(t) + \lim_{n\rightarrow \infty}\frac{1}{n}\sum_{i=1}^{n}I(h_{i} = 1)\E_{1}h(p_{i})I(T_{i}\le t).\]
\end{lemma}

\subsection{Re-analysis of accumulation tests}
The accumulation test of \citet{li2016accumulation} corresponds to the choice $T_{i} = i / n$. In this case,
\[F_{1}(t) = \lim_{n\rightarrow\infty}\frac{1}{n}\sum_{i=1}^{\lfloor nt\rfloor}I(h_{i} = 1) = \lim_{n\rightarrow\infty}\frac{\lfloor nt\rfloor}{n}\frac{1}{\lfloor nt\rfloor}\sum_{i=1}^{\lfloor nt\rfloor}I(h_{i} = 1) = tf(t).\]
Similarly,
\[F_{0}(t) = t(1 - f(t)), \quad F(t) = t.\]
Let $\mu = \E_{1}h(p_{i})$,
\[H(t) = F_{0}(t) + \lim_{n\rightarrow\infty}\frac{1}{n}\sum_{i=1}^{\lfloor nt\rfloor}I(h_{i} = 1)\mu = t(1 - f(t)) + \mu tf(t) = t(1 - (1 - \mu)f(t)).\]
By definition \eqref{eq:tstar},
\begin{align}
  t_{\mathrm{AT}}^{*} &= \sup\{t\in [0, 1]: H(t)\le \alpha F(t)\} = \sup\{t\in [0, 1]: t(1 - (1 - \mu)f(t))\le \alpha t\}\nonumber\\
& = \sup\left\{t\in [0, 1]: f(t) \ge \frac{1 - \alpha}{1 - \mu}\right\}\label{eq:tstarAT}.
\end{align}
Note that $t_{\mathrm{AT}}^{*} > 0$ only if $\mu\ge \alpha$ since $f(t)\le 1$ for any $t\in [0,1]$. If $f(t)$ is non-increasing and $tf(t)$ is non-decreasing, as considered in \cite{li2016accumulation}, it is easy to verify the conditions \eqref{eq:tstar0} and \eqref{eq:tstar1} in Lemma \ref{lem:thattstar}. Thus, by Theorem \ref{thm:FDR_power}, if $t_{\mathrm{AT}}^{*} > 0$,
\begin{align}
  \FDR_{n}&\rightarrow \frac{F_{0}(t_{\mathrm{AT}}^{*})}{F(t_{\mathrm{AT}}^{*})} = \alpha \frac{F_{0}(t_{\mathrm{AT}}^{*})}{H(t_{\mathrm{AT}}^{*})} = \alpha \frac{t_{\mathrm{AT}}^{*}(1 - f(t_{\mathrm{AT}}^{*}))}{t_{\mathrm{AT}}^{*}(1 - f(t_{\mathrm{AT}}^{*})) + \mu t_{\mathrm{AT}}^{*}f(t_{\mathrm{AT}}^{*})}\nonumber\\
& = \alpha \frac{1 - f(t_{\mathrm{AT}}^{*})}{1 - (1 - \mu)f(t_{\mathrm{AT}}^{*})} = \frac{\alpha - \mu}{1 - \mu} = \alpha - \frac{1 - \alpha}{1 - \mu}\mu.  \label{eq:FDR_AT}
\end{align}
Thus the conservatism of false discovery rate is $\frac{1 - \alpha}{1 - \mu}\mu$, which is decreasing in $\mu$. Similarly,
\begin{equation}
  \label{eq:power_AT}
  \Pow_{n}\rightarrow \frac{F_{1}(t_{\mathrm{AT}}^{*})}{F_{1}(1)} = \frac{t_{\mathrm{AT}}^{*}f(t_{\mathrm{AT}}^{*})}{f(1)}.
\end{equation}
This recovers Theorem 3 of \cite{li2016accumulation}. Since $tf(t)$ is non-decreasing, the asymptotic power is non-decreasing in $t_{\mathrm{AT}}^{*}$ and is thus non-increasing in $\mu$. Therefore, when $\mu = \E_{1}[h(p_{i})]$ becomes smaller, the conservatism of false discovery rate is reduced while the power is enhanced. By Lemma 2 of \cite{li2016accumulation}, $h(t) = I(p > \pth) / (1 - \pth)$ yields the smallest $\mu$ among all functions bounded by $1 / (1 - \pth)$. 

\subsection{Analysis of our method without informative pre-ordering}
In most applications, an informative pre-ordering is not available. A typical two-group model assumes that $h_{i}\stackrel{i.i.d.}{\sim}\mathrm{Ber}(\pi_{1})$. In this case,
\[f(t)\equiv \pi_{1}.\]
Therefore, as observed by \cite{lei2016power}, the plain accumulation test either has zero power or full power since
\[t_{\mathrm{AT}}^{*} = I\lb\pi_{1}\ge \frac{1 - \alpha}{1 - \mu}\rb.\]
Now we investigate our method with canonical score $T_{i} = g(p_{i})$. Without loss of generality we assume $g([0,1]) \subseteq [0,1]$; otherwise we can compose it with a strictly monotone transformation to map it to the unit interval. It is not hard to see that 
\[F_{1}(t) = \pi_{1}\tau_{g1}(t), \quad F_{0}(t) = (1 - \pi_{1})\tau_{g0}(t), \quad F(t) = F_{0}(t) + F_{1}(t)\]
where
\[\tau_{g1}(t) = \P_{1}(g(p_{i})\le t), \quad \tau_{g0}(t) = \P_{0}(g(p_{i})\le t),\]
and 
\[H(t) = F_{0}(t) + \pi_{1}\mu_{g}(t)\]
where
\[\mu_{g}(t) = \E_{1}h(p_{i})I(g(p_{i})\le t).\]
By definition \eqref{eq:tstar},
\begin{align*}
  t^{*} &= \sup\{t\in [0, 1]: (1 - \pi_{1})\tau_{g0}(t) + \pi_{1}\mu_{g}(t)\le \alpha ((1 - \pi_{1})\tau_{g0}(t) + \pi_{1}\tau_{g1}(t))\}\\
& = \sup\{t\in [0, 1]: (1 - \alpha)(1 - \pi_{1})\tau_{g0}(t)+ \pi_{1}\mu_{g}(t) \le \alpha\pi_{1}\tau_{g1}(t)\}\\
& = \sup\left\{t\in [0, 1]: \pi_{1}\ge \frac{(1 - \alpha)\tau_{g0}(t)}{(1 - \alpha)\tau_{g0}(t) + \alpha\tau_{g1}(t) - \mu_{g}(t)}\right\}
\end{align*}
In the regime where the plain accumulation test has full power, i.e. $\pi_{1}\ge (1 - \alpha) / (1 - \mu)$ and $t_{\mathrm{AT}}^{*} = 1$, it is easy to see that our method has full power as well because $\tau_{g0}(1) = \tau_{g1}(1) = 1$ and $\mu_{g}(1) = \E_{1}h(p_{i}) = \mu$. In the regime where the plain accumulation test has zero power, i.e. $\pi_{1} < (1 - \alpha) / (1 - \mu)$ and $t_{\mathrm{AT}}^{*} = 0$, our method may still have non-zero power if $\tau_{g1}(t) >\!\!> \tau_{g0}(t), \mu_{g}(t)$ for some $t$. For instance, when $h(p) = I(p > \pth) / (1 - \pth)$ and $g(p) = \min\{p, \pth (1 - p) / (1 - \pth)\}$ as described in Subsection \ref{subapp:masking_functions}, 
\[\tau_{g0}(t) = \P_{0}(p_{i}\le t) + \P_{0}(p_{i}\ge 1 - (1 - \pth) t / \pth ) = t + \frac{1 - \pth}{\pth}t = \frac{t}\pth\]
and
\[\tau_{g1}(t) = \P_{1}(p_{i}\le t) + \P_{1}(p_{i}\ge 1 - (1 - \pth) t / \pth ) = G_{1}(t) + 1 - G_{1}\lb 1 - \frac{1 - \pth}{\pth}t\rb.\]
On the other hand, since $g(p)\le \pth$ for all $p\in [0, 1]$, we can assume $t\le \pth$ as well. Then
\begin{align*}
  \mu_{g}(t) &= \E_{1}h(p_{i})I(p_{i}\le t) + \E_{1}h(p_{i})I\lb p_{i}\ge 1 - \frac{1 - \pth}{\pth}t\rb\\
 & = \E_{1}h(p_{i})I\lb p_{i}\ge 1 - \frac{1 - \pth}{\pth}t\rb = \frac{1 - G_{1}\lb 1 - \frac{1 - \pth}{\pth}t\rb}{1 - \pth}.
\end{align*}
Assume that $G_{1}$ has density $g_{1}$. Then 
\begin{equation}
  \label{eq:preorderinglimit}
  \lim_{t\rightarrow 0}\frac{\tau_{g0}(t)}{t} = \frac{1}{\pth}, \quad \lim_{t\rightarrow 0}\frac{\tau_{g1}(t)}{t} = g_{1}(0) + \frac{1 - \pth}{\pth}g_{1}(1), \quad \lim_{t\rightarrow 0}\frac{\mu_{g}(t)}{t} = \frac{g_{1}(1)}{\pth}.
\end{equation}
Therefore, 
\begin{align*}
  \lim_{t\rightarrow 0}\frac{(1 - \alpha)\tau_{g0}(t)}{(1 - \alpha)\tau_{g0}(t) + \alpha\tau_{g1}(t) - \mu_{g}(t)} &= \frac{1 - \alpha}{1 - \alpha + \alpha (\pth g_{1}(0) + (1 - \pth)g_{1}(1)) - g_{1}(1)}.
\end{align*}
Thus if $g_{1}(0)$ is sufficiently large, this limit is below $\pi_{1}$, implying that $t^{*} > 0$ and hence a non-zero asymptotic power by Theorem \ref{thm:FDR_power}. In fact, $g_{1}(0) = \infty$ for both examples in Section \ref{subapp:EM}. 

\subsection{Analysis of our method with informative pre-ordering}
When the pre-ordering is indeed informative in the sense that $f(t)$ is decreasing with $f(1) < (1 - \alpha)/ (1 - \mu)$ so that $t_{\mathrm{AT}}^{*} > 0$. We show that using the masking function $g(p_{i})$ can still improve the power. For illustration, we consider the score 
\begin{equation}
  \label{eq:AT_dominate}
  T_{i} = \min\left\{\frac{i}{n}, \frac{g(p_{i})}{b}\right\}
\end{equation}
for some $b \ge 0$. The plain accumulation test is a special case with $b = 0$. Intuitively, this method not only rejects the first $\lfloor n t^{*}\rfloor$ hypotheses as in accumulation tests but also the remaining ones with tiny p-values. 

Before analyzing this procedure rigorously in theory, we illustrate it using a simple simulation. Suppose the p-values are computed from one-sided z-tests, i.e.
\[p_{i} = 1 - \Phi(z_{i})\]
where $\Phi$ is the cumulative distribution function of a standard normal distribution, and $z_{i}$'s are independently generated from normal mixture models with unit variance, i.e.
\[z_{i}\sim \pi_{i1}N(0, 1) + (1 - \pi_{i1})N(3, 1).\]
We consider the case where 
\[\pi_{i1} = Ci^{-1/2},\]
where $C$ is a constant governing the proportion of non-nulls. In this case, the ordering is fully informative as $\pi_{i1}$ is strictly decreasing. We simulate the \textFDP ~and the power for both accumulation tests with $h(p) = 2I(p\ge 0.5)$ and our method with the same accumulation function and scores \eqref{eq:AT_dominate} with $b = 0.01$. The number of hypotheses is chosen as $20000$ and the constant $C$ is chosen from $\{0.5, 1, 2, 5\}$. For each setting the \textFDP ~and the power are recorded for $10000$ independent replicates. The box-plots are displayed in Figure \ref{fig:AT_dominate}. The advantage of using masking functions is clear: it reduces the variability of \textFDP ~while enhances the power significantly.

\begin{figure}[htp]
  \centering
  \includegraphics[width = \textwidth]{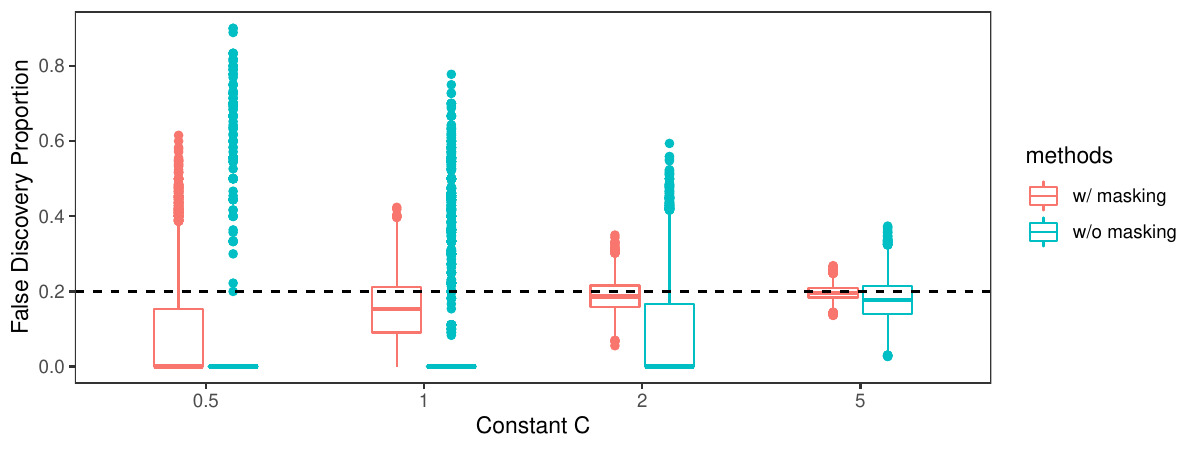}
  \includegraphics[width = \textwidth]{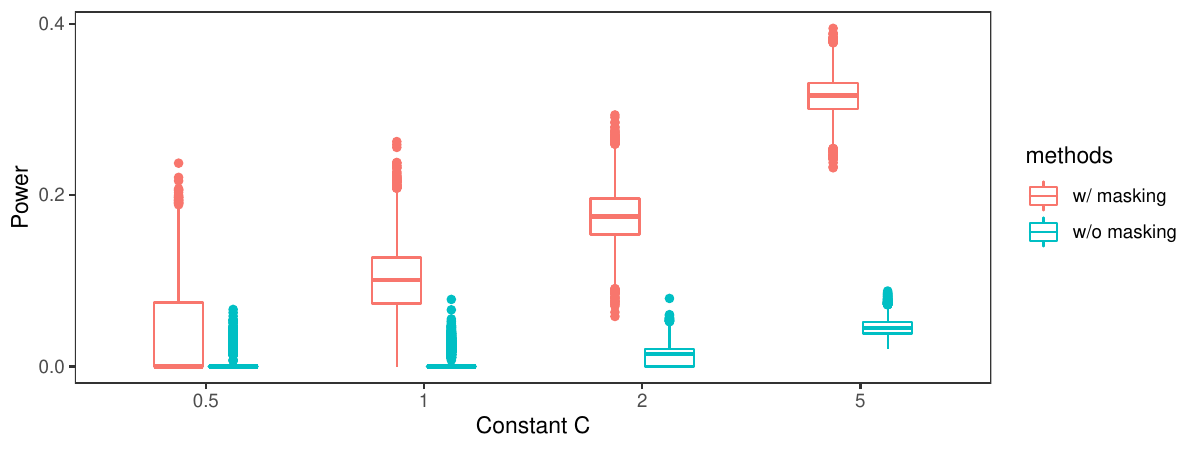}
  \caption{Simulation study comparing accumulation tests and our method using the score \eqref{eq:AT_dominate}.}\label{fig:AT_dominate}
\end{figure}

Now we show that for appropriately chosen positive $b$, the asymptotic power is higher. By Lemma \ref{lem:uniform}, 
\begin{align*}
  F_{1}(t) &= \lim_{n\rightarrow \infty}\frac{1}{n}\sum_{i=1}^{\lfloor nt\rfloor}I(h_{i} = 1) + \frac{1}{n}\sum_{i > \lfloor nt \rfloor}I(h_{i} = 1)\P_{1}(g(p_{i})\le bt)\\
& = tf(t) + (f(1) - tf(t))\tau_{g1}(bt),
\end{align*}
\begin{align*}
  F_{0}(t) &= \lim_{n\rightarrow \infty}\frac{1}{n}\sum_{i=1}^{\lfloor nt\rfloor}I(h_{i} = 0) + \frac{1}{n}\sum_{i > \lfloor nt \rfloor}I(h_{i} = 0)\P_{0}(g(p_{i})\le bt)\\
& = t(1 - f(t)) + (1 - f(1) - t(1 - f(t)))\tau_{g0}(bt),
\end{align*}
and 
\begin{align*}
  H(t)& = \lim_{n\rightarrow \infty}\frac{1}{n}\sum_{i=1}^{\lfloor nt\rfloor}I(h_{i} = 0)\E_{0}h(p_{i}) + \frac{1}{n}\sum_{i=1}^{\lfloor nt\rfloor}I(h_{i} = 1)\E_{1}h(p_{i})\\
&\quad  + \frac{1}{n}\sum_{i > \lfloor nt \rfloor}I(h_{i} = 0)\E_{0}h(p_{i})I(g(p_{i})\le bt)  + \frac{1}{n}\sum_{i > \lfloor nt \rfloor}I(h_{i} = 1)\E_{1}h(p_{i})I(g(p_{i})\le bt)\\
& = \lim_{n\rightarrow \infty}\frac{1}{n}\sum_{i=1}^{\lfloor nt\rfloor}I(h_{i} = 0) + \frac{1}{n}\sum_{i=1}^{\lfloor nt\rfloor}I(h_{i} = 1)\mu\\
&\quad  + \frac{1}{n}\sum_{i > \lfloor nt \rfloor}I(h_{i} = 0)\P_{0}(g(p_{i})\le bt)  + \frac{1}{n}\sum_{i > \lfloor nt \rfloor}I(h_{i} = 1)\E_{1}h(p_{i})I(g(p_{i})\le bt)\\
& = F_{0}(t) + tf(t)\mu + (f(1) - tf(t))\mu_{g}(bt),
\end{align*}
where the second equality uses the fact that $\E_{0} h(p_{i}) = 1, \E_{1} h(p_{i}) = \mu$ and 
\[\E_{0}h(p_{i})I(g(p_{i})\le bt) = \E_{0}[I(g(p_{i})\le bt)\E_{0}[h(p_{i})\mid g(p_{i})]] = \E_{0}[I(g(p_{i})\le bt)] \P_{0}(g(p_{i})\le bt).\]
By definition \eqref{eq:tstar},
\[t^{*} = \sup\left\{t: H(t) \le \alpha (F_{0}(t) + F_{1}(t))\right\}.\]
Note that for any given $t$, by definition this method rejects no less than the plain accumulation test with same $t$. For this reason, this method is more powerful asymptotically if $t^{*} > t_{\mathrm{AT}}^{*}$. If $\tau_{g1}, \tau_{g0}, \mu_{g}$ are all continuous, then it is left to show that
\[H(t_{\mathrm{AT}}^{*}) < \alpha (F_{0}(t_{\mathrm{AT}}^{*}) + F_{1}(t_{\mathrm{AT}}^{*})).\]
By some algebra and the fact that $f(t_{\mathrm{AT}}^{*}) = (1 - \alpha) / (1 - \mu)$, this is equivalent to
\[\frac{(1 - t_{\mathrm{AT}}^{*})\tau_{g0}(bt_{\mathrm{AT}}^{*}) + (f(1) - t_{\mathrm{AT}}^{*}f(t_{\mathrm{AT}}^{*}))(\mu_{g}(bt_{\mathrm{AT}}^{*}) - \tau_{g0}(bt_{\mathrm{AT}}^{*}))}{(1 - t_{\mathrm{AT}}^{*})\tau_{g0}(bt_{\mathrm{AT}}^{*}) + (f(1) - t_{\mathrm{AT}}^{*}f(t_{\mathrm{AT}}^{*}))(\tau_{g1}(bt_{\mathrm{AT}}^{*}) - \tau_{g0}(bt_{\mathrm{AT}}^{*}))} < \alpha.\]
For instance, when $h(p) = I(p > \pth) / (1 - \pth)$ and $g(p) = \min\{p, \pth (1 - p) / (1 - \pth)\}$ as in last subsection, by \eqref{eq:preorderinglimit},
\begin{align*}
  &\lim_{b\rightarrow 0}\frac{(1 - t_{\mathrm{AT}}^{*})\tau_{g0}(bt_{\mathrm{AT}}^{*}) + (f(1) - t_{\mathrm{AT}}^{*}f(t_{\mathrm{AT}}^{*}))(\mu_{g}(bt_{\mathrm{AT}}^{*}) - \tau_{g0}(bt_{\mathrm{AT}}^{*}))}{(1 - t_{\mathrm{AT}}^{*})\tau_{g0}(bt_{\mathrm{AT}}^{*}) + (f(1) - t_{\mathrm{AT}}^{*}f(t_{\mathrm{AT}}^{*}))(\tau_{g1}(bt_{\mathrm{AT}}^{*}) - \tau_{g0}(bt_{\mathrm{AT}}^{*}))} \\
& =  \frac{(1 - t_{\mathrm{AT}}^{*}) + (f(1) - t_{\mathrm{AT}}^{*}f(t_{\mathrm{AT}}^{*}))(g_{1}(1) - 1)}{(1 - t_{\mathrm{AT}}^{*}) + (f(1) - t_{\mathrm{AT}}^{*}f(t_{\mathrm{AT}}^{*}))(\pth g_{1}(0) + (1 - \pth)g_{1}(1) - 1)}.
\end{align*}
Thus the limit is strictly below $\alpha$ if $g_{1}(0)$ is sufficiently large. As commented at the end of last subsection, $g_{1}(0) = \infty$ in many applications. 

Therefore, even with an informative pre-ordering, using the masking function may further improve the power. This inspires an interesting question on how to combine the pre-ordering and the masked p-values in an optimal way to enhance power. However, this is beyond the main focus of this paper and we leave it to future research.

\subsection{Proofs}\label{subapp:power_proof}
\begin{proof}[of Lemma \ref{lem:thattstar}]
Let $H_{\alpha n}(t) = H_{n}(t) - \alpha F_{n}(t)$ and 
\[\delta_{n} = \sup_{t\in [0, 1]}|H_{\alpha n}(t) - H_{\alpha}(t)|.\]
By definition, 
\[\hat{t}_{n} = \sup\{t: H_{\alpha n}(t)\le 0\}, \quad t^{*} = \sup\{t: H_{\alpha}(t)\le 0\}.\]
Then \eqref{eq:uniform_convergence} implies that $\delta_{n}\stackrel{p}{\rightarrow}0$. We prove for each case separately.
  \begin{enumerate}
  \item For any $\eps > 0$, 
    \[\inf_{t > \eps}H_{\alpha n}(t) \ge\inf_{t > \eps} H_{\alpha}(t) - \delta_{n}.\]
By \eqref{eq:tstar0}, 
\[\P\lb \inf_{t > \eps}H_{\alpha n}(t) > 0\rb\rightarrow 1.\]
As a result,
\[\P(\hat{t}_{n}\le \eps)\rightarrow 1.\]
Since this holds for any $\eps$, we conclude that $\hat{t}_{n}\stackrel{p}{\rightarrow} 0$. 
  \item By \eqref{eq:uniform_convergence}, for any $m$
\[H_{\alpha n}(t_{m}) \le H_{\alpha}(t_{m}) + \delta_{n}.\]
By \eqref{eq:tstar1}, 
\[\P(H_{\alpha n}(t_{m})\le 0)\rightarrow 1.\]
This entails that
\[\P(\hat{t}_{n}\ge t_{m})\rightarrow 1.\]
Since this holds for all $m$ and $t_{m}\uparrow 1$, we arrive at $\hat{t}_{n}\stackrel{p}{\rightarrow} 1$.
\item Let $t_{m}'$ be any sequence that $t_{m}'\downarrow t^{*}$. For any $m$, 
\[H_{\alpha n}(t_{m})\le H_{\alpha}(t_{m}) + \delta_{n}\]
and 
\[\inf_{t > t_{m}'}H_{\alpha n}(t) > \inf_{t > t_{m}'}H_{\alpha}(t) - \delta_{n}.\]
By \eqref{eq:tstar0} and \eqref{eq:tstar1}, we have
\[\P\lb H_{\alpha n}(t_{m}) < 0, \,\, \inf_{t > t_{m}'}H_{\alpha n}(t) > 0\rb\rightarrow 1.\]
This implies that
\[\P(\hat{t}_{n}\in [t_{m}, t_{m}'])\rightarrow 1.\]
Since $t_{m}\uparrow t^{*}$ and $t_{m}'\downarrow t^{*}$, we conclude that $\hat{t}_{n}\stackrel{p}{\rightarrow}t^{*}$.
  \end{enumerate}
\end{proof}

\begin{proof}[of Theorem \ref{thm:FDR_power}]
Note that $F_{n0}(t), F_{n1}(t), F_{n}(t), F_{0}(t), F_{1}(t), F(t)$ are all non-decreasing functions. Let $t_{m}$ and $t_{m}'$ be any two sequences such that $t_{m}\uparrow t^{*}, t_{m}'\downarrow t^{*}$. Let $\cE_{n, m}$ denotes the event that $\hat{t}_{n}\in [t_{m}, t_{m}']$. Then Lemma \ref{lem:thattstar} implies that $\P(\cE_{n, m})\rightarrow 1$ for each $m$. 

First we prove that $\FDP_{n}\stackrel{p}{\rightarrow} F_{0}(t^{*}) / F(t^{*})$ when $t^{*} > 0$. Without loss of generality we assume that $t_{m} > 0$. On event $\cE_{n, m}$,
\[\FDP_{n}\in \left[\frac{F_{n0}(t_{m})}{F_{n}(t_{m}')}, \min\left\{1, \frac{F_{n0}(t_{m}')}{F_{n}(t_{m})}\right\}\right], \]
By \eqref{eq:Fn01}, Slusky's theorem and the fact that $F(t_{m}) > 0$, 
\[\frac{F_{n0}(t_{m})}{F_{n}(t_{m}')}\stackrel{p}{\rightarrow} \frac{F_{0}(t_{m})}{F(t_{m}')}, \quad \frac{F_{n0}(t_{m}')}{F_{n}(t_{m})}\stackrel{p}{\rightarrow} \frac{F_{0}(t_{m}')}{F(t_{m})}.\]
Thus, for any $\eps > 0$, 
\[\P\lb \FDP_{n}I_{\cE_{n, m}}\in \left[\frac{F_{0}(t_{m})}{F(t_{m}')} - \eps, \frac{F_{0}(t_{m}')}{F(t_{m})} + \eps\right]\rb\rightarrow 1.\]
Since $\P(\cE_{n, m})\rightarrow 1$, 
\[\P\lb \FDP_{n}\in \left[\frac{F_{0}(t_{m})}{F(t_{m}')} - \eps, \frac{F_{0}(t_{m}')}{F(t_{m})} + \eps\right]\rb\rightarrow 1.\]
Since $t_{m}\rightarrow t^{*}, t_{m}' \rightarrow t^{*}$ and $F_{0}, F$ are both continuous at $t^{*}$, 
\[\P\lb \bigg|\FDP_{n} - \frac{F_{0}(t^{*})}{F(t^{*})}\bigg|\le \eps\rb\rightarrow 0.\]
This holds for arbitrary $\eps > 0$. Thus,
\begin{equation}
  \label{eq:FDPn}
  \FDP_{n}\stackrel{p}{\rightarrow} \frac{F_{0}(t^{*})}{F(t^{*})}.
\end{equation}
Since $\FDP_{n}\in [0, 1]$, \eqref{eq:FDPn} implies the convergence in $L_{1}$, i.e.
\[\FDR_{n} = \E [\FDP_{n}]\rightarrow \frac{F_{0}(t^{*})}{F(t^{*})}.\]

For the asymptotic power, the monotonicity of $F_{n1}$ implies that 
\[\TPR_{n} \in \left[\frac{F_{n1}(t_{m})}{F_{n1}(1)}, \frac{F_{n1}(t_{m}')}{F_{n1}(1)}\right].\]
Using the same argument as above, we have
\[\TPR_{n}\stackrel{p}{\rightarrow} \frac{F_{1}(t^{*})}{F_{1}(1)}.\]
Since $\TPR_{n}\in [0, 1]$, this implies 
\[\Pow_{n}\rightarrow \frac{F_{1}(t^{*})}{F_{1}(1)}.\]
\end{proof}

\begin{proof}[of Lemma \ref{lem:uniform}]
Since $p_{i}$'s are independent, $T_{i}$'s are independent. Note that $F_{n0}(t), F_{n1}(t), F_{n}(t), H_{n}(t)$ can be in the form of 
\[\frac{1}{n}\sum_{i=1}^{n}m_{i}(p_{i})I(T_{i}\le t).\]
for some deterministic bounded functions $m_{1}, \ldots, m_{n}$. Let $B$ denote the bound of $m_{i}$'s. Then $B = 1$ for $F_{n0}(t), F_{n1}(t), F_{n}(t)$ and $B = h(1)$ for $H_{n}(t)$. Let $f_{i}(p_{i}; t) = m_{i}(p_{i})I(T_{i}\le t)$. Then $F_{i}\equiv B$ is an upper envelop of $f_{i}$. Also, for any given $(p_{1}, \ldots, p_{n})$, by Sauer's lemma,
\[\#\{(f_{1}(p_{1}; t), \ldots, f_{n}(p_{n}; t)): t\in [0, 1]\} = \#\{(I(T_{1}\le t), \ldots, I(T_{n}\le t)): t\in [0, 1]\}\le n + 1.\]
This implies that
\[\log \#\{(f_{1}(p_{1}; t), \ldots, f_{n}(p_{n}; t)): t\in [0, 1]\} = O(\log n) = o(n).\]
By Theorem 8.2 of \cite{pollard1990empirical}, 
\[\sup_{t\in [0, 1]}\bigg|\frac{1}{n}\sum_{i=1}^{n}(f_{i}(p_{i}, t) - \E f_{i}(p_{i}, t))\bigg| = o_{p}(1).\]
It is easy to compute $\E F_{n0}(t), \E F_{n1}(t)$ and $\E F_{n}(t)$. For $\E H_{n}(t)$, recalling that $\E_{0} [h(p_{i})\mid g(p_{i})] \le 1$ almost surely and $T_{i}$ depends on $p_{i}$ through $g(p_{i})$,
\begin{align*}
  \E H_{n}(t) &= \frac{1}{n}\sum_{i=1}^{n}\E [h(p_{i})I(T_{i}\le t)]\\ 
& = \frac{1}{n}\sum_{i=1}^{n}I(h_{i} = 0)\E_{0} [h(p_{i})I(T_{i}\le t)] + \frac{1}{n}\sum_{i=1}^{n}I(h_{i} = 1)\E_{1} [h(p_{i})I(T_{i}\le t)]\\
& = \frac{1}{n}\sum_{i=1}^{n}I(h_{i} = 0)\P_{0}(T_{i}\le t) + \frac{1}{n}\sum_{i=1}^{n}I(h_{i} = 1)\E_{1} [h(p_{i})I(T_{i}\le t)]
\end{align*}
\end{proof}

\end{document}